\documentclass[a4paper,11pt]{article}
\usepackage[utf8]{inputenc}
\usepackage{epigraph}
\usepackage{appendix}
\usepackage{amsmath}
\usepackage{amssymb}
\usepackage{amsthm}
\usepackage{amsfonts}
\usepackage{amsbsy}
\usepackage{esint}
\usepackage{mathtools}
\usepackage{bbm}
\usepackage{array}
\usepackage{tikz}
\usepackage{url}
\usepackage{graphicx}
\usepackage{physics}
\usepackage{dsfont}
\usepackage{xcolor}
\usepackage{subcaption}
\usepackage{empheq}
\usepackage[colorlinks=true, allcolors=blue]{hyperref}

\newcommand{\Schatten}{{\mathfrak S}}
\newcommand{\per}{\mathrm{per}}
\newcommand{\Ran}{\mathrm{Ran}}
\newcommand{\Trper}{\underline{\Tr}}
\newcommand{\BZ}{\mathcal{B}}
\theoremstyle{plain}
\newtheorem{theorem}{Theorem}[section]

\newtheorem{assumption}[theorem]{Assumption}
\newtheorem{lemma}[theorem]{Lemma}
\newtheorem{corollary}[theorem]{Corollary}
\newtheorem{proposition}[theorem]{Proposition}

\theoremstyle{remark}
\newtheorem{remark}[theorem]{Remark}

\renewcommand{\vec}[1]{#1}
\usepackage{tikz-cd}

\makeatletter
\tikzset{
  column sep/.code=\def\pgfmatrixcolumnsep{\pgf@matrix@xscale*(#1)},
  row sep/.code   =\def\pgfmatrixrowsep{\pgf@matrix@yscale*(#1)},
  matrix xscale/.code=%
    \pgfmathsetmacro\pgf@matrix@xscale{\pgf@matrix@xscale*(#1)},
  matrix yscale/.code=%
    \pgfmathsetmacro\pgf@matrix@yscale{\pgf@matrix@yscale*(#1)},
  matrix scale/.style={/tikz/matrix xscale={#1},/tikz/matrix yscale={#1}}}
\def\pgf@matrix@xscale{1}
\def\pgf@matrix@yscale{1}
\makeatother

\usepackage[margin=2.0cm]{geometry}

\def\R{{\mathbb R}}
\def\N{{\mathbb N}}
\def\C{{\mathbb C}}
\def\Z{{\mathbb Z}}
\def\eps{{\varepsilon}}

\def\cB{{\mathcal B}}

\def\cH{{\mathcal H}}
\def\cL{{\mathcal L}}

\def\cR{{\mathcal R}}

\def\1{{\mathds 1}}

\def\dps{\displaystyle }

\title{Coherent electronic transport in periodic crystals}
\author{Eric Canc\`es$\dagger$, Clotilde Fermanian Kammerer*, Antoine Levitt$\dagger$, Sami Siraj-Dine$\dagger$*}

\begin{document}
\maketitle
\abstract{We consider independent electrons in a periodic crystal in
  their ground state, and turn on a uniform electric field at some
  prescribed time. We rigorously define the current per unit volume
  and study its properties using both linear response and adiabatic
  theory. Our results provide a unified framework for various
  phenomena such as the quantization of Hall conductivity of
  insulators with broken time-reversibility, the ballistic regime of
  electrons in metals, Bloch oscillations in the long-time response of
  metals, and the static conductivity of graphene. We identify
  explicitly the regime in which each holds. }

\tableofcontents
\section{Introduction} 

We consider a $d$-dimensional perfect crystal ($d=1$, $2$ or $3$) with periodic lattice $\cR$ and assume that its electronic structure can be described by an effective linear Hamiltonian~$H$ acting on some Hilbert space $\cH$. We focus here on the case of spinless continuous models, for which $\cH=L^2(\R^d;\C)$ and 
\begin{equation}\label{eq:periodic_Hamiltonian}
H=\dps \frac 1 2 (-i\nabla + {\mathcal A})^{2} + V,
\end{equation}
where ${\mathcal A} \in L^4_{\rm per}(\R^d;\R^d)$ and
$V \in L^2_{\rm per}(\R^d;\R)$ are $\cR$-periodic functions. We adopt
the Coulomb gauge choice, $\nabla \cdot \mathcal{A} = 0 $ in the sense
of distributions. At zero temperature, the ground-state density matrix
is given by
\begin{equation}\label{eq:GSDM}
\gamma(0) = \mathds 1(H \leq \mu_{\rm F}), 
\end{equation}
where $\mu_{\rm F} \in \R$ is the Fermi level, chosen to have a
prescribed number of electrons per unit cell.

\medskip

Depending on the position of $\mu_{\rm F}$ in the spectrum $\sigma(H)$
of $H$, this can model different types of physical systems. If
$\mu_{\rm F} \notin \sigma(H)$, the system is an insulator. If
$\mu_{\rm F}$ is an interior point of $\sigma(H)$, the system is a metal, or a semi-metal,
depending on the density of states of $H$ at $\mu_{\rm F}$. We refer
to Section~\ref{sec:insmetsem} for the precise hypotheses we use in
each case.

\medskip

The vector potential $\mathcal A$ is chosen to be periodic, which
excludes the case of a uniform external magnetic field. Our analysis
therefore does not directly cover the quantum Hall effect, but can be
adapted to do so (see Remark~\ref{rem:only_depends_on_Hk}). It is
directly applicable to the quantum anomalous Hall effect
\cite{haldane1988model}. We perform our analysis with this particular
Hamiltonian, but it can easily be extended to spin-dependent
continuous models, tight-binding models, or 2D materials such as
graphene (for which the physical space is three-dimensional while the
periodic lattice is two-dimensional); see
Remark~\ref{rem:only_depends_on_Hk} for the exact structure needed.

\medskip

The purpose of this article is to analyze mathematically the behavior
of the electrical current appearing in the crystal when a uniform
external electric field is turned on instantaneously at the initial
time $t=0$. In the case of a uniform stationary electric field of
magnitude $\eps > 0$ along a (not necessarily normalized) vector
$e_{\beta} \in \R^{d}$, the Hamiltonian of the system at time $t > 0$ is
\begin{align} \label{def:Hbetaeps}
    H^{\eps}_\beta = H + \eps x_{\beta},
  \end{align}
where $x_{\beta}=x\cdot e_{\beta}$. This operator is self-adjoint
on $L^2(\R^d;\C)$ (see Proposition~\ref{prop:setting} below),
and therefore gives rise to a unitary group $({\rm e}^{-itH^{\eps}_\beta})_{t \in \R}$ on $L^2(\R^d;\C)$. The electronic state of the system at time $t \ge 0$ then is
\begin{equation}\label{def:gammaeps}
\gamma_\beta^\eps(t) = {\rm e}^{-itH^{\eps}_\beta} \gamma(0) {\rm e}^{itH^{\eps}_\beta}.
\end{equation}
The electrical current in the $e_\alpha$-direction at time $t \ge 0$ is defined as
\begin{equation} \label{eq:current}
j^\eps_{\alpha,\beta}(t) = \underline {\rm Tr}(J_\alpha\gamma_\beta^\eps(t)),
\end{equation}
where $\underline {\rm Tr}$ is the trace per unit volume (which will
be precisely defined in Section \ref{sec:notation}) and $J_\alpha$ the
current operator along the vector $e_{\alpha} \in \R^{d}$ (not necessarily normalized nor orthogonal to $e_\beta$), defined as
\begin{align}\label{def:Jalpha}
  J_\alpha =  -\left( -i \nabla + \mathcal A\right) \cdot e_{\alpha}.
\end{align}

\begin{remark}[on units and sign convention] If a spinless particle
  with mass $m$ and charge $q$ is subjected to a electromagnetic field
  generated by a vector potential $\mathcal A$ and a scalar potential
  $-\eps x_\beta$ generated by a uniform electric field
  $\eps e_\beta$, its Hamiltonian in atomic units is
  $H=\frac 1{2m} (-i\nabla - q{\mathcal A})^2 -q \eps x_\beta$ and the
  charge current operator is $J=q(-i\nabla - q {\mathcal A})$. In our
  definitions ~\eqref{eq:periodic_Hamiltonian}, \eqref{def:Hbetaeps} and
  \eqref{def:Jalpha}, we have set $m=1$ and $q=-1$ (atomic units)
  which are the physical values for the electron: this corresponds to
  applying a force in the direction $-e_{\beta}$ to the electrons, and
  measuring their velocity in the direction $-e_{\alpha}$.
\end{remark}

In the limit of weak external fields ($\eps \ll 1$), the qualitative
properties of the function $t \mapsto j^\eps_{\alpha,\beta}(t)$
heavily depends on the physical nature of the material (insulator,
metal, semi-metal), as well as on the regime (short,
intermediate or long times). Our main results, stated in
Theorems~\ref{thm:insulator}, \ref{thm:metal}
and~\ref{theo:semi-metal}, show that the behavior is as follows (see
Figure \ref{fig:num_results} in Section~\ref{sec:numerics})
\begin{itemize}
\item For insulators, the time-averaged conductivity
  \begin{align}\label{def:conductivity}
    \sigma_{\alpha,\beta} = \lim_{t \to \infty} \frac 1 t \int_{0}^{t} \lim_{\varepsilon \to 0} \frac {j_{\alpha,\beta}^{\varepsilon}(t')}{\varepsilon} \, dt'
  \end{align}
  has a finite value, which is zero in longitudinal directions, and, for 2D materials, is
  proportional to the Chern number in the transverse direction (quantum anomalous Hall effect).
\item For metals, when $t \ll \varepsilon^{-1}$, the electrons are in
  the ballistic regime, and the current increases linearly:
  $j_{\alpha,\beta}^{\varepsilon}(t) \approx D_{\alpha,\beta}
  \varepsilon t$. Under some additional assumptions on the Bloch
  bands, the current displays Bloch oscillations of order $1$ when
  $\varepsilon^{-1} \ll t \ll {\varepsilon^{-1}{\rm log}
    (\eps^{-\zeta})}$ for some small enough $\zeta > 0$.
\item For time-reversible 2D semimetals such as graphene, the
  time-averaged conductivity $\sigma_{\alpha,\beta}$ defined
  in~\eqref{def:conductivity} has a finite value equal to $\frac 1
  {16} e_{\alpha} \cdot e_{\beta}$ times the number of Dirac points in
  the Brillouin zone.
\end{itemize}
Although our formalism is different, our results for insulators and
metals are formally consistent with those obtained using the semiclassical
equations of motion
$\dot x = \nabla \lambda_{n,k}, \;\dot k = -\nabla V + \dot x \times
(\nabla \times {\cal A})$ and their higher-order refinements in the case when the $n^{\rm th}$ band is
isolated, where the $\lambda_{n,k}$'s are the Bloch eigenvalues of $H$
(see Section~\ref{sec:notation}). We refer to
\cite{panati2003effective, stiepan2013semiclassical} for a mathematical analysis of the insulating case.

\medskip

Note that our results use an averaging in time, and we are unable to
conclude anything about what would be the naive definition of the
conductivity
\begin{align}
  \label{eq:conductivity_naive}
  \lim_{t \to  \infty} \lim_{\varepsilon \to 0} \frac{j_{\alpha,\beta}^{\varepsilon}(t)}{\varepsilon}.
\end{align}
A form of averaging of time fluctuations is always necessary to infer
zero-frequency behavior from step responses in non-dissipative
systems, even in the linear case. The easiest way to see this is by
the very simple model for the relationship between an input $I(t)$ and
an output $O(t)$:
  \begin{equation} \label{eq:toymodel}
    i \dot O(t) = \omega O(t) + I(t).
  \end{equation}
  This simplified model describes a forced oscillator with eigenfrequency $\omega$, and arises from the linear response of the time-dependent Schr\"odinger equation of a two-level system. For a
  constant input~$I_0$, there is a steady state solution
  $O_0 = \widehat R_0 I_0$, where $\widehat R_0 = -\frac 1 \omega$ is
  the zero-frequency transfer function of the system. However, since
  this system is oscillatory, this steady state may never be reached:
  if $I$ is brutally switched on at time $0$ with $I(t) = O(t) = 0,$
  if $t \leq 0$, $I(t) = I_0$ if $t > 0$, then
  $ O(t) = O_0(1- {\rm e}^{-i\omega t}) = \widehat{R}_0 I_0 (1- {\rm
    e}^{-i\omega t})$ and we cannot define $\widehat R_0$ as the limit
  of $O(t)/I_0$ when $t$ goes to infinity. However, by averaging we
  obtain
  \begin{align*}
     \lim_{t\to\infty}\frac 1 t  \int_{0}^{t} \frac{O(t')}{I_0} \, dt' = \widehat R_0.
   \end{align*}
   Another common way of retrieving the value of $\widehat R_0$ is by an
   adiabatic switching of the electric field
   $I(t) = I_0 {\rm e}^{\eta t}$ for $t \leq 0$, $I(t) = I_0$ for
   $t > 0$ \cite{bouclet2005linear,DeNittis2017}. Another possibility
   is to represent the relationship between $O$ and $I$ by a
   convolution with a causal response function $R(t)$:
   $O(t) = (R \ast I)(t)$, and define the zero-frequency transfer
   function as $\lim_{\eta \to 0^{+}} \widehat R(i \eta)$, as is often
   done implicitly in the physics literature. Yet another, more
   physical, possibility is to use a model with dissipation (in this
   case
   $i \dot O_\eta(t) + i \eta O_\eta(t) = \omega O_\eta(t) + I(t)$),
   compute the zero-frequency transfer function as the long-time limit
   of $O_\eta(t)/I_0$, and then let the dissipation $\eta$ tend to
   zero. A particular variant of this scheme is known as the
   relaxation time approximation \cite{bellissard1994noncommutative}
   (the relaxation time being proportional to $1/\eta$). For simple
   systems, all these methods are
   equivalent.
   
   Note that the problems in the toy model \eqref{eq:toymodel} are related to the
   presence of a \textit{resonance} at $\omega$ in the model, i.e. a
   pole in the Fourier transform of the response function. For our
   perfect crystal model however, the oscillatory components of the response
   are integrated over the Brillouin zone of the periodic crystal,
   which induces an averaging. Therefore, these procedures might not
   be necessary. Indeed, we observe numerically in simple tight-binding models that
   the naive limit in~\eqref{eq:conductivity_naive} seems to be
   well-defined (see Section \ref{sec:numerics}). Identifying precise
   conditions on the band structure so that this holds will be the
   subject of future work.
   
   \medskip 

In the metallic case, the conductivity is either infinity or zero,
depending on the definition adopted. Indeed, our results imply that
\begin{align*}
  \lim_{t \to \infty} \frac 1 t \int_{0}^{t}\lim_{\varepsilon \to 0} \frac {j_{\alpha,\alpha}^{\varepsilon}(t')}{\varepsilon} \, dt' = +\infty
\end{align*}
is infinite, because
$j_{\alpha,\alpha}^{\varepsilon}(t) \approx D_{\alpha\alpha} \varepsilon t$ in
the regime $t \ll \varepsilon^{-1}$. On the other hand, in
tight-binding models, a simple argument \cite[Proposition
4]{bellissard1994noncommutative} shows that
\begin{align*}
  \lim_{\varepsilon \to 0}\frac 1 \varepsilon \lim_{t \to \infty}\frac 1 t \int_{0}^{t}{j_{\alpha,\alpha}^{\varepsilon}(t')} \, dt' = 0.
\end{align*}
These two limits correspond to different regimes. In the first one, the
electrons undergo ballistic transport, being uniformly accelerated by
the electric field. In the second one, the electrons undergo Bloch
oscillations, a phenomenon whereby particles in a periodic potential
accelerated by a constant force oscillate rather than propagate, as
first noticed by Zener \cite{zener1934theory}.

\medskip 

Of course, our model is extremely simple. We assume that the electrons
are at zero temperature and we ignore electron-electron interactions,
the reaction of the lattice (electron-phonon interactions), and
electron scattering by impurities in the crystal. These collision
events play a relatively minor role in insulators at low temperatures, with the quantum
Hall effect in particular being very robust to perturbations
\cite{bellissard1994noncommutative}. However, they impact significantly the properties of
metals. In fact, in the linear response regime ($\varepsilon \ll 1$,
$t \ll \varepsilon^{-1}$), the current increases linearly,
representing ballistic transport (see Theorem \ref{thm:metal}). This
increase in the velocity of the electrons physically results in an increased
collision rate, which acts as dissipation and eventually limits the
current. This results in the finite conductivity observed
experimentally in macroscopic physics (Ohm's law). The mathematical understanding
of this effect in the mathematical framework considered here is left
to future work.

\medskip

The question of quantum transport in solids has attracted significant
attention in the mathematical community, with one of the main drivers
being the explanation of Anderson localization on the one hand, and
the quantum Hall effect and its relation to topological properties on
the other hand
\cite{Teufel2019,Monaco2017,FiorenzaDomenicoPanati,Panati2007,Marcelli2019}.
Other topics of interest include the properties of graphene (see for
instance \cite{giuliani2012universality}), and mesoscopic transport in
the Landauer-Buttiker formalism. Comparatively few works have looked
specifically at transport in metals. To the best of our knowledge, the
present work is the first to present mathematically rigorous results
on insulators, metals  and semi-metals in a unified framework.

\medskip

In our results in the linear response regime, we consider the quantity
\begin{align*}
  j_{\alpha,\beta}^{\rm LR}(t) = \lim_{\varepsilon \to 0} \frac {j_{\alpha,\beta}^{\varepsilon}(t)}{\varepsilon}.
\end{align*}
and then compute the conductivity in the limit $t \to \infty$. This
order of limits is necessary to ensure that the electronic state never
deviates significantly from its equilibrium. Considering the opposite
limit, i.e. the infinite-time dynamics of
$j_{\alpha,\beta}^{\varepsilon}(t)$ at $\varepsilon$ fixed, is an extremely
hard problem, as it is for any dynamics of non-dissipative systems.

\medskip 

The validity of this linear response approximation to derive
conductivities, pioneered by Kubo in \cite{Kubo1957}, is by no means
obvious. As was pointed out in \cite{vankampen1971case}, the most
obvious way to derive it, ``microscopic linearity'' is not physically
relevant: assuming a mean free path length $d$ for the electrons, the
validity condition that a free electron of (effective) mass $m$ and
charge $e$ is not significantly affected by an electric field $E$ in the time interval $[0,t]$ is $\frac{t^{2}}{2} \frac{e E}{m} \ll d$. Taking
macroscopic times $t$ and reasonable microscopic values for $e$ and
$m$, this limits fields to microscopic values
(\cite{vankampen1971case} quotes $10^{-18}$ Volt/cm), which is
unrealistic in practice.
The solution of this paradox is ``macroscopic linearity'':
$\gamma_\beta^\eps(t)$ has to be understood not as the state of a
single set of electrons, but rather as a thermodynamic ensemble. The
effective evolution of $\gamma_{\beta}^{\varepsilon}(t)$ then involves
a dissipative term (coming from electron-phonon, electron-impurity or
electron-electron interaction) that tends to restore the density
matrix to its equilibrium state. The condition of validity of linear
response is then that the driving force is negligible compared to the
restoring force, which is usually satisfied in practice
\cite{van1988van}. The mathematical justification of linear response
is then to consider a more sophisticated model involving a dissipation
strength $\eta$, and to perform the van Hove limit
$t \to \infty, \eta \to 0, t^{2} \eta = {\rm cst}$, at
$\varepsilon$ fixed \cite{van1988van}. Then $\varepsilon$ can be taken
to zero, and the results of linear response are recovered.
Alternatively, a Drude-type model like the relaxation-time
approximation \cite{bellissard2002coherent} can be used, with the
similar effect of returning the density matrix to equilibrium. Yet
another possibility is to never let the density matrix get out of its
equilibrium state by switching on adiabatically the external field
from negative infinity as $e^{\eta t} \varepsilon x_{\beta}$, and
to consider the limit $\eta \to 0$ first then $\varepsilon \to 0$
\cite{Bachmann2018}. In our naive model, we do not consider a
dissipation term, and therefore simply assume the validity of linear
response.

\medskip

Our method of proof is based on the standard gauge change
$\widetilde \psi(x,t) = {\rm e}^{i \varepsilon t x_{\beta}} \psi(x,t)$
that transforms the constant in time but non-spatially-periodic
Hamiltonian $H^{\eps}_\beta = H + \eps x_{\beta}$ into the
time-dependent Hamiltonian
$\widetilde{H}^\eps_\beta(t)= \frac 12 (-i\nabla + \mathcal{A} -
\eps\vec{e}_\beta t)^2 +V$. This Hamiltonian is spatially periodic,
and the study of its dynamics can be reduced via Bloch-Floquet theory
to that of its fibers
$\widetilde{H}^\eps_{\beta,k}(t)= \frac 12 (-i\nabla + k + \mathcal{A}
- \eps\vec{e}_\beta t)^2 +V$ acting on periodic functions (Section
\ref{sec:preliminaries}), for all values of the pseudo-momentum $k \in \R^{d}$. Fiber by fiber, this
time-dependent Hamiltonian can then be treated using the tools of
time-dependent perturbation theory (Section \ref{sec:adiabatic}).
Since time is scaled by $\varepsilon$, the Hamiltonian can be seen as
either a small perturbation of the rest Hamiltonian $H$ for small
times (in which case we can use linear response to expand
$j_{\alpha,\beta}^{\varepsilon}(t)$ to first order in $\varepsilon$
for a fixed $t$, Proposition \ref{prop:Dyson}), or as a slow
perturbation (in which case the adiabatic theorem allows us to access
larger time scales $t \approx \frac 1 \varepsilon$, Proposition
\ref{prop:adiab2}). For insulators and metals in the short-time
regime, both tools are applicable and yield the same result. For
metals in the Bloch oscillations regime, only the adiabatic theorem is
applicable, and for semimetals, only linear response is applicable due
to the gap closing at the Dirac points.

\medskip 

The techniques we use (linear response and adiabatic theory) are not
new, nor are our results particularly surprising to experts in the
field. Rather, we see the contribution of this paper as unifying in
the same framework disparate studies on different systems, as well as
providing insights on the current response without any specific
regularization technique (such as adiabatic switching or dissipation).
Our results on Bloch oscillations also appear to be new in the
mathematical literature.

\medskip

The structure of the paper is as follows. We describe our results in Section \ref{sec:results}: we define the
current in Proposition~\ref{prop:setting}, and study its properties
for insulators, metals and semi-metals in
Theorems~\ref{thm:insulator}, \ref{thm:metal}
and~\ref{theo:semi-metal}. We illustrate numerically the different
behaviors we obtain in each of the three settings in Section
\ref{sec:numerics}. We devote Section~\ref{sec:preliminaries} to
preliminaries about the regularity and Bloch decomposition of the
current. Section \ref{sec:adiabatic} states and proves results in
adiabatic and linear response perturbation theory.
Sections~\ref{sec:ins}, \ref{sec:met} and~\ref{sec:semimet} are
devoted to the proof of our results in the case of insulators, metals
and semi-metals. Finally two short Appendices are devoted to technical
issues.


\section{Main results: electrical current in periodic materials}\label{sec:results}

\subsection{Notation}
\label{sec:notation}
In
this paper we fix
${\mathcal A} \in L^4_{\rm per}(\R^d;\R^d)$,
$V \in L^2_{\rm per}(\R^d;\R)$ (see below for the definition of these
spaces), $\mu_{\rm F} \in \R$, and 
$\cR$ is the lattice of the $d$-dimensional crystal. We fix a (non-necessarily orthonormal) basis
$(\vec{e_{\alpha}})_{\alpha=1,\dots,d}$ of the momentum space $\R^d$, and set
$x_{\alpha} = \vec x \cdot \vec e_{\alpha}, \mathcal A_{\alpha} =
\mathcal A \cdot \vec e_{\alpha}$ for $\alpha = 1,\dots,d$.
We denote by $\cR^\ast$ the {\it dual lattice} of the periodic lattice
$\cR$, by $\Omega$ an arbitrary unit cell in the physical space, and
by $\cB$ an arbitrary unit cell in the reciprocal space (which we will
call by abuse of language the {\it Brillouin zone}). In the special case of
a cubic crystal of lattice parameter $a > 0$, we have $\cR=a\Z^d$,
$\cR^\ast=\frac{2\pi}a\Z^d$, and we can take $\Omega=[0,a)^d$,
$\cB=[-\frac \pi a, \frac \pi a)^d$. 

\medskip 

The {\it $\cR$-periodic Lebesgue and
Sobolev spaces} are denoted by
\begin{align*}
L^p_{\rm per}&:=\left\{ u \in L^p_{\rm loc}(\R^d;\C) \; | \; u \mbox{ $\cR$-periodic} \right\}, \\
H^s_{\rm per}&:=\left\{ u \in H^s_{\rm loc}(\R^d;\C) \; | \; u \mbox{ $\cR$-periodic} \right\}.
\end{align*}
The space of bounded linear operators on a Hilbert space~$\mathcal{H}$ is denoted by ${\mathcal L}(\mathcal{H})$, and the {Schatten class} of bounded operators~$A \in {\mathcal L}(\mathcal{H})$ such that $\Tr(|A|^p) < \infty$ by $\Schatten_{p}(\mathcal{H})$.
For $R \in \cR$, we denote by~$\tau_R$ the {translation operator} formally defined by $\tau_R \phi = \phi(\cdot - R)$. Depending on the context, $\tau_R$ will be seen as a unitary operator on $L^2(\R^d;\C)$, or as a linear operator on some $\cR$-translation invariant subspace of ${\mathcal D}'(\R^d;\C)$.
A bounded operator on $L^2(\R^d;\C)$ is called $\cR$-periodic if it commutes with~$\tau_R$ for all $R \in \cR$. An unbounded self-adjoint operator on $L^2(\R^d;\C)$ is called $\cR$-periodic if its resolvent is $\cR$-periodic. A bounded $\cR$-periodic operator $A \in {\mathcal L}(L^2(\R^d;\C))$ is called locally trace-class if $\chi A \chi \in \Schatten_{1}(L^2(\R^d;\C))$ for any compactly supported function $\chi \in L^\infty(\R^d;\C)$. For $p \ge 1$, we denote by $\Schatten_{p,\per}$ the space of $\cR$-periodic operators $A \in {\mathcal L}(L^2(\R^d;\C))$ such that $|A|^p$ is locally trace class. Any operator $A \in \Schatten_{1,\per}$ has a density $\rho_A \in L^1_{\rm per}$ characterized by 
$$
\forall \chi \in C^\infty_{\rm c}(\R^d;\C), \quad \Tr(A\chi) = \int_{\R^d} \rho_A \chi.
$$
The trace per unit volume of an operator $A \in \Schatten_{1,\per}$ is defined as
$$
\Trper(A) = \frac{1}{|\Omega|} \Tr_{L^2(\R^d;\C)}(\1_\Omega A \1_\Omega) =
\fint_{\Omega}  \rho_A,
$$
where $\1_\Omega$ is the characteristic function of the unit cell
$\Omega$, and $\fint_\Omega$ is a shorthand notation for $\frac 1
{|\Omega|} \int_{\Omega}$. This formula is independent of the choice
of the unit cell $\Omega$.

\medskip

Since we are dealing here with periodic materials, we will use the {Bloch transform} (also
called Bloch-Floquet transform) \cite{ReedSimon4, Kuchment, Wilcox}. For $K \in \cR^\ast$, let $T_K$ be the unitary multiplication operator on $L^2_{\rm per}$ defined by 
$$
\forall v \in L^2_{\rm per}, \quad (T_Kv)(x) = {\rm e}^{-iK\cdot x} v(x) \quad \mbox{ for a.a. } x \in \R^d,
$$
and 
\begin{align*}
&L^2_{\rm qp}(L^2_{\rm per})\\
&:=\left\{ \R^d \ni k \mapsto u_k \in L^2_{\rm per} \; | \; \int_\cB \|u_k\|_{L^2_{\rm per}}^2 \, dk < \infty, \; u_{k+K}=T_K u_k \mbox{ for all } K \in \cR^\ast \mbox{ and a.a. } k \in \R^d \right\},
\end{align*}
the Hilbert space of $\cR^\ast$-quasi-periodic $L^2_{\rm per}$-valued functions on $\R^d$ endowed with the inner product
$$
\langle u,v\rangle_{L^2_{\rm qp}(L^2_{\rm per})}= \fint_\cB \langle u_k,v_k\rangle_{L^2_{\rm per}} \, dk.
$$
Here and below, the subscript qp refers to the quasi-periodicity property.
The Bloch transform then is the unitary map from $L^2(\R^d;\C)$ to $L^2_{\rm qp}(L^2_{\rm per})$ defined for $u \in C^\infty_{\rm c}(\R^d;\C)$ by
\begin{equation}\label{eq:Bloch_transform}
\forall k \in \R^{d}, \quad \forall x \in \R^d, \quad u_k(x) = \sum_{R \in \cR} u(x+R) \, {\rm e}^{-ik \cdot (x+R)}.
\end{equation}
Its inverse is given by
\begin{equation}\label{eq:inverse_Bloch_transform}
u(x) =\fint_\cB {\rm e}^{ik\cdot x} u_{k}(x) \, dk, \quad \mbox{for a.a. } x \in \R^d.
\end{equation}
Any $\cR$-periodic
operator $A \in {\mathcal L}(L^2(\R^d;\C))$ is decomposed by the Bloch
transform in the sense that there exists a function
$k \mapsto A_k$ in $L^\infty_{\rm qp}({\mathcal L}(L^2_{\rm per}))$ such that for any
$u \in L^2(\R^d;\C)$ and almost all $k \in \R^d$, $(Au)_k = A_ku_k$, and
\begin{equation}\label{eq:qpBloch}
A_{k+K} = T_K A_k T_K^\ast, \quad \mbox{ for all } K \in \cR^\ast \mbox{ and a.a. } k \in \R^d.
\end{equation}
The $A_{k}$'s are called the fibers of the operator $A$.
If $A \in \Schatten_{1,\per}$, then the function $k \mapsto A_{k}$ is
in $L^1_{\rm qp}(\Schatten_{1}(L^2_{\rm per}))$, the function $k
\mapsto \Tr(A_k)$ is in $L^1_{\rm loc}(\R^{d})$, $\cR^{*}$-periodic, and we have
$$
\Trper(A) =(2\pi)^{-d} \int_\cB \Tr(A_k) \, dk.
$$
The Bloch decomposition theorem can be extended to unbounded
$\cR$-periodic self-adjoint operators using the resolvent \cite{ReedSimon4}. 

\medskip

In the case of the periodic Hamiltonian operator $H$ given by \eqref{eq:periodic_Hamiltonian}, we have 
  \begin{align} \label{eq:Hk}
    H_{k} = \frac 1 2 (-i\nabla + k+ {\mathcal A})^{2} + V.
  \end{align}
  For each $k \in \R^d$, $H_k$ is a bounded below self-adjoint operator on $L^{2}_{\per}$ with domain $H^{2}_{\per}$ and compact resolvent.
  Let $(\lambda_{n,k})_{n \in \N^\ast}$ be the non-decreasing sequence of eigenvalues of $H_k$ counting multiplicities
  $$
  \lambda_{1,k}\leq \lambda_{2,k}\leq \lambda_{3,k} \leq \cdots , \quad \lim_{n \to \infty} \lambda_{n,k}=+\infty,
  $$
  and we use the convention $\lambda_{0,k} = -\infty$. We denote by
  $(u_{n,k})_{n \in \N^\ast} \in (H^{2}_{\per})^{\mathbb N^{*}}$ an $L^2_{\rm per}$-orthonormal basis of associated eigenfunctions:
  \begin{align*}
    H_{k} u_{n,k} = \lambda_{n,k} u_{n,k}, \quad \langle u_{m,k},u_{n,k}\rangle_{L^2_{\rm per}}=\delta_{m,n}.
  \end{align*}
  For $N \in \N^{*}$ and $k \in \R^d$, we will denote by
  \begin{align}
    P_{N,k} = \1(H_{k} \le \lambda_{N,k}).  
  \end{align}
  Whenever $\lambda_{N,k} < \lambda_{N+1,k}$, $P_{N,k}$ is the spectral projector on the eigenspace associated with the lowest
  $N$ eigenvalues of $H_{k}$ (counting multiplicities):
  \begin{align}
    \label{eq:def_PNk}
    P_{N,k} = \sum_{n=1}^{N} |u_{n,k}\rangle\langle  u_{n,k}|.
\end{align}
  Since $H_{k}$ is quasi-periodic, so is $P_{N,k}$, and the eigenvalues $\lambda_{n,k}$ are
  ${\mathcal R}^\ast$-periodic functions of $k$. By a min-max argument
  (see e.g. \cite{ReedSimon4,cances2018numerical}), there exists
  $\underline{C_1}, \overline{C_1} \in \R$, and
  $\underline{C_2},\overline{C_2}>0$ such that
  \begin{align}
    \label{eq:asymptotics_eigenvalues}
 \underline{C_1}+ \underline{C_2} n^{2/d}  \leq \lambda_{n,k} \leq  \overline{C_1}+ \overline{C_2} n^{2/d}.
\end{align}
 Denoting by $N_k$ the number of eigenvalues below the Fermi level $\mu_{\rm F}$ at $k$
  \begin{equation}\label{def:Nk}
  N_k = \Big|\Big\{\lambda_{n,k}\leq \mu_{\rm F}, n \in \N^\ast \Big\}\Big|,
  \end{equation}
we see that $N_k$ is bounded uniformly in $k$.

\medskip 

Let us now consider the {ground-state density matrix}
$\gamma(0) = \mathds 1(H \leq \mu_{\rm F})$ defined in~\eqref{eq:GSDM}. Its Bloch fibers  are
  \begin{equation}\label{eq:gk0}
  \gamma_k(0) =  \mathds 1(H_k \leq \mu_{\rm F}) = P_{N_{k},k}.
  \end{equation}
  The {\it current operator}
  $J_\alpha= -(-i \nabla + {\mathcal A}) \cdot e_{\alpha}$ defined
  in~(\ref{def:Jalpha}) is also $\cR$-periodic, with fibers
 $$
 J_{\alpha,k} = -(-i \nabla + k + {\mathcal A}) \cdot e_{\alpha} =
 - \nabla_k H_k \cdot e_{\alpha} =: -\partial_{\alpha} H_{k}.
 $$
 Note that the notation $\partial_\alpha$ denotes a derivative
 along the (not necessarily normalized) vector $e_{\alpha}$.
 
 \medskip 
  
Lastly, for each $q \in \R^d$, we denote that $G_q$ the unitary multiplication operator on $L^2(\R^d;\C)$ defined by 
\begin{equation}\label{eq:defGq}
\forall u \in L^2(\R^d;\C), \quad (G_qu)(x) = {\rm e}^{iq\cdot x} u(x) \quad \mbox{ for a.a. } x \in \R^d.
\end{equation}
The operator $G_q$ is not $\cR$-periodic, except 
when $q \in \cR^\ast$ (in which case~$G_{q}$ is fibered, with
$G_{q,k} = T_{-q}$ for all $k$). However, for any $\cR$-periodic
operator $A \in {\mathcal L}(L^2(\R^d;\C))$ and any $q \in \R^{d}$,
the operator $G_qAG_q^\ast$ is $\cR$-periodic and its Bloch
decomposition is given by
\begin{equation}\label{eq:fibers_eiqx}
(G_qAG_q^\ast)_k = A_{k-q}, \quad \mbox{ for a.a. } k \in \R^d.
\end{equation}

\subsection{The Bloch theorem}
Before attacking the well-posedness of the current
$j_{\alpha,\beta}^{\eps}(t) = \underline {\rm Tr}(J_\alpha
\gamma_\beta^\eps(t))$ for $\varepsilon, t \neq 0$, we first study an easier special case.
\begin{proposition}[Bloch theorem]\label{prop:null_current} The current satisfies
\begin{align*}
& j^0_{\alpha,\beta}(t) = 0, \quad \forall t \ge 0 \;\;\; \mbox{(no current in the absence of external field)}, \\
& j^\eps_{\alpha,\beta}(0)=0, \quad \forall \eps \ge 0 \;\;\;  \mbox{(continuity of the current at $t=0$).}
\end{align*}
\end{proposition}
\begin{proof}
  This is a classical statement going back to Bloch,
valid in a more general context. We adapt here the proof in \cite{bohm1949note}. We have
\begin{align*}
  J := j_{\alpha,\beta}^{0}(t) = j_{\alpha,\beta}^{\varepsilon}(0) = \underline{\Tr}(J_{\alpha} \gamma(0)) = {- (2\pi)^{-d}\int_{\BZ} } \Tr(\partial_{\alpha} H_{k} \gamma_{k}(0)) \, dk.
\end{align*}
Assume that this quantity is non-zero. Construct for $\delta \in \R$ a trial state
\begin{align*}
  \gamma^{\delta} = e^{-i \delta \vec{e}_{\alpha} x} \gamma(0) e^{i \delta \vec{e}_{\alpha} x},
\end{align*}
a periodic operator with fibers
\begin{align*}
  \gamma^{\delta}_{k} = \gamma_{k+\delta \vec{e}_{\alpha}},
\end{align*}
and compute
\begin{align*}
  \underline{\Tr}(H \gamma^{\delta}) &= (2\pi)^{-d} \int_{\BZ} \Tr(H_{k} \gamma_{k + \delta \vec{e}_{\alpha}}(0))\, dk\\
  &= (2\pi)^{-d} \int_{\BZ} \Tr(H_{k-\delta \vec{e}_{\alpha}} \gamma_{k}(0))\, dk\\
  &= \underline{\Tr}(H \gamma(0)) - \delta {(2\pi)^{-d}\int_{\BZ} } \Tr(\partial_{\alpha} H_{k} \gamma_{k}(0)) \, dk + O(\delta^{2}).
\end{align*}
Since $\gamma(0)$ is the ground state, $\underline{\Tr}(H \gamma^{\delta}) \le \underline{\Tr}(H
\gamma(0))$ for all $\delta$, and therefore $J = 0$.

\end{proof}

\subsection{Definition of the current}
  
For $\eps > 0$, the operator 
\begin{equation*}
H_\beta^\eps=H + \eps x_\beta=\frac 12 (-i\nabla + \mathcal{A} )^2 +V + \eps x \cdot e_{\beta}
\end{equation*} 
already introduced in~\eqref{def:Hbetaeps} is not $\cR$-periodic, and we would naively expect that the density matrix 
$$
\gamma_\beta^\eps(t) = {\rm e}^{-itH^{\eps}_\beta} \gamma(0) {\rm e}^{itH^{\eps}_\beta}
$$ 
at time $t > 0$ (already introduced in~\eqref{def:gammaeps}) is not either. Yet, this operator is in fact
$\cR$-periodic. Physically, this is due
to the fact that although the potential $V_{\rm el}(x):=\eps x \cdot e_{\beta}$
is not periodic, the field
${\mathcal E}=-\nabla V_{\rm el} = -\eps \vec \vec e_\beta$ to which
the electrons are subjected is constant, hence periodic. The proof of
this result relies on the standard gauge transform
\begin{align}
  \label{eq:gauge_change}
  \widetilde \psi(x,t) = \left((G_{\varepsilon t \vec{e}_{\beta}}\psi(\cdot,t)\right)(x) = {\rm e}^{i \varepsilon t x_{\beta}} \psi(x,t),
\end{align}
where the operator $G_q$ has been defined in~(\ref{eq:defGq}), and the introduction of the gauge-transformed operators
\begin{align}
  \label{gauge_u}
\widetilde {\mathcal U}_\beta^\eps(t,t') := G_{\varepsilon t \vec e_{\beta}} {\rm e}^{-i(t-t')H^{\eps}_\beta} G_{\varepsilon t' \vec e_{\beta}}^\ast, 
\end{align}
and
\begin{equation}\label{gauge}
\widetilde \gamma_\beta^\eps(t) := G_{\varepsilon t \vec e_{\beta}}  \gamma_\beta^\eps(t)  G_{\varepsilon t \vec e_{\beta}}^\ast = \widetilde {\mathcal U}_\beta^\eps(t)  \gamma(0) \widetilde {\mathcal U}_\beta^\eps(t)^\ast,
\end{equation}
where $\widetilde {\mathcal U}_\beta^\eps(t)$ is a short-hand notation for
$$
\widetilde {\mathcal U}_\beta^\eps(t) := \widetilde {\mathcal U}_\beta^\eps(t,0)  = G_{\varepsilon t \vec e_{\beta}} {\rm e}^{-itH^{\eps}_\beta}.
$$ 
Through the change of gauge \eqref{eq:gauge_change}, the dynamics induced by the time-independent but non-periodic Hamiltonian $H_{\beta}^{\varepsilon}$ is equivalent to the dynamics induced by the time-dependent periodic Hamiltonian
\begin{align}
  \label{widetilde_H}
  \widetilde{H}^\eps_\beta(t)= G_{\eps t \vec e_\beta} H_{\beta}^{\varepsilon} G_{\eps t \vec e_\beta}^\ast=\frac 12 (-i\nabla + \mathcal{A} - \eps\vec{e}_\beta t)^2 +V.
\end{align}
This change of gauge is standard in both the mathematical and physical
literature, as it turns the spatially inhomogeneous electric potential
$V_{\rm el} = - \eps x_{\beta}$ into a homogeneous (but
time-dependent) magnetic potential
${\mathcal A}_{\rm el} = - \eps \vec{e}_{\beta} t$, more convenient to
deal with here because it does not break periodicity. Physically, this
is a manifestation of the gauge invariance of the Schr\"odinger
equation, where an electric field
${\mathcal E} = -\nabla V_{\rm el} - \frac{\partial \mathcal A_{\rm el}}{\partial t}$
can be realized either through a scalar or vector potential. The Bloch
fibers of $\widetilde{H}^\eps_\beta(t)$ are
\begin{align}
  \label{widetilde_H_k}
\widetilde{H}^\eps_{\beta,k}(t) = \frac 1 2 (-i\nabla + k+ \mathcal A - \varepsilon \vec e_{\beta} t)^{2} + V = H_{k-\varepsilon \vec e_{\beta} t}.
\end{align}
We sum up these arguments in the proposition below, together with elements that we shall use for defining the current. The reader can refer to the articles~\cite{bouclet2005linear,leinfelder1981schrodinger} where part of the results of that Proposition are proved. 

 \begin{proposition}\label{prop:setting} 
 Let ${\mathcal A} \in L^4_{\rm per}(\R^d;\R^d)$ such that $\nabla \cdot {\mathcal A}=0$, and $V \in L^2_{\rm per}(\R^d;\R)$.
  \begin{enumerate}
  \item  For all $\eps \in \R$, the operator 
  $H^{\varepsilon}_{\beta} $ defined in~(\ref{def:Hbetaeps})
   is essentially self-adjoint on $C^{\infty}_{\rm c}(\R^{d};\C)$, and therefore admits
    a unitary propagator $({\rm e}^{-it H_{\varepsilon}^{\beta}})_{t \in \R}$ in $L^{2}(\R^{d};\C)$. 
  \item For all $t\in \R$, and $\eps \in \R$, the operator
$\widetilde{H}^\eps_\beta(t)$ defined in~(\ref{widetilde_H}) is
self-adjoint on $L^2(\R^d)$ with domain $H^2(\R^d;\C)$, and
$\cR$-periodic. The strongly continuous unitary propagator
$(\widetilde{\mathcal{U}}^\eps_\beta(t,t'))_{(t,t') \in \R \times \R}$
on $L^{2}(\R^{d};\C)$ defined in \eqref{gauge_u} is
 $\cR$-periodic for all $t,t'\in \R$, with fibers $\widetilde{\mathcal{U}}^\eps_{\beta,k}(t,t')$ solving
\begin{equation} \label{eq:propagator2}
i\partial_t \widetilde{\mathcal{U}}^\eps_{\beta,k}(t,t') =\widetilde{H}^\eps_{\beta,k}(t)\widetilde{\mathcal{U}}^\eps_{\beta,k}(t,t'), \quad \widetilde{\mathcal{U}}^\eps_{\beta,k}(t',t') = \mathrm{Id}_{L^2_{\rm per}}. 
\end{equation}
\item For all $t \ge 0$, and $\eps \in \R$,
  $J_\alpha \gamma_\beta^\eps(t) \in \Schatten_{1,{\rm per}}$. The
  current
  $j_{\alpha,\beta}^{\eps}(t) = \underline {\rm Tr}(J_\alpha
  \gamma_\beta^\eps(t))$ is well-defined and
    \begin{align} \label{eq:current-k}
j_{\alpha,\beta}^{\eps}(t) &= -  (2\pi)^{-d}\int_{\BZ} \Tr(\partial_{\alpha} \widetilde H_{\beta,k}^{\varepsilon}(t) \widetilde \gamma^{\varepsilon}_{\beta,k}(t))\, dk\\
&= - (2\pi)^{-d} \int_{\BZ} \Tr(\partial_{\alpha} H_{k-\eps \vec{e}_\beta t}  \widetilde{\mathcal U}_{\beta,k}^{\varepsilon}(t)\gamma_k(0) \widetilde{\mathcal U}_{\beta,k}^{\varepsilon}(t)^{*}) \, dk
  \end{align}

   \end{enumerate}
 \end{proposition}

The results of Proposition~\ref{prop:setting} are not new (some are classical) but are nevertheless proved in Section~\ref{sec:preliminaries} for the sake of completeness. The situation can be summed up in the commutative diagrams of Figure~\ref{fig:diagram}.

\begin{figure}[h!]
  \centering
  {
  \begin{tikzcd}[matrix scale=3]
    \gamma(0) \arrow{r}{{\rm e}^{-it H_\beta^\varepsilon}} \arrow{d}{Id} & \gamma^\varepsilon_\beta(t) \arrow{d}{{G_{\varepsilon t e_\beta}}} \\
    \widetilde{\gamma}(0)\arrow{r}{\widetilde{\mathcal U}_\beta^\varepsilon(t)}&\widetilde{\gamma}^\varepsilon_\beta(t)
  \end{tikzcd}
  \hspace{1.4cm}
  \begin{tikzcd}[matrix scale=3]
    \{\gamma_k(0)\}_{k\in \R^d } = \{\gamma^\varepsilon_{\beta,k}(0)\}_{k\in \R^d} \arrow{r}{} \arrow{d}{\{Id\}_{k\in \R^d}} & 
    \{\gamma^\varepsilon_{\beta,k}(t)\}_{k\in \R^d} \arrow{d}{{\rm translation} \;{\rm in} \; k \; {\rm space}\; {\rm by }  -\eps t \vec{e}_\beta } \\
    \{\widetilde \gamma^\varepsilon_{\beta,k}(0)\}_{k\in \R^d}= \{\gamma^\varepsilon_{\beta,k}(0)\}_{k\in \R^d} \arrow{r}{\{\widetilde{\mathcal U}_{\beta,k}^\varepsilon(t)\}_{k\in \R^d}}&   
    \{\widetilde{\gamma}^\eps_{\beta,k}(t) \}_{k\in \R^d} = \{\gamma^\varepsilon_{\beta,k-\varepsilon{ t e_\beta}}(t)\}_{k\in \R^d} 
  \end{tikzcd}
}

  \caption{\small Commutative diagrams of the relationships between density
    matrices $\gamma^{\varepsilon}_{\beta}$ and $\widetilde{\gamma}^{\varepsilon}_{\beta}$ (top) and the fibers
    $\gamma^{\varepsilon}_{\beta,k}$ of $\gamma^{\varepsilon}_{\beta}$ which decompose both  $\gamma^{\varepsilon}_{\beta}$ and $\widetilde
    \gamma^{\varepsilon}_{\beta}$
    (bottom). In the top diagram, $A \overset{U}{\to} B$ means that $B = U A U^{*}$. In the bottom diagram $\{A_k\}_{k\in \R^d} \overset{\{U_k\}_{k\in \R^d}}{\to} \{B_k\}_{k\in \R^d}$ means that $A$ and $B$ are $\cR$-periodic and that their fibers are related by $B_k= U_k A_k U_k^{*}$.}
  \label{fig:diagram}
\end{figure}

\medskip 

\begin{remark}
  \label{rem:only_depends_on_Hk}
  This proposition reduces the study of
  $j_{\alpha,\beta}^{\varepsilon}(t)$ to that of the dynamics of the
  time-dependent Hamiltonian $H_{k-\varepsilon e_{\beta} t}$. In
  particular, although we have focused on the specific Hamiltonian $H$
  given by~\eqref{eq:periodic_Hamiltonian}, all computations beyond
  the proof of this proposition will be based on the use of the three
  formulae: for all $k \in \R^{d}, t \in \R_+$,
\newcommand*\widefbox[1]{\fbox{\hspace{2em}#1\hspace{2em}}}
  \begin{empheq}[box=\widefbox]{align}
      &  \gamma_k(0) = \mathds 1(H_k \leq \mu_{\rm F}), \label{eq:current1} \\
      & i\partial_t \widetilde{\mathcal{U}}^\eps_{\beta,k}(t) =H_{k-\varepsilon \vec e_{\beta} t}\widetilde{\mathcal{U}}^\eps_{\beta,k}(t), \quad \widetilde{\mathcal{U}}^\eps_{\beta,k}(0) =\mathrm{Id}_{{\mathcal H}_{\rm f}}, \label{eq:current2} \\
      & j_{\alpha,\beta}^{\eps}(t) = - (2\pi)^{-d} \int_{\BZ} \Tr(\partial_{\alpha} H_{k-\eps \vec{e}_\beta t} \widetilde{\mathcal U}_{\beta,k}^{\varepsilon}(t)\gamma_k(0) \widetilde{\mathcal U}_{\beta,k}^{\varepsilon}(t)^{*}) \, dk, \label{eq:current3}
  \end{empheq}
 where the fiber ${\cal H}_{\rm f}$ is equal to $L^2_{\rm per}$ in our setting. 
Our results in the following sections can therefore be extended to
other Hamiltonians where $(H_{k})_{k \in \R^d}$ is a family of
bounded below self-adjoint operators on a Hilbert space $\cH_{\rm f}$ with
compact resolvent satisfying the quasi-periodicity conditions
  $$
 H_{k+K} = T_K H_k T_K^\ast, \quad \forall K \in \cR^\ast, \;  k \in \R^d,
  $$
  where $(T_K)_{K \in {\cal R}^\ast}$ is a unitary representation of
  the group ${\cal R}^\ast$ on~$\cH_{\rm f}$
  (see \eqref {eq:qpBloch}), and
  the boundedness conditions in Section \ref{sec:adiabatic}. This
  includes in particular spin-dependent continuous models,
  tight-binding lattice models (for which $\cH_{\rm f}=\C^M$), and 2D
  materials. It also contains the case of systems with a constant
  magnetic field where the flux per unit cell satisfies an adequate
  commensurability condition (see \cite{Freund_Teufel} or \cite{MoPaPiTe}, where ideas from~\cite{Zak} are implemented).
\end{remark}

\subsection{Insulators, non-degenerate metals, semimetals}
\label{sec:insmetsem}
As we said before, the position of the Fermi level in the band diagram $(\lambda_{n,k})_{n \in \N^\ast,\, k \in \cB}$ is key to determining the electronic properties of the medium. We define the Fermi surface sheets
\begin{align*}
  \mathcal S_{n}= \{k\in{\mathcal B} \; | \; \lambda_{n,k} =\mu_{\rm F}\}, \;\;n\in\N^\ast
\end{align*}
and the Fermi surface
\begin{equation}\label{eq:Fermi_surface}
  \mathcal S= \bigcup_{n \in \N^\ast} \mathcal S_{n} = \{k\in{\mathcal B} \; | \; \exists n\in\N^\ast \mbox{ s.t. } \lambda_{n,k}
   =\mu_{\rm F}\}.
   \end{equation}
We will be interested here in three types of systems that we now describe in
three mutually exclusive assumptions.

\begin{assumption}[insulator]
  \label{ass:insulators}
  The Fermi surface $\mathcal S$ is empty, and there exists $N_{\rm ins} \in
  \N^\ast$ such that $N_{k} = N_{\rm ins}$ for all $k \in \BZ$, i.e.
    $$
    \forall k \in \cB, \quad \lambda_{N_{\rm ins},k} < \mu_{\rm F} < \lambda_{N_{\rm ins}+1,k},
    $$
    or equivalently $\mu_{\rm F} \notin \sigma(H)$. 
\end{assumption}
In the case of
insulators, we have for all $k \in \R^{d}$
\[
  \gamma_{k}(0) = P_{N_{\rm ins},k},
\]
and $\gamma_k(0)$ is a real-analytic
${\mathcal R}^\ast$-quasi-periodic function.

\begin{assumption}[non-degenerate metal]
  \label{ass:metal}
The Fermi surface~$\mathcal S$ is
  non-empty and the following conditions are satisfied: for all $n \in \N^\ast$,
     \begin{itemize}
     \item $\mathcal S_{n} \cap \mathcal S_{n+1} = \emptyset$ (no crossing at the Fermi level);
     \item for all $k \in \mathcal S_{n}$, $\nabla \lambda_{n,k} \neq 0$ (no
       flat bands at the Fermi level).
     \end{itemize}
   \end{assumption}
   Note that this assumption was used in~\cite{cances2018numerical}.
   It ensures a smooth density of states at the Fermi level.
   In this case, the Fermi surface consists of a finite union of
   disjoint smooth closed surfaces~$\mathcal S_n$. Letting
   \begin{align*}
     \mathcal B_{n} = \{k \in \BZ \; | \;  \lambda_{n,k} < \mu_{\rm F} < \lambda_{n+1,k}\},
   \end{align*}
   we obtain a partitioning
   \begin{align*}
     \mathcal B = \mathcal S \bigcup \left(\bigcup_{n \in \N^\ast}
   \BZ_{n}\right).
   \end{align*}
   Both $N_{k}$ and the fibers $\gamma_{k}(0) = P_{N_{k},k}$ of the
   density matrix $\gamma(0)$ are smooth on each $\BZ_{n}$, and have
   {discontinuities} on the sheets $\mathcal S_{n}$.
   
   \medskip
   
   \begin{assumption}[semimetal]
  \label{ass:semimetal}
  The dimension $d$ is equal to 2, there is $N_{\rm sm}$ such that
  $\lambda_{N_{\rm sm}, k} \le \mu_{F}$ for all $k \in \BZ$, and the
  Fermi surface ${\mathcal S}$ consists of a finite number of isolated
  points $(k_{i})_{i \in \mathcal I}$ (``Dirac points''). All these points are conical
  crossings: for all $i \in \mathcal I$,
  \begin{eqnarray}
    \label{hyp:semimetals_1}
    &\lambda_{N_{\rm sm}-1,k_i} < \lambda_{N_{\rm sm},k_i}=\mu_{\rm F} =\lambda_{N_{\rm sm}+1,k_i} < \lambda_{N_{\rm sm}+2,k_i},&\\
    \label{hyp:semimetals_2}
   & \lambda_{N_{\rm sm},k} = \mu_{\rm F} - v_{F,i} |k - k_{i}| + O(|k-k_{i}|^{2}),&\\
    \label{hyp:semimetals_3}
  &  \lambda_{N_{\rm sm}+1,k} = \mu_{\rm F} + v_{F,i} |k - k_{i}| + O(|k-k_{i}|^{2}), &
  \end{eqnarray}
  for some $v_{F,i} \in \R$. Furthermore, in this case we assume that
  $\mathcal A = 0$, so that the system has the time-reversal symmetry
  $H_{-k} = \overline{H_{k}}$.
\end{assumption}

Note that we assumed in Assumption~\ref{ass:semimetal} that
$\mathcal A = 0$ to ensure time-reversal symmetry. We require more
regularity on $V$ than in the previous assumptions to be able to prove a
Dyson expansion for the propagator (see Proposition~\ref{prop:Dyson}).
For the sake of clarity, we consider a model of 2D semimetals set in
$\R^2$, but our arguments can be adapted to the more physical case of
a model set in $\R^3$ (see also Remark~\ref{rem:only_depends_on_Hk}).

Assumption \ref{ass:semimetal} is generic in the case of potentials possessing the
symmetry of honeycomb lattices, such as graphene
\cite{fefferman2012honeycomb}.  In this case, there are two non-equivalent
Dirac points in the Brillouin zone ($|\mathcal I| = 2$), usually denoted by $K$ and $K'$, and we have $K'=-K$ and $v_{F,1}=v_{F,2}$. The constant $v_F=v_{F,1}=v_{F,2}$ is known as the
Fermi velocity.
More generally, Dirac points generate specific dynamical behaviors that have been studied in~\cite{FerGe2002,FerMehats2016} in the context of the Dirac operator. Such phenomena also appear in molecular dynamics (see~\cite{Hag94,FermanianKammerer2017,FermanianLasser2008}).

     \medskip 

     In the semimetal case, $N_{k} = N_{\rm sm}$ for almost every $k
     \in \R^{2}$, and $\gamma_{k}(0)$ is singular at each $k_{i} \in
     \mathcal S$.
     
\subsection{Main results: the current}
In the following results, we use the notation $O(f(\varepsilon,t))$ to
denote a quantity bounded by $C f(\varepsilon,t)$ where $C$ is a
constant that might depend on the material through $V$,
$\mathcal A$ and $\mu_{F}$, but not on~$t$ and~$\varepsilon$.

\begin{theorem}[insulators]
  \label{thm:insulator}
  Assume the system is an insulator (Assumption \ref{ass:insulators}).
  Then there exists $\eta>0$ such that for all $\varepsilon, t \in \R_+$,
 \begin{align*}
\frac 1{t}\int_{0}^{t} \frac{j_{\alpha,\beta}^{\varepsilon}(t')}{\varepsilon} dt' = {-i(2\pi)^{-d}\int_{\BZ}} \Tr \left( \gamma_{k}(0) [\partial_{\alpha} \gamma_{k}(0),\partial_{\beta}\gamma_{k}(0)]\right) dk +{O\left(\left(\frac 1 t + \eps (1+t) \right) {\rm e}^{\eta \eps t}\right).}
  \end{align*}
\end{theorem}
Note that this implies in particular that
$$
 \sigma_{\alpha,\beta} = \lim_{t\to \infty} \frac 1{t}\int_{0}^{t} \lim_{\varepsilon \to 0}\frac{j_{\alpha,\beta}^{\varepsilon}(t')}{\varepsilon} dt' 
 = e_\alpha^T \sigma^\perp e_\beta,
 $$
 where $\sigma^\perp$ is a real antisymmetric matrix with components
 \begin{equation}\label{eq:intBerry}
\sigma^\perp_{ij} := {(2\pi)^{-d}\int_{\BZ}} -i \Tr \left( \gamma_{k}(0) \left[\frac{\partial \gamma_{k}}{\partial k_i}(0),\frac{\partial \gamma_{k}}{\partial k_j}(0)\right]\right) \,  dk.
\end{equation}
The integrand in \eqref{eq:intBerry} is related to the well-known Berry curvature associated to the first
$N_{\rm ins}$ bands, that is to the 2-form
$$
\sum_{1 \le i < j \le d} \Omega_{ij}(k) \, dk_i \wedge dk_j \quad \mbox{where} \quad \Omega_{ij}:=-i \Tr \left( \gamma_{k}(0) \left[\frac{\partial \gamma_{k}}{\partial k_i}(0),\frac{\partial \gamma_{k}}{\partial k_j}(0)\right]\right)
$$
For $d=2$, we have 
$$
\sigma^\perp_{12}=(2\pi)^{-1} \mbox{Ch}_1(\gamma_\bullet(0)),
$$ 
where $\mbox{Ch}_1(\gamma_\bullet(0)) \in \Z$ is the first Chern of the fiber bundle defined by the quasi-periodic function $k \mapsto \gamma_k(0)$~\cite{thouless1982quantized,simon1983holonomy}. This relationship between the
transverse bulk transport properties and the Chern number,
characteristic of the integer quantum Hall effect, is known as the
TKNN formula. 

If $\mathcal A = 0$, then the system has the time-reversal symmetry $H_{-k}= \overline{H_k}$. 
As is classical, the Berry curvature is then odd, and the transverse conductivity matrix $\sigma^\perp$ equal to zero \cite{thouless1982quantized}.

\medskip 

\begin{theorem}[conductivity in non-degenerate metals]
  \label{thm:metal}
 Assume the system is a non-degenerate metal (Assumption~\ref{ass:metal}). 
 \begin{enumerate}
 \item Let $\theta > 0$. For all $\varepsilon > 0$ small enough and 
   $0 \le t \le \frac 1 \varepsilon \varepsilon^{\theta}$, we have
   \begin{align}
     \label{eq:result_metal}
     j_{\alpha,\beta}^{\varepsilon}(t) =D_{\alpha,\beta} \varepsilon t + O(\varepsilon + \varepsilon^{2} t^{2})
   \end{align}
   where
  \begin{align} \label{eq:Dab}
    D_{\alpha,\beta}:=  (2\pi)^{-d} \sum_{n\in\N^\ast} \int_{\mathcal S_n} \partial_{\alpha} \lambda_{n,k} \, (ds \cdot \vec e_{\beta})
  \end{align}
\item If furthermore there exists $N_{\rm met} \in \N^\ast$ such that
  $\lambda_{N_{\rm met}-1,k} < \mu_{\rm F} < \lambda_{N_{\rm
      met}+1,k}$ for all $k \in \BZ$ and there are uniform gaps between
  $\lambda_{N_{\rm met}-1,k}$ and $\lambda_{N_{\rm met},k}$ on the one
  hand, and $\lambda_{N_{\rm met},k}$ and $\lambda_{N_{\rm met}+1,k}$
  on the other hand, then there exists $\eta > 0$ such that, for all
  $\varepsilon, t \in \R_+$,
  \begin{align} \label{eq:jab_metal}
   j_{\alpha,\beta}^{\eps}(t) 
   =-(2\pi)^{-d}{\int_{\BZ}}  \mathds 1(\lambda_{N_{\rm met}, k}
\leq \mu_{\rm F}) \partial_{\alpha} \lambda_{N_{\rm met}, k+\varepsilon e_{\beta} t}dk + O((\eps+\varepsilon^{2} t)  {\rm e}^{\eta \eps t} ).
 \end{align}
  \end{enumerate}
\end{theorem}

\medskip 
   
Note that under the assumptions of the case 2 above, the lowest $N-1$
bands are completely filled, the $N^{\rm th}$ band is partially
filled, and the other bands are empty. Still in the setup of case 2,  it follows  from \eqref{eq:result_metal} and \eqref{eq:jab_metal} that
four different regimes can be observed for $\varepsilon \ll 1$
  \begin{enumerate}
\item For very short times $t \ll 1$, quantum fluctuations of order $O(\eps)$ dominate the current:
$$
j_{\alpha,\beta}^{\eps}(t) = O(\eps);
$$
\item For $1 \ll t \ll \frac{1}{\eps}$, the electrons undergo
  ballistic transport:
$$
j_{\alpha,\beta}^{\eps}(t) \approx D_{\alpha,\beta} \eps t,
$$
where $D_{\alpha,\beta}$ is defined in \eqref{eq:Dab};
\item For $ \frac{1}{\eps} \ll t \ll
  \frac{1}{\eps}\log(\eps^{-\zeta})$ with $\zeta\in(0,\eta^{-1})$, we observe Bloch oscillations
$$
j_{\alpha,\beta}^{\eps}(t) \approx {-(2\pi)^{-d}\int_{\mathcal B}  \mathds 1(\lambda_{N_{\rm met}, k}
\leq \mu_{\rm F})}{ \partial_{\alpha} \lambda_{N_{\rm met},k+\eps \vec e_\beta t} \, dk}.
$$
In particular, when $\vec e_{\beta}$ is commensurate with the
reciprocal lattice $\cR^\ast$, the current is well approximated in this regime by a periodic function of time with zero mean;
\item for times $t \gg \frac{1}{\eps}\log(\eps^{-\zeta})$, our
  estimates do not allow us to conclude. The proofs show that the factor
  ${\rm e}^{\eta \varepsilon t}$ is due to the unboundedness of the operator $H$ defined in~\eqref{eq:periodic_Hamiltonian}. 
  For tight-binding models, this factor
  ${\rm e}^{\eta\eps t}$ is not present, and we would observe Bloch
  oscillations up to times $t \ll \frac 1 {\varepsilon^{2}}$. The
  behavior for larger times is open.
\end{enumerate}

Note that some periodic metallic systems have a more complex crossing
structure than that assumed in the second case of Theorem
\ref{thm:metal}. This is the case in particular for the free electron
gas ($V=0$, ${\mathcal A}=0$, seen as a periodic system with an
arbitrary periodic lattice), which does not display Bloch
oscillations.

\medskip

\begin{remark}
  The coherent electronic transport model considered here neglects all
  sources of dissipation (phonons, impurities, electron-electron
  interactions). In the Drude approximation, these phenomena give rise
  to an effective timescale $\tau$ such that $1 \ll \tau \ll 1/\eps$
  (larger than the coherence timescale of the electrons, but smaller
  than the Bloch oscillations timescale), yielding a finite DC
  conductivity $\sigma_{\alpha,\beta} \sim D_{\alpha,\beta} \tau$. In
  usual metals at room temperature, dissipation is dominated by phonon
  scattering, and the relaxation time $\tau$ is of the order of tens
  of femtoseconds \cite{gall2016electron}. By contrast, the timescale
  of Bloch oscillations in most experiments is much larger. Only in structures such as semiconductor superlattices
  or cold atoms have Bloch oscillations been observed experimentally~\cite{leo1992observation}.
\end{remark}

\medskip

\begin{theorem}[conductivity in semi-metals] \label{theo:semi-metal}
  Assume that the system is a semimetal
  (Assumption~\ref{ass:semimetal}). Assume furthermore that
  $V \in H^{1}_{\per}$. Then,
  \begin{eqnarray*}
  \sigma_{\alpha,\beta}=  \lim_{t\to \infty} \frac 1{t}\int_{0}^{t} \lim_{\varepsilon \to 0}\frac{j_{\alpha,\beta}^{\varepsilon}(t')}{\varepsilon} dt' =  \frac{ |\mathcal I|}{16} \, 
    e_{\alpha} \cdot e_{\beta}.
\end{eqnarray*}
\end{theorem}

Semimetals are intermediate between insulators and metals, possessing a
finite longitudinal conductivity in the linear response regime. This
is due to the peculiar properties of the Dirac points. Note that the
value of the conductivity is universal, not depending on the
characteristics of the Hamiltonian but only on the number of conical crossings. More precisely, the conductivity tensor is isotropic and each conical intersection contributes as $\frac{1}{16}$ to the total conductivity. Note that this result is consistent with formula (1.17a) in \cite{DasSarma2011}.


\section{Numerics}
\label{sec:numerics}
Before turning to the proofs, we illustrate our results with numerical
simulations. As mentioned in Remark~\ref{rem:only_depends_on_Hk}, our
results also apply to tight-binding models, and only depend on the
form of $H_{k}$. We test on a very simple model of $H_{k}$, adapted
from the Haldane model \cite{haldane1988model} (itself based on a tight-binding model of
graphene), that can support many phases depending on the values of its parameters.
The graphene lattice $\cR$ is spanned by the vectors
\begin{align*}
  a_{1} = \left(\frac{\sqrt 3}{2},\frac 1 2\right), \quad a_{2} = \left(\frac{\sqrt 3}{2},-\frac 1 2\right),
\end{align*}
and $\cR^{*}$ by the vectors
\begin{align*}
  b_{1} = 2\pi \left(\frac{1}{\sqrt 3},1\right), \quad b_{2} = 2\pi \left(\frac{1}{\sqrt 3},-1\right).
\end{align*}
The Hamiltonian fibers are
 \[ H_k=
  \begin{pmatrix}
    m(k)&\overline{{f}(k)}\\f(k) &-m(k)
  \end{pmatrix},\]
  with
 \begin{align*}
  m(k) &= g + 2 t_{2} \Big(\sin( k \cdot a_{1}) + \sin(k \cdot a_{2}) + \sin( k \cdot (a_{1} - a_{2}))\Big), \\
  f(k) &= \sum_{i=1}^{3} {\rm e}^{i k \cdot \delta_{i}},\\
  \delta_{1} &= \left(\frac 1 {\sqrt 3}, 0\right), \quad \delta_{2} = \left( -\frac{1}{2\sqrt 3},\frac 1 2\right), \quad \delta_{3} =\left( -\frac{1}{2\sqrt 3},-\frac 1 2\right).
\end{align*}
The eigenvalues of $H_k$ are
$\lambda_{\pm} = \pm \sqrt{m(k)^{2} + |f(k)|^{2}}$. With $g=0$,
$t_{2} = 0$, this is the standard model of graphene: two bands
touching at level $0$ at two inequivalent points in the Brillouin
zone, where $f(k)$ vanishes. The parameter $g \neq 0$ opens a gap of
size~$2g$. The parameter $t_{2}$ models an internal magnetic field,
and can turn the system into a Chern insulator (in particular, with
$g=1, t_{2}=-1$, the system is a Chern insulator with Chern number
$+1$). Therefore, varying the parameters $g$, $t_{2}$ and $\mu_{\rm F}$,
we can obtain a normal insulator, a Chern insulator, a semimetal or a
metal.

For a given set of parameters, we compute the current by using
formulae \eqref{eq:current1}-\eqref{eq:current3}. We sample the Brillouin zone using a uniform
grid with $N_{\rm grid}=300$ points per direction, and solve the ordinary differential equation
\begin{align*}
  i \frac{d u}{d t}(t) = H_{k-\eps \vec{e}_{\beta}t} u(t), \quad u(0) = u_{n,k}, 
\end{align*}
for various $n$ and $k$ using the \texttt{DifferentialEquations.jl}
Julia package \cite{rackauckas2017differentialequations} with the default
Tsitouras method of order 5.

Our parameter values are collected in
Table~\ref{tab:parameter_values}.
\begin{table}[h!]
  \centering
  \begin{tabular}[h!]{|c|c|c|c|c|}
    \hline
    Panel & $g$ & $\mu_{\rm F}$ & $t_{2}$ & Phase\\
    \hline
    (a) & $1$ & $0$ & $0$ & Normal insulator\\
    (b) & $1$ & $0$ & $-1$ & Chern insulator\\
    (c) & $1$ & $-2$ & $0$ & Metal\\
    (d) & $0$ & $0$ & $0$ & Semimetal\\
    \hline
  \end{tabular}
  \caption{\small Parameter values for the experiments in Figure \ref{fig:num_results}}
  \label{tab:parameter_values}
\end{table}

Our results are presented in the linear response regime
($\varepsilon = 10^{-6}$, $t \ll \frac 1 \varepsilon$) in
Figure~\ref{fig:num_results}.
\begin{figure}[h!]
  \centering
  \begin{subfigure}[h]{.49\textwidth}
    \includegraphics[width=\textwidth]{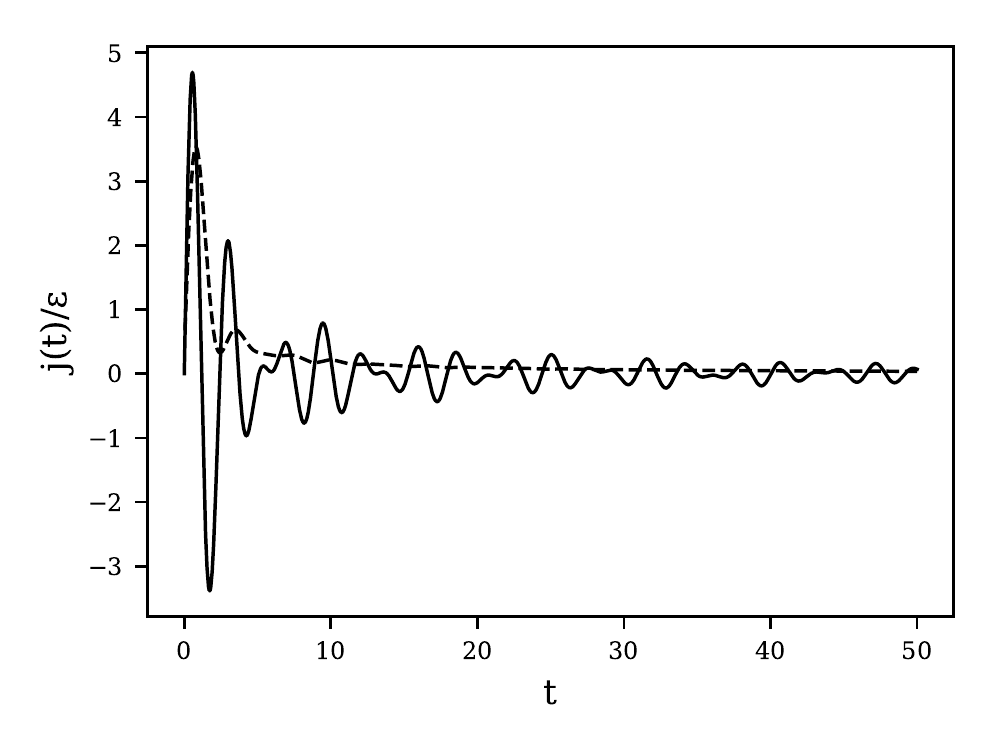}
   \caption{\small Normal insulator phase, longitudinal current.}
  \end{subfigure}
  \begin{subfigure}[h]{.49\textwidth}
    \includegraphics[width=\textwidth]{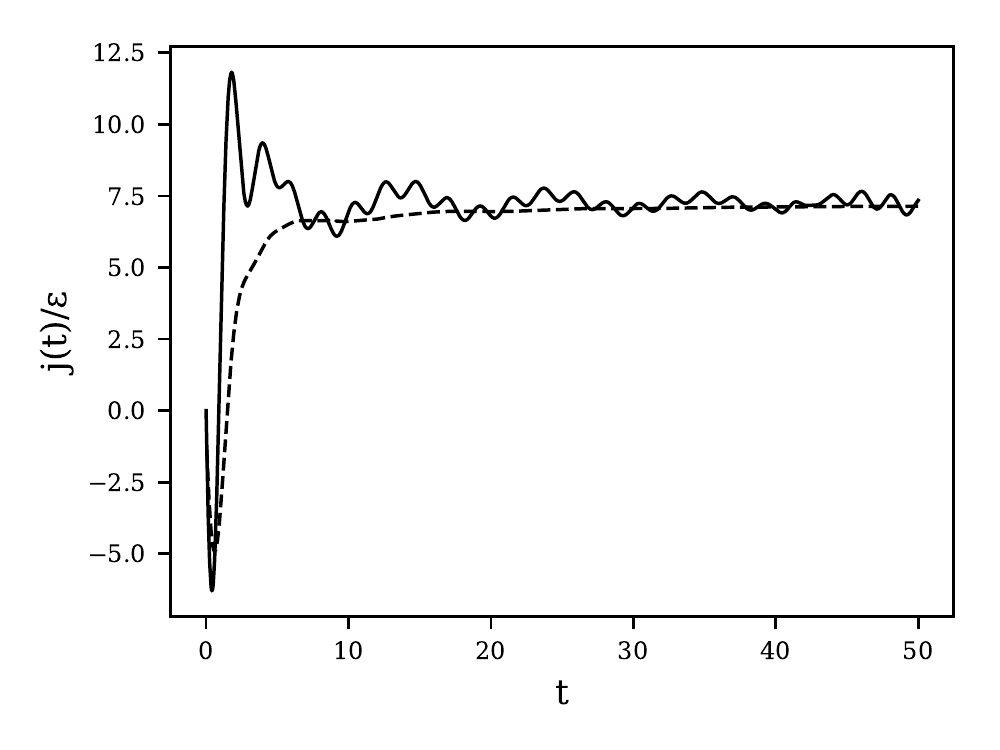}
    \caption{\small Chern insulator phase, transverse current. }
  \end{subfigure}
  
  \begin{subfigure}[h]{.49\textwidth}
    \includegraphics[width=\textwidth]{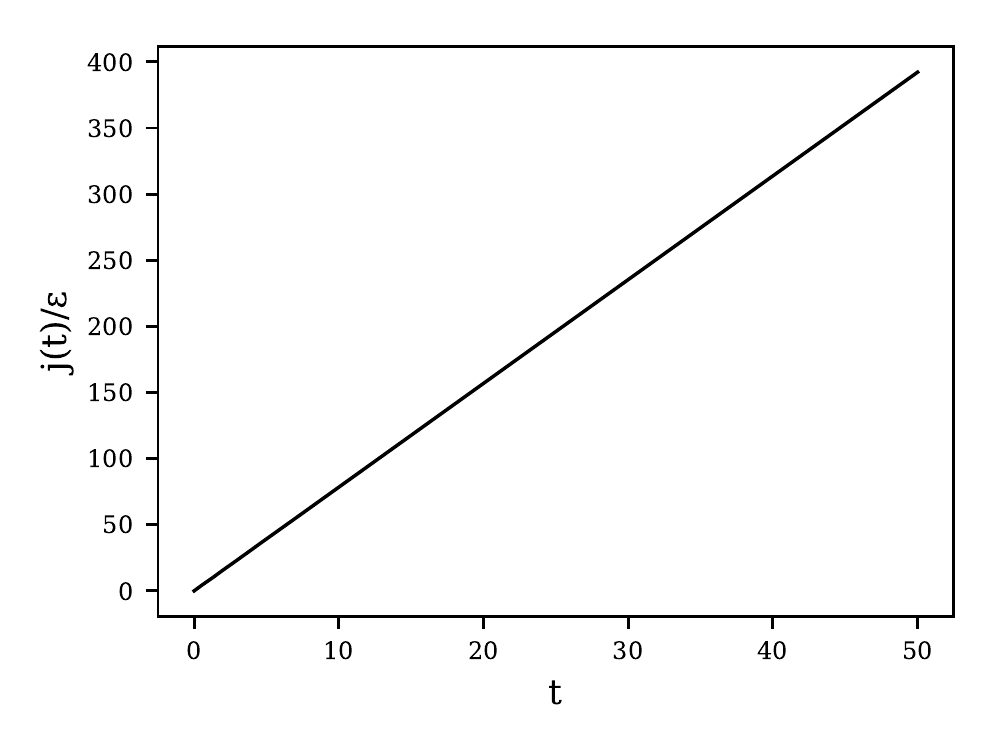}
    \caption{Metallic phase.}
  \end{subfigure}
  \begin{subfigure}[h]{.49\textwidth}
    \includegraphics[width=\textwidth]{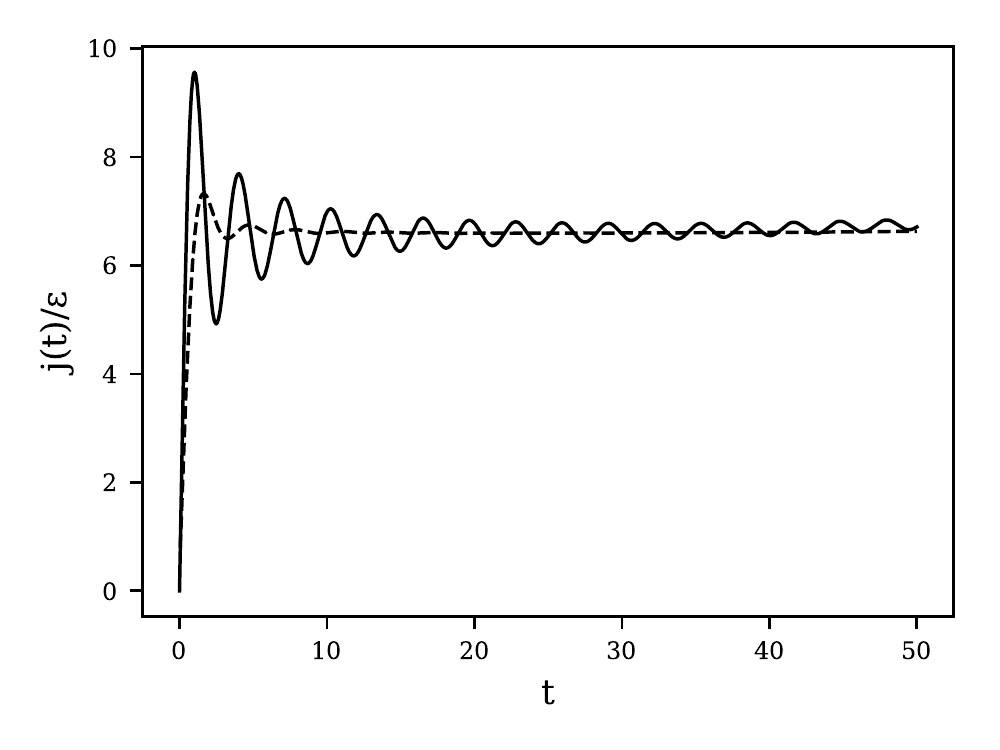}
    \caption{\small Graphene phase.
    }
  \end{subfigure}

  \caption{\small Instantaneous conductivity  $\frac{j_{\alpha,\beta}^{\eps}(t)}\eps$ (solid line) and running
    average $\frac 1 {t} \int_{0}^{t} \frac {j_{\alpha,\beta}^{\eps}(t')}{\varepsilon}  dt'$
    (dotted line) for several phases, in the
    linear response regime ($\varepsilon = 10^{-4}$, $t \ll \frac 1 \varepsilon$). In all cases
    $\vec e_{\beta} = b_{1}$, and $\vec e_{\alpha} = \vec e_{\beta}$,
    except in panel (b) where $\vec e_{\alpha} = b_{2}$.}
  \label{fig:num_results}
\end{figure}

These results are consistent with our theoretical results, including
the limit values of the conductivity in cases (b) and (d), where we obtain
$4\pi/\sqrt 3 \approx 7.26$ 
and $|b_1|^2/8=2\pi^2/3 \approx 6.58$
respectively. However, there is an additional phenomenon worth of
note: in the case of insulators and graphene, the linear response
instantaneous conductivity
$j_{\alpha,\beta}(t) = \lim_{\varepsilon \to 0}
\frac{j_{\alpha,\beta}^{\varepsilon}(t)}{\varepsilon}$ seems to possess a finite
limit as $t \to +\infty$. This is not captured by our results, where
we used an averaging process to suppress the oscillations. Note
that for a finite $N_{\rm grid}$, the linear response oscillates with
frequencies $\lambda_{n',k} - \lambda_{n,k}$ for
$\lambda_{n,k} < \mu_{\rm F} < \lambda_{n',k}$, and $k$ in the
discrete Brillouin zone. Only in the limit $N_{\rm grid} \to \infty$
do these resonances merge together to yield a finite limit for the
current. This is linked to the absence of resonances (parallel bands)
in our model. A deeper investigation of this effect would be
interesting future work.

We also investigate the Bloch oscillations regime
$\varepsilon \ll 1, \;\frac 1 \eps \ll t$ in Figure
\ref{fig:bloch_oscillations}, where we use the same parameters as in case
(c) above. The result is consistent with our theoretical result:
periodic or quasi-periodic oscillations, depending on whether
$e_{\beta}$ is commensurate with the reciprocal lattice or not.

\begin{figure}[h!]
  \centering
    \includegraphics[width=.49\textwidth]{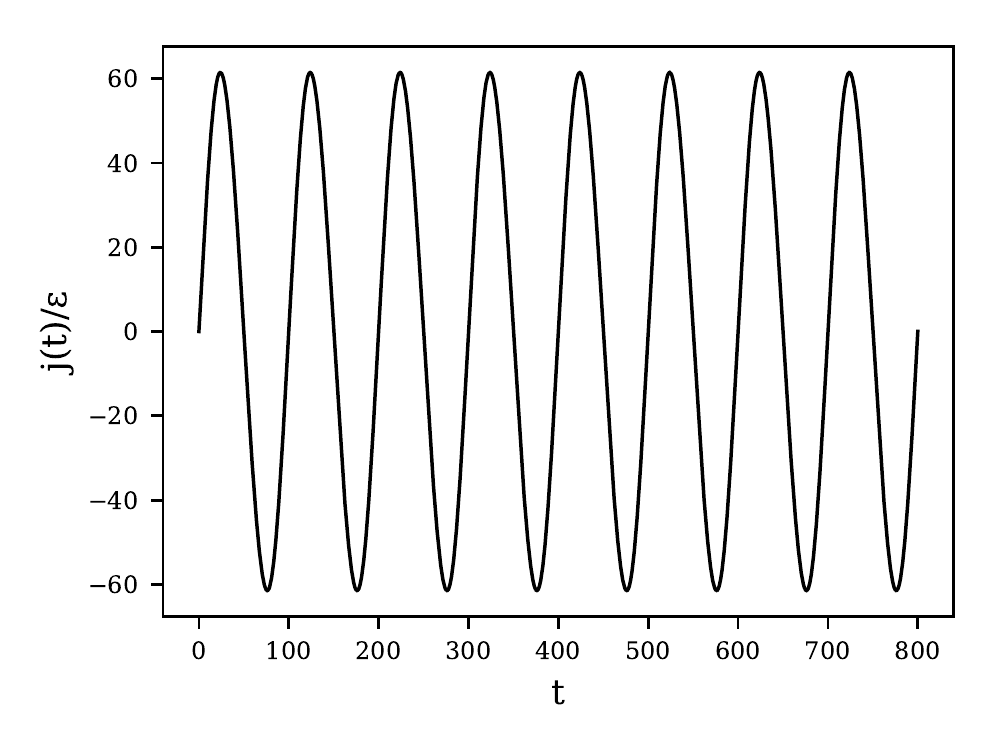}
    \includegraphics[width=.49\textwidth]{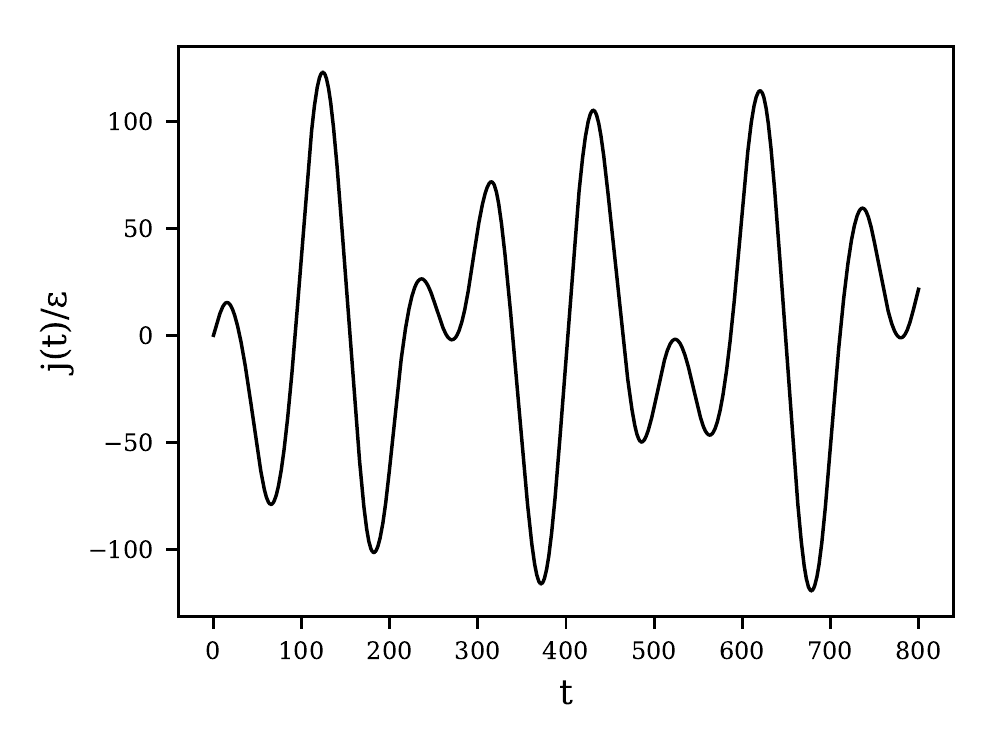}
    \caption{\small Instantaneous conductivity $\frac{j_{\alpha,\beta}^{\varepsilon}(t)}\varepsilon$ in the Bloch oscillations regime ($\varepsilon = 10^{-2}$, $\frac 1 \eps \ll t$). We
      take $e_{\alpha} = b_{1}$. The left figure is with
      $e_{\beta} = b_{1} + b_{2}$, the right one with
      $e_{\beta} = b_{1} + \frac{1+\sqrt 5}{2} b_{2}$.}
  \label{fig:bloch_oscillations}
\end{figure}


\section{Bloch decomposition of $\gamma^\eps_\beta(t)$ and regularity of the current}\label{sec:preliminaries}

In this section, we prove Proposition~\ref{prop:setting}. We first
point out an alternative interpretation that helps shedding some light
on the gauge change $G_{\varepsilon t \vec e_{\beta}}$. Formally,
$\gamma^\eps_\beta(t)$ satisfies the equation
$$i\partial_{t} \gamma^\eps_\beta = [H,\gamma^\eps_\beta] + \eps [x_{\beta},\gamma^\eps_\beta ].$$
The operator $[x_{\beta},\gamma^\eps_\beta]$ can easily be seen to be $\cR$-periodic, with
fibers $i\partial_{\beta} \gamma^\eps_{\beta,k}$ ({where $\partial_\beta = e_\beta \cdot \nabla_{k}$}), and therefore,~$\gamma^\eps_\beta(t)$
 is $\cR$-periodic and its fibers
$\gamma^\eps_{\beta,k}(t)$ satisfy the equation
\begin{align*}
i  \partial_{t} \gamma^\eps_{\beta,k} -i \eps \partial_{\beta} \gamma^\eps_{\beta,k} =   [H_{k},\gamma^\eps_{\beta,k}] =  L_{H_k} \gamma^\eps_{\beta,k},
\end{align*}
where $L_{H_k}:= [H_k, \cdot]$ is the Liouvillian associated with the operator $H_k$ (see Section~\ref{sec:louvillian_intro}).
The left-hand side of this equation is a linear advection equation,
which suggests the use of the method of characteristics: setting
\begin{align}
  \label{eq:change_of_variable}
\widetilde{\gamma}^{\eps}_{\beta,k}
\left(t\right) = \gamma^\eps_{\beta,k - \eps \vec{e}_{\beta}t}(t),
\end{align}
we obtain
\begin{align*}
  \label{eq:liouville-kL}
  i  \partial_t \widetilde{\gamma}^{\eps}_{\beta,k}(t) =  [H_{k- \eps   \vec{e}_{\beta}t}, \widetilde{\gamma}^{\eps}_{\beta,k}(t) ] = L_{H_{k-\eps \vec{e}_\beta t}} \widetilde{\gamma}^\eps_{\beta,k} (t),
\end{align*}
which is equivalent to \eqref{eq:current2}. The use of the gauge
transform operator $G_{\eps t\vec e_\beta}$, equivalent to the change
of variable \eqref{eq:change_of_variable}, makes these remarks
rigorous.

\medskip

We now prove Proposition~\ref{prop:setting}. As outlined above, the
results of this proposition are well-known; they can in fact be
extended to the more general setting of ergodic magnetic Schr\"odinger
operators (see~\cite{bouclet2005linear}). We provide here an
elementary proof specific to the periodic case, and take this
opportunity to introduce notations and tools which will be useful in
the sequel.

\medskip 

{\bf Proof of the first assertion.} 
The essential self-adjointness of $H_\beta^\eps$ follows from an extension of the Faris-Lavine theorem~\cite[Theorem~X.38]{ReedSimon2}.  Let $\mathcal{C} = C^\infty_{\rm c}(\R^d;\C)$ be the set of infinitely differentiable, compactly supported functions.
\begin{lemma}[Faris-Lavine theorem with periodic vector potentials]
\label{MagneticFarisLavine}
Let $V$ and $W$ be real-valued measurable functions on $\R^d$, 
$W \in L^2_{\rm loc}(\R^d;\R)$ and $\mathcal{A} \in L^4_{\rm per}(\R^d;\R^d)$ such that $\nabla \cdot \mathcal{A} = 0$ in the sense of distributions. Suppose that
\begin{enumerate}
\item there exists $c,f \in \R_+$ such that
$W(x) \geq - c|x|^2 - f$, for a.a. $x\in \R^d$;
\item  
$\frac 12 (-i\nabla +\mathcal{A})^2 + V + W +2c|x|^2$ is essentially self-adjoint on $\mathcal{C}$;
\item for some $a<1$, $\frac a2 (-i\nabla + \mathcal{A})^2 + V$ is bounded below on $\mathcal{C}$.
\end{enumerate}
Then $\frac 12 (-i\nabla +\mathcal{A})^2 + V + W$ is essentially self-adjoint on $\mathcal{C}$.
\end{lemma}
The proof of the above lemma is postponed until Appendix~\ref{sec:FL_lemma}.
\medskip

We apply Lemma~\ref{MagneticFarisLavine} with $V \in L^2_{\rm per}(\R^d;\R), W=\varepsilon x_{\beta}$. The
operator $\frac 12 (-i\nabla+\mathcal{A})^2 + V + \eps x_\beta + 2 |x|^2$ is
essentially self-adjoint on the core $\mathcal{C}$ in view of
\cite[Theorem~3]{leinfelder1981schrodinger} (note that $\eps x \cdot e_{\beta} \geq - |x|^2 -
\frac{\eps^2 |e_{\beta}|^{2}}{4} $).
Moreover, since $V$ is $L^2_{\rm per}(\R^d;\R)$, there exists $0<a<1$,
such that $\frac a2(-i\nabla+\mathcal{A})^2 + V$ is bounded below. This can be seen directly, or as a consequence of \cite[Theorem~3]{leinfelder1981schrodinger}.  Then,
Lemma~\ref{MagneticFarisLavine} gives that $H_\beta^\eps$ is
essentially self-adjoint on $\mathcal{C}$ and therefore admits a unique self-adjoint extension on $L^2(\R^d;\C)$. Hence,
the propagator of the associated Schr\"odinger equation is
well-defined, and explicitly given by
$(\mathrm{e}^{-i t H_\beta^\eps })_{t \in \R}$.

\medskip

{\bf Proof of the second assertion.} The self-adjointness of the time-dependent Hamiltonian
\begin{align*}
  \widetilde{H}^\eps_\beta(t) = \frac 12 (-i \nabla +\mathcal{A} - \eps \vec{e}_\beta t)^2+ V,
\end{align*}
is a consequence of Lemma~\ref{MagneticFarisLavine}, by replacing $\mathcal A$ with $(\mathcal A - \varepsilon e_\beta t)$. 
To show the well-posedness of the dynamics, since $\widetilde{H}^\eps_\beta(t)$ is $\mathcal{R}$-periodic, it suffices to study its fibers.
Hence we consider the dynamics of a Schr\"odinger equation with Hamiltonian 
\begin{align*}
  \widetilde{H}^\eps_{\beta,k}(t) = \frac 12 (-i \nabla +\mathcal{A} + k - \eps \vec{e}_\beta t)^2+ V
\end{align*}
and we use the following lemma on the dynamics generated by
time-dependent perturbations of the free-particle Hamiltonian on
$L^2_{\rm per}$. 

\begin{lemma}
\label{lemmaUnitaryPropagator} Let $H_0:=-\frac 12\Delta$ be the free-particle Hamiltonian on $L^2_{\rm per}$, and a map 
\[[0,T] \ni t \mapsto H_1(t)\]
taking its values in the set of $H_0$-bounded self-adjoint operators on $L^2_{\rm per}$ with relative bound lower than $1$, that is: there exist $0<a<1$ and $b>0$ such that
\begin{align}
  \label{eq:hyp_lap_bounded}
\forall t \in [0,T],\quad \forall \phi \in H^2_{\rm per}, \quad \|H_1(t)\phi\|_{L^2_{\rm per}} \leq a \|H_0 \phi \|_{L^2_{\rm per}} + b \|\phi\|_{L^2_{\rm per}}.
\end{align}
Then, for all $t \in [0,T]$, the operator defined by
 $H(t) = H_0 + H_1(t)$
 is self-adjoint on $L^2_{\rm per}$ with domain $H^2_{\rm per}$, and there exists a unique unitary propagator $(\mathcal{U}(t))_{t\in[0,T]}$ on $L^2_{\rm per}$ such that for $t \in [0,T]$, and $\phi_0 \in H^2_{\rm per}$, $\phi: t \mapsto \mathcal{U}(t)\phi_0$ is in $C^1([0,T];H^2_{\rm per})$, and solves the time-dependent Schr\"odinger equation
\[i \partial_t \phi(t) = H(t) \phi (t), \qquad \phi(0) = \phi_0.\] 
\end{lemma}
The proof of the above lemma is postponed to Appendix~\ref{sec:unitaryprop}.

\medskip

For $t \in [0,T], k \in \R^d$, we have $\widetilde{H}^\eps_{\beta,k}(t) = H_0 + H_{1}(t)$, with
 \begin{align*}
   H_{1}(t) = \frac{1} 2 \Big[ (-i \nabla) \cdot (\mathcal A + k - \varepsilon e_{\beta} t) +  (\mathcal A + k - \varepsilon e_{\beta} t) \cdot (-i \nabla)  + (\mathcal A + k - \varepsilon e_{\beta} t)^{2} \Big]+ V.
 \end{align*}
 Using the Sobolev embeddings
 $H^{2}_{\per} \subset L^{\infty}_{\per}$,
 $H^{1}_{\per} \subset L^{6}_{\per}$ (recall that we assume $d\le 3$), the Coulomb gauge choice
 $\nabla \cdot \mathcal A = 0$ and the fact that
 $\mathcal{A}\in L^4_{\rm per}(\R^d;\R^d)$ and
 $V \in L^2_{\rm per}(\R^d;\R)$, it is standard that $H_{1}$ satisfies
 the conditions of Lemma~\ref{lemmaUnitaryPropagator}, and the result follows.

 \medskip

\noindent
{\bf Proof of the third assertion.} We first compute the fibers of the $\cR$-periodic operator 
$\gamma_{\beta}^{\varepsilon}(t)$. Using~(\ref{eq:fibers_eiqx}), we have 
\begin{align}
  \gamma^\eps_{\beta,k}(t) &= \left(G_{\eps t \vec{e}_\beta}^*  \widetilde{\mathcal{U}}^\eps_\beta(t) \gamma(0) \widetilde{\mathcal{U}}^\eps_\beta(t)^* G_{\eps t \vec{e}_\beta} \right)_k \\
  &= \left( \widetilde{\mathcal{U}}^\eps_\beta(t) \gamma(0) \widetilde{\mathcal{U}}^\eps_\beta(t)^*  \right)_{k+\eps t\vec{e}_\beta}  \nonumber \\ 
  &= \widetilde{\mathcal{U}}^\eps_{\beta,k+\eps t\vec{e}_\beta}(t) \gamma_{k+\eps t\vec{e}_\beta}(0) \widetilde{\mathcal{U}}^\eps_{\beta,k+\eps t\vec{e}_\beta}(t)^* \label{eq:gebkt} \\
  &= \sum_{n=1}^{N_{k+\eps t\vec{e}_\beta}}|\widetilde{\mathcal{U}}^\eps_{\beta,k+\eps t\vec{e}_\beta}(t) u_{n,k+\eps t\vec{e}_\beta}\rangle\langle \widetilde{\mathcal{U}}^\eps_{\beta,k+\eps t\vec{e}_\beta}(t) u_{n,k+\eps t\vec{e}_\beta}|. \nonumber
\end{align}
Since the $u_{n,k}$ are in $H^2_{\per}$, we deduce that 
$$
(J_{\alpha} \gamma_{\beta}^{\varepsilon}(t))_{k} = - (-i \nabla + k +
\mathcal A) \cdot e_{\alpha} \gamma^\eps_{\beta,k}(t)
= - \partial_{\alpha} H_{k} \gamma^\eps_{\beta,k}(t)
$$
is trace-class (and finite-rank) uniformly in $k \in \BZ$ and
therefore that the current $j^\eps_{\alpha,\beta}(t)= \underline {\rm Tr}\left(J_\alpha
    \gamma^\eps_\beta(t)\right)$ is well-defined. 
    
    \medskip 

As the function $k \mapsto  \Tr(\partial_{\alpha} H_{k} \gamma^\eps_{\beta,k}(t))$ is $\cR^\ast$-periodic, we also have
$$
j^\eps_{\alpha,\beta}(t) = -  (2\pi)^{-d}\int_{\BZ} \Tr(\partial_{\alpha} H_{k} \gamma^\eps_{\beta,k}(t)) \, dk= -  (2\pi)^{-d}\int_{\BZ} \Tr(\partial_{\alpha} H_{k-\eps \vec{e}_\beta t} \gamma^\eps_{\beta,k-\eps \vec{e}_\beta t}(t)) \, dk.
$$


\section{Perturbation theory for time-dependent Hamiltonians}
\label{sec:adiabatic}

In this section we consider the dynamics generated by a Hamiltonian
$H(s) = H(\varepsilon t)$, and in particular its action on eigenspaces
of $H(0)$. We begin with some elementary properties of the Liouvillian
in Section~\ref{sec:louvillian_intro}, then use it to study
subspace perturbation theory in Section~\ref{subsec:subspace}. We
establish an adiabatic theorem in Section~\ref{subsec:adiabatic},
and use it to study the time-dependent Hamiltonian
$H_{k-\varepsilon e_{\beta} t}$ in Section~\ref{subsec:appli}.
Finally, we prove a result in linear response with a remainder
independent of the gap in Section~\ref{subsec:lr}.

\subsection{The Liouvillian and its partial inverse}
\label{sec:louvillian_intro}
In order to formulate and interpret our results, it is convenient to
make use of the formalism of the Liouvillian and its partial inverse,
a classical tool in adiabatic theory and eigenvalue perturbation
theory \cite{teufel2003adiabatic, Kato}, although sometimes used
implicitly. This formalism was for instance used in the context of
transport properties in \cite{Bachmann2018,Teufel2019,Monaco2017}.
Recall that if $h$ is a bounded self-adjoint operator on a Hilbert
space ${\mathcal H}_{\rm f}$, the Liouvillian $L_h$ associated with
$h$ is the bounded linear operator on
${\mathcal L}({\mathcal H}_{\rm f})$ (such a mathematical object is
sometimes called a superoperator in the physics literature) defined by
\begin{equation}\label{eq:Liouvillian}
\forall A \in {\mathcal L}(\cH_{\rm f}), \quad L_h A = [h,A].
\end{equation}
Note that if $A$ is self-adjoint, then $L_{h} A$ is anti-self-adjoint
($iL_hA=i[h,A]$ is self-adjoint). The restriction of $L_h$ to the
space $\Schatten_2(\cH_{\rm f})$ of Hilbert-Schmidt operators on $\cH_{\rm f}$ is self-adjoint: for all $A,B \in \Schatten_2(\cH_{\rm f})$,
\begin{equation}\label{eq:LhAB}
(L_hA,B)_{\Schatten_2} = \Tr\left( [h,A]^\ast  B\right)  = \Tr\left(  (A^\ast h-hA^\ast) B\right) =  \Tr\left( A^\ast (hB-Bh) \right) = (A,L_hB)_{\Schatten_2}.
\end{equation}
The operator $L_h$ is to density matrices what the Hamiltonian $h$ is
to pure states: it is the infinitesimal generator of the
norm-continuous unitary group $({\mathfrak U}_h(t))_{t \in \R}$ on
$\mathcal L(\mathcal H_{\rm f})$ defined by
\begin{equation}\label{eq:Liouvillian2}
\forall A \in \mathcal L(\cH_{\rm f}), \quad {\mathfrak U}_h(t) A = e^{-ith}Ae^{ith}.
\end{equation}

In the  case when $h$ is an unbounded self-adjoint operator,
\eqref{eq:Liouvillian} does not make sense for all $A \in {\mathcal
  L}(\cH_{\rm f})$, but it is still possible to define the Liouvillian
$L_h$ as the infinitesimal generator of the strongly-continuous
unitary group $({\mathfrak U}_h(t))_{t \in \R}$ on $\mathcal
L(\mathcal H_{\rm f})$ defined by \eqref{eq:Liouvillian2}. It is then
an unbounded operator on $\mathcal L(\mathcal H_{\rm f})$,
self-adjoint on $\Schatten_2(\cH_{\rm f})$.

\medskip

If $\mathcal H_{f}$ is of  finite-dimension~{$N_{f}$}, the action of $L_{h}$ is
easily understood in an orthonormal eigenbasis $(e_{n})_{1\leq n\leq N_{f}}$ of $h$ with associated
 eigenvalues {$\lambda_{1}\leq\cdots\leq \lambda_{N_f}$}. Then,
\[
L_{h} |e_{n}\rangle\langle e_{m}| = (\lambda_{n} - \lambda_{m})
|e_{n}\rangle\langle e_{m}|.
\]
The operator $L_{h}$ is not invertible (for instance,
$L_{h} |e_{n} \rangle\langle e_{n}| = 0$). However, it is invertible
when restricted to the subspace of block off-diagonal matrices, i.e.
matrices $A$ such that $A_{nn'} = A_{mm'} = 0$ for $n,n' \le N < m,m'$
for a given $N$ such that $\lambda_{N+1} > \lambda_{N}$. Its partial
inverse $L_{h,N}^{+}$ is given by
\begin{align}\nonumber 
  L_{h,N}^{+}\Bigl( \sum_{1 \le n \le N,\, N<
    m  \le N_f} A_{mn} |e_{m}\rangle\langle e_{n}| &+ A_{nm}
  |e_{n}\rangle\langle e_{m}| \Bigr) \\
    \label{eq:explicit_inverse_liouvillian}
  &= \Bigl( \sum_{1 \le n \le N, \,N<
    m \le N_f} \frac{A_{mn} |e_{m}\rangle\langle e_{n}| - A_{nm}
  |e_{n}\rangle\langle e_{m}|}{\lambda_{m} - \lambda_{n}} \Bigr)
\end{align}
and $L_{h,N}^{+}$ is bounded in operator norm by $\frac{1}{\lambda_{N+1}-\lambda_{N}}$.

\medskip 

More generally, if $h$ is an unbounded self-adjoint operator, let $I$ be a closed bounded interval of~$\R$,
and assume that
$$
g := \min\left( 1,{\rm dist}\left(I,\sigma(h) \setminus (\sigma(h) \cap I) \right) \right) > 0.
$$
The associated spectral projector is
\begin{equation}
  P_{I,h}:={\mathds 1}_I(h)= \frac{1}{2\pi i} \oint_{\mathcal C} (z-h)^{-1} \, dz, \label{eq:int_rep}
\end{equation}
where $\mathcal C$ is a Cauchy contour in the complex plane such that
$\sigma(h) \cap I$ is inside $\mathcal C$ and $\sigma(h) \setminus
(\sigma(h) \cap I)$ is outside ${\mathcal C}$.
Generalizing the terminology of the finite-dimensional case, we call
off-diagonal operators (with respect to the splitting of
$\mathcal H_{f}$ induced by $P_{I,h}$) the elements of the closed
subspace
$$
\mathcal L^{\rm OD}_{h,I}:=\left\{ A \in \mathcal L(\cH_{\rm f}) \; | \; P_{h,I}AP_{h,I}=(1-P_{h,I})A(1-P_{h,I})=0 \right\}
$$
of $\mathcal L(\cH_{\rm f})$. This defines a
$\Schatten_{2}$-orthogonal splitting of operators into their diagonal
and off-diagonal parts. It is easily seen that
$\mathcal L^{\rm OD}_{h,I}$ is $L_h$-stable, and that $L_h$ is
invertible on $\mathcal L^{\rm OD}_{h,I}$ with a bounded inverse. We denote its partial
inverse by $L_{h,I}^{+}$, extended to all of
$\mathcal L(\cH_{f})$ by imposing that it vanishes on diagonal
operators. We then have
\begin{align*}
  L_{h,I} L_{h,I}^{+} A = L_{h,I}^{+} L_{h,I} A = P_{h,I} A (1-P_{h,I}) + (1-P_{h,I}) A P_{h,I}
\end{align*}
for all $A \in \mathcal L(\cH_{f})$.

It is easy to check that $L_{h,I}^{+}$ has an explicit contour integral
representation:
\begin{align}
&L_{h,I}^{+} A=  \frac{1}{2\pi i} \oint_{\mathcal C} (z-h)^{-1} [P_{h,I},A]  (z-h)^{-1} \, dz,  \quad \forall A \in \mathcal \mathcal {\cal L}(\mathcal H_{\rm f}),  \label{eq:inv_Liouvillian}
\end{align}
where $\mathcal C$ is a contour as above. From~\eqref{eq:inv_Liouvillian}, we see that, when $P_{h,I}$ is of finite rank~$\cal N$, $L_{h,I}^+A$ is of rank of most $2{\cal N}$.


\subsection{Subspace perturbation theory}
\label{subsec:subspace}

The Liouvillian is a powerful tool to write concisely the results of
subspace perturbation theory, which studies the time dependence of a
gapped subspace of a time-dependent Hamiltonian. We consider $T > 0$
and $(H(s))_{s \in [0,T)}$ a family of self-adjoint operators on a
Hilbert space~$\cH_{\rm f}$ sharing the same domain
$D \subset {\mathcal H}_f$ and satisfying the following assumptions:
  \begin{description}
  \item[H1] $H(s) \ge 1$ for all $s \in [0,T)$;
  \item[H2] for each $\phi \in D$, the map $s \mapsto H(s) \phi$ is in
    $C^{n}([0,T),\cH_{\rm f})$ for some $n \ge 1$. For all $1 \le l \le n$, the
    operator $H^{(l)}(s)$ is self-adjoint on
    $\cH_{\rm f}$ for all $s \in [0,T)$, and
    \begin{align} \label{eq:alphal}
      \alpha_{l} := \sup_{s \in [0,T)} \|H^{(l)}(s) H(s)^{-1}\|_{\cal L(\cH_{\rm f})} < \infty;
    \end{align}
\item[H3] there exist $M \in \R_+$ and bounded continuous functions
  $a_\pm : [0,T] \to \R$ with $0 \le a_- \le a_+ \le M$ defining
  bounded closed intervals $I(s)=[a_-(s),a_+(s)] \subset \R$ such
  that, for all $s \in [0,T),$
\begin{align*}
 &g(s):= \min \left( 1,  {\rm dist}(I(s),\sigma(H(s)) \setminus (\sigma(H(s)) \cap I(s))) \right) > 0, \\
& P(s):=\1_{I(s)}(H(s)) \mbox{ has a finite (constant) rank ${\cal N} \in \N^\ast$},
\end{align*}
\end{description}
Under these assumptions, we set 
\[ L^{+}(s):=L_{H(s),I(s)}^{+}.\]

\begin{proposition}
  \label{prop:adiab1}
  Assume {\rm {\bf H1}}, {\rm {\bf H2}} and {\rm {\bf H3}}.
Then, $P \in C^{n}([0,T), \mathcal L(\cH_{\rm
f}))$,  $L^{+} \in C^{n}([0,T), \mathcal L(\mathcal L(\cH_{\rm
f})))$, and
\begin{align}
  \dot{P}(s) &= L(s)^{+}[P(s), \dot{H}(s)]. \label{eq:Pprime}
\end{align}
Furthermore, there exist constants $C_1,\cdots,C_n \in \R_+$ depending only on
$\alpha_{1}, \cdots, \alpha_n$ and $M$ such that the following bounds hold
for all $0 \le l \le n$, $s \in [0,T)$ and $A \in \mathcal L(\cH_{\rm f})$:
\begin{align}
  \label{eq:bound_HP0}
  \|H(s) P(s)\| &\le M, \\
  \label{eq:bound_HP}
  \|H(s) P^{(l)}(s)\| &\le \frac{C_l}{g(s)^{l+1}},\\
  \label{eq:bound_HL}
  \|H(s) (L^{+})^{(l)}(s) A\| &\le \frac{C_l}{ g(s)^{l+3}} \|A\|.
\end{align}
In addition, $P^{(l)}(s)$ has rank at most $(l+1) \mathcal N $, and $(L^+)^{(l)}(s)
A$ has rank at most $c_{l}\mathcal N$ where $c_{l}$ is a constant
that only depends on $l$ (in particular, $c_0=2$ and $c_1=10$).
\end{proposition}
\begin{remark}
  The powers of the gap in the bounds \eqref{eq:bound_HP} and
  \eqref{eq:bound_HL} are too pessimistic, as could be shown by a more
  detailed analysis. For instance, in the case $l=0$, $L^{+}(s)$ can be
  seen from the arguments at the beginning of this section to be
  bounded by a constant times $\frac 1 {g(s)}$. Similarly, the operator $\dot{P}$ is
  bounded by a constant times $\frac 1 {g(s)}$, using
  \eqref{eq:Pprime}. Nevertheless, the above bounds are more straightforward to
  establish and will suffice for our purposes. 
\end{remark}

\begin{proof}
  Differentiating $L_{H(s)} P(s) = 0$, we get
\[
  L_{H(s)} \dot{P}(s) = [P(s), \dot{H}(s)].
\]
Since both $[P(s), \dot{H}(s)]$ and $\dot{P}(s)$ are
off-diagonal operators (the first by direct calculation, the second by
differentiating the relationship $P(s)^{2} = P(s)$), we deduce~\eqref{eq:Pprime}.
By the functional calculus, $\|H(s)P(s)\|=\|H(s)\1_{I(s)}(H(s))\| \le a_+(s) \le M$, whence~\eqref{eq:bound_HP0}.

\medskip 

  In the following we take for $\mathcal C(s)$ the rectangular
  contour centered at the center of $I(s)$, of length $|I(s)| + g(s)$ and
  height $g(s)$, so that
  \begin{equation}\label{eq:bound0}
   |{\cal C}(s)| \le 2M+4 \quad \mbox{and for all } z \in {\cal C}(s), \quad \left\| \frac{1}{z-H(s)} \right\| \le \frac{2}{g(s)}. 
  \end{equation}
  We use the integral representation~\eqref{eq:int_rep}:
  \begin{align}
    \label{eq:P_contour}
    P(s) = \frac{1}{2\pi i} \oint_{{\cal C}(s)} \frac 1 {z-H(s)} \, dz.
  \end{align}
Using for all $z \in {\cal C}(s)$ the bound
\begin{align}
  \label{eq:bound_resolvent}
  \left\|\frac{H(s)}{z-H(s)}\right\| = \sup_{\lambda \in \sigma(H(s))} \left| \frac \lambda {z-\lambda} \right| \le 1 + \sup_{\lambda \in \sigma(H(s))} \left| \frac z {z-\lambda} \right| \le 1 + \frac {2(M+g(s))} {g(s)} \le \frac{2M+3}{g(s)}
\end{align}
establishes \eqref{eq:bound_HP} for $l = 0$.

\medskip 

The contour ${\cal C}(s)$ in \eqref{eq:P_contour} above can be kept
fixed equal to ${\cal C}(s_0)$ for $s$ in a neighborhood of any
$s_0 \in [0,T)$. Using
\begin{align}
  \label{eq:diff_resolvent}
  \frac{d}{ds}\frac 1 {z-H(s)} = \frac 1 {z-H(s)} \dot{H}(s) \frac 1 {z-H(s)}
\end{align}
it follows that $P \in C^{1}([0,T), {\mathcal L}(\cH_{\rm f}))$ and
\begin{align*}
  \dot{P}(s) = \frac{1}{2\pi i} \oint_{{\cal C}(s)} \frac 1 {z-H(s)} \dot{H}(s) \frac 1 {z-H(s)} \, dz.
\end{align*}
Using the bounds \eqref{eq:alphal}, \eqref{eq:bound0} and \eqref{eq:bound_resolvent}, it follows that
\begin{align*}
  \|H(s) \dot{P}(s)\| \le 
  \frac{(2M+3)(2M+4)\alpha_1}{\pi g(s)^2}
  \end{align*}
which proves \eqref{eq:bound_HP} for $l=1$.
The general case for $l > 1$ follows from repeated application of the
chain rule to~\eqref{eq:P_contour} and \eqref{eq:diff_resolvent}, and the bounds \eqref{eq:alphal}, \eqref{eq:bound0} and \eqref{eq:bound_resolvent}.

\medskip 

The differentiability and bounds on the inverse Liouvillian are treated using the same arguments on
the representation
\begin{align*}
  L^{+}(s)A=  \frac{1}{2\pi i} \oint_{\mathcal C(s)} \frac 1 {z-H(s)} [P(s),A]  \frac 1 {z-H(s)} \, dz.
\end{align*}

\medskip 

Let $(u_{n}^{0})_{n = 1, \dots, \mathcal N}$ be an orthonormal basis of $P(0)$.
Then the solutions to the parallel transport equation $\dot{u}_{n}(s) = \dot{P}(s) u_{n}(s)$ with $u_{n}(0) =
u_{n}^{0}$ are easily checked to be a $C^{n}$ orthogonal basis of
$\Ran P(s)$.
It follows that one has 
\begin{align*}
  P^{(l)}(s) = \sum_{n=1}^{\mathcal N} \sum_{m=0}^l {l \choose m} | u_{n}^{(m)}(s)\rangle\langle  u_{n}^{(l-m)}(s)|.
\end{align*}
Therefore,  $P^{(l)}(s)$ is of rank at most $(l+1) \mathcal N$. From the integral representation of
$L^{+}(s)$ (see~\eqref{eq:inv_Liouvillian}), it follows that, for any bounded operator $A$, $L^{+}(s)A$
is of rank at most $2 \mathcal N$. Its derivatives are sums of terms
which all contain as a factor $P(s)$ or one of its derivative, and the
result follows with 
$\dps c_l=2 \sum_{k_1+k_2+k_3=l, \, k_j \in \N} (k_2+1)$.
\end{proof}

\subsection{Adiabatic theory}
\label{subsec:adiabatic}
The following proposition is an adaptation in our context of the
classical adiabatic theorem that the Schrödinger evolution with a
slowly evolving Hamiltonian $H(\varepsilon t)$ approximately preserves gapped
eigenspaces \cite{teufel2003adiabatic}. We explicitly compute the
corrections to first order in $\varepsilon$.
\begin{proposition}
  \label{prop:adiab2}
  Assume the same hypotheses as in Proposition \ref{prop:adiab1}. Let $(U^\eps(t,t'))_{0 \le t' \le t < \eps^{-1}T}$ be the propagator associated with the family of time-scaled Hamiltonians $(H(\eps t))_{t \in [0,\eps^{-1}T)}$, i.e. 
  \begin{align} \label{eq:Uepstt'}
    i   \frac{\partial U^\eps}{\partial t}(t,t') = H(\eps t) U^\eps(t,t'), \quad t \in [t',\eps^{-1}T) \qquad 
    U^\eps(t',t') = {\rm Id},
  \end{align}
and $U^\eps(t)=U^\eps(t,0)$. For all $\eps \ge 0$ and $t \in [0,\eps^{-1}T)$, it holds
  \begin{align} \label{eq:adiab1}
    U^\eps(t) P(0) U^\eps(t)^{*} &= P(\eps t)+i  \eps L^{+}(\eps t) \dot P(\eps t)-i \eps U^\eps(t) \left(L^{+}(0) \dot P(0)\right) U^\eps(t)^{*} + R^\eps(t),
  \end{align}
  with  
  \begin{align} \label{eq:adiab2}
    R^\eps(t)
 = -i \eps^2 \int_0^t U^\eps(t,t') \left. \frac{d}{ds} \left(L(s)^{-1}\dot P(s)\right)\right|_{s=\eps t'} U^\eps(t,t')^* \, dt'.  
  \end{align}
In addition, we have the following estimates:
\begin{align}
\forall 0 \le t' \le t < \eps^{-1} T, \quad &\left\| H(\eps t) U^\eps(t,t') H(\eps t')^{-1} \right\|_{\cL(\cH_{\rm f})} \le e^{\alpha_1 \eps (t-t')},\label{eq:boundHUH-1}\\
&\left\| H(\eps t)^{1/2} U^\eps(t,t') H(\eps t')^{-1/2} \right\|_{\cL(\cH_{\rm f})} \le e^{\alpha_1 \eps (t-t')/2}.  \label{eq:boundHUH-1/2}
\end{align}
\end{proposition}
\begin{proof}
    The existence and uniqueness of the strongly-continuous unitary propagator $(U^\eps(t,t'))$ satisfying \eqref{eq:Uepstt'} can be obtained using  \eqref{eq:alphal} for $l=1$, and Theorem X.70 and the arguments in the proof of Theorem X.71 in~\cite{ReedSimon2}.
    We pass to the interaction picture defined by $H(\eps t)$ and compute the
    evolution of a $C^1$ time-dependent Hilbert-Schmidt observable $A^\eps(t)$ in that picture:
\begin{equation}\label{eq:obs}
  \frac{d}{dt} \left( U^\eps(t)^{*} A^\eps(t) U^\eps(t) \right)=  
  U^\eps(t)^{*} \left(\dot A^\eps(t) +i [H(\eps t),A^\eps(t)]\right) U^\eps(t).
\end{equation}
We first apply \eqref{eq:obs} to $A^\eps(t) = P(\eps t)$ and obtain
\begin{equation}\label{eq:dUepsdt}
\frac{d}{dt} \left( U^\eps(t)^{*} P(\eps t) U^\eps(t) \right) = 
\eps\ U^\eps(t)^{*} \dot P(\eps t) U^\eps(t).
\end{equation}
Estimating this to be of size $\eps$ is not enough because we look at long time scales. What allows us to proceed further is that this quantity is oscillating on
a timescale of order $O(1)$. Indeed, applying \eqref{eq:obs} to $A^\eps(t) = L^{+}(\eps t) \dot P(\eps t)$, for which
$[H(\eps t),A^\eps(t)] = \dot P(\eps t)$, we obtain
\begin{align*}
 U^\eps(t)^{*} \dot P(\eps t) U^\eps(t) =-i  \frac{d}{dt} \left( U^\eps(t)^{*} (L^{+}(\eps t) \dot P(\eps t)) U^\eps(t) \right) +i U^\eps(t)^{*} \frac{d}{dt} \left(L^{+}(\eps t) \dot P(\eps t)\right) U^\eps (t).
\end{align*}
Integrating \eqref{eq:dUepsdt} over $[0,t]$ and using the above equality leads to
\begin{align*}
  U^\eps(t)^{*} P(\eps t) U^\eps(t) &= 
  P(0) + \eps \int_0^t U^\eps(t')^* \dot P(\eps t') U^\eps(t') \, dt'\\
  &=
  P(0) -i \eps U^\eps(t)^{*} \left(L^{+}(\eps t) \dot P(\eps t)\right) U^\eps(t) +i \eps L^{+}(0) \dot P(0) + r^\eps(t) 
\end{align*}
with 
\[
r^\eps(t) = i \eps \int_0^t U^\eps(t')^*  \frac{d}{dt'} \left(L^{+}(\eps t')\dot P(\eps t')\right) \ U^\eps (t') \ dt' = i \eps^2 \int_0^t U^\eps(t')^*  \left. \frac{d}{ds} \left(L^{+}(s)\dot P(s)\right)\right|_{s=\eps t'} \ U^\eps (t') \ dt' .
\]
This implies
\begin{align} \label{eq:UP0U*}
U^\eps(t) P(0) U^\eps(t)^{*} &= P(\eps t) +i \eps L^{+}(\eps t) \dot P(\eps t) -i \eps U^\eps(t) \left(L^{+}(0) \dot P(0)\right) U^\eps(t)^{*} + R^\eps(t),
\end{align}
with
\[
R^\eps(t) = -i \eps^2 \int_0^t U^\eps(t,t')  \left. \frac{d}{ds} \left(L^{+}(s)\dot P(s)\right)\right|_{s=\eps t'}\  U^\eps(t,t')^*  \ dt',
\]
which establishes \eqref{eq:adiab1}.

\medskip 

Let us now prove \eqref{eq:boundHUH-1}. Let $\psi \in D$. For all $t \in [t',\eps^{-1}T)$, we set $\psi_\eps  (t) = U^\eps(t,t') \psi$ and $\phi_\eps(t)=H(\eps t) \psi_\eps(t)$. We have 
$$
i \frac{d\phi_\eps}{dt} (t) = i \frac{d}{dt} \left( H(\eps t) \psi_\eps(t)  \right) =  H(\eps t) \phi_\eps(t) + i \eps \dot H(\eps t) H(\eps t)^{-1} \phi_\eps(t),
$$
from which we obtain
$$
\phi_\eps(t) = U^\eps(t,t') H(\eps t') \psi + i \eps \int_{t'}^t \dot H(\eps t'') H(\eps t'')^{-1} \phi_\eps(t'') \, dt'',
$$
and finally
$$
\|\phi_\eps(t)\|_{\cH_{\rm f}} \le \|H(\eps t') \psi\|_{\cH_{\rm f}} + \alpha_1 \eps \int_{t'}^t \|\phi_\eps(t'')\|_{\cH_{\rm f}} \, dt''.
$$
By the Gr\"onwall lemma,
$$
\| H(\eps t) U^\eps(t,t') \psi \|_{\cH_{\rm f}} =  \|\phi_\eps(t)\|_{\cH_{\rm f}} \le \|H(\eps t') \psi\|_{\cH_{\rm f}} e^{\alpha_1 \eps (t-t')}.
$$
Applying this inequality to $\psi=H(\eps t')^{-1}\phi$ for all $\phi
\in \cH_{\rm f}$ gives \eqref{eq:boundHUH-1}. We obtain \eqref{eq:boundHUH-1/2} by interpolation (see e.g. \cite[Section IX.4, Proposition 9]{ReedSimon2}).
\end{proof}

The third term
$$
-U^\eps(t) \left(i L^{+}(0) \dot P(0)\right) U^\eps(t)^{*}
$$
of the right-hand side of \eqref{eq:adiab1} is oscillatory, and can be written as the derivative of a bounded function up to higher order terms. Its time-average therefore becomes negligible in the considered regimes. Let us introduce the space 
$$
{\cal L}^{\rm OD}(s):= \left\{ A \in \mathcal L(\cH_{\rm f}) \; | \; P(s)AP(s)=(1-P(s))A(1-P(s))=0 \right\}
$$
of bounded off-diagonal operators relatively to the decomposition $\cH_{\rm f}=\mbox{Ran}(P(s))\oplus \mbox{Ker}(P(s))$.

\begin{lemma} \label{lem:oscillatory}
Under the assumptions of Propositions~\ref{prop:adiab1} and~\ref{prop:adiab2}, we have for any self-adjoint operator
$A \in {\cal L}^{\rm OD}(0)$, 
\begin{equation}
U^\eps(t) A U^\eps(t)^{*} = \frac{d}{dt} \left( i L^+(\eps t) \left( U^\eps(t)AU^\eps(t)^*  \right) \right) + R_A^\eps(t),
\end{equation}
where 
$$
R_A^\eps(t)=2 U^\eps(t) r^\eps(t) A r^\eps(t)  U^\eps(t)^{*} - \Big( U^\eps(t) (1-2P(0)) A r^\eps(t)  U^\eps(t)^{*} + {\rm h.c.} \Big) + \eps i \frac{dL^+}{ds}(\eps t)\Big( U^\eps(t)AU^\eps(t)^*\Big)
$$
and
$$
r^\eps(t)=-i\eps  U^\eps(t)^* \left( L^+(\eps t)  \dot P(\eps t) \right)  U^\eps(t) + i \eps L^+(0)  \dot P(0) + U^\eps(t)^* R^\eps(t)  U^\eps(t).
$$
\end{lemma}

\begin{proof} We have
\begin{align*}
\frac{d}{dt} \left(  i L^+(\eps t) \left(U^\eps(t)AU^\eps(t)^*\right)  \right) &= \eps i \frac{dL^+}{ds}(\eps t) \left( U^\eps(t)AU^\eps(t)^*\right)
+ L^+(\eps t) L(\eps t) \left( U^\eps(t) AU^\eps(t)^*\right)  \\
&= \eps i \frac{dL^+}{ds}(\eps t) U^\eps(t)AU^\eps(t)^* + P(\eps t) U^\eps(t) AU^\eps(t)^* (1-P(\eps t)) + {\rm h.c.},
\end{align*}
and we deduce from \eqref{eq:adiab1} that $P(\eps t) U^\eps(t)=U^\eps(t) \left( P(0)+r^\eps(t) \right)$. We therefore have
\begin{align*}
P(\eps t) U^\eps(t) AU^\eps(t)^* (1-P(\eps t)) + {\rm h.c.} =& \, U^\eps(t) \left( P(0)+r^\eps(t) \right) A (1-P(0)-r^\eps(t) ) U^\eps(t)^*  + {\rm h.c.} \\
=& \,U^\eps(t)  A  U^\eps(t)^*  + \left( U^\eps(t) \left( 1- 2P(0) \right) A r^\eps(t)  U^\eps(t)^*  + {\rm h.c.} \right) \\ & + 2 U^\eps(t) r^\eps(t) A r^\eps(t)  U^\eps(t)^*,
\end{align*}
where we have used that $A=P(0)A(1-P(0))+(1-P(0))AP(0)$.
\end{proof}

\subsection{Application to coherent transport in Bloch representation}
\label{subsec:appli}



Let $H$ be the periodic magnetic Hamiltonian defined in~\eqref{eq:periodic_Hamiltonian}, $J$ the current operator whose components are defined in~\eqref{def:Jalpha}, $\mu_{\rm F}$ the Fermi level,
\begin{align*}
 \mu := 1 + \min \sigma(H) \quad  \mbox{and} \quad   \eta = \max_{|\vec e|\le |e_{\alpha}|,|e_{\beta}|} \| (J\cdot e) (H+\mu)^{-1}\|_{\cL(L^2(\R^d;\C))} < \infty.
\end{align*}

Let $k \in \R^{d}$.
Assume that $\lambda_{N_k+1,k}-\lambda_{N_k,k} > 0$ and set
\begin{align*}
s_k &= \inf \{s> 0 \, | \, g_k(s) =0 \} \quad \mbox{where} \quad
g_k(s):=\min(1,\lambda_{N_k+1,k-s \vec e_\beta}-\lambda_{N_k,k-s \vec
  e_\beta}).
\end{align*}
We consider the family of Hamiltonians
\begin{align}
  H(s) := H_{k - e_{\beta} s} + \mu.
\end{align}
We have
\begin{align}
\dot{H}(s) &= -\partial_{\beta}H_{k-s\vec e_\beta} = -((-i \nabla + k  + \mathcal A - s e_{\beta}) \cdot e_{\beta})  = J_{\beta,k-s\vec e_\beta},\\
\ddot{H}(s) &= |e_{\beta}|^{2}{\rm Id}_{L^2_{\rm per}},
\end{align}
and so hypotheses {\rm H1-H3} of Proposition~\ref{prop:adiab1} are
satisfied with ${\cal H}_{\rm f}=L^2_{\rm per}$, $D=H^2_{\rm per}$,
$T=s_k$, $n$ arbitrarily large, $\alpha_1 \le \eta$,
$\alpha_2=|e_{\beta}|^{2}$, $\alpha_l=0$ for $l \ge 3$,
$a_-(s)=\min\sigma(H)+\mu$,
$a_+(s)=\lambda_{N_k,k-s \vec e_\beta}+\mu$,
$M=\max_{k' \in \BZ} \lambda_{N_{k'}+1,k'}+\mu$, $g(s)=g_k(s)$, and
${\cal N}=N_k$.
\begin{corollary}
  \label{cor:adiab} 
  Let $k \in \R^d$ such that $\lambda_{N_k+1,k}-\lambda_{N_k,k} > 0$.
  Then, for all $\eps > 0$ and $t \in [0,\eps^{-1}s_k)$, the operator
  $\partial_{\alpha} {H}_{k-\eps e_\beta
    t}{\gamma}^{\eps}_{\beta,k-\eps \vec{e}_{\beta}t}(t)$ is in
  $\Schatten_{1,\rm per}$, and we have
  \begin{align}
  {\rm Tr}(\partial_{\alpha} {H}_{k-\eps e_\beta t} \widetilde{\mathcal U}_{\beta,k}^{\varepsilon}(t)\gamma_k(0) \widetilde{\mathcal U}_{\beta,k}^{\varepsilon}(t)^{*}) 
=& 
  \partial_{\alpha} \left( \Tr(H_{k{-}\eps\vec{e}_\beta t} P_{N_k,k{-}\eps \vec{e}_{\beta}t})\right)  \nonumber \\
  &
  {+i}\eps
  \Tr(\partial_{\alpha} H_{k{-}\eps\vec{e}_\beta t} 
  L_{N_k,k{-}\varepsilon \vec{e}_{\beta}t}^{+}\partial_{\beta}
  P_{N_k,k{-}\eps \vec{e}_{\beta}t})  \nonumber \\
  & {-i} \eps \Tr(\partial_{\alpha} H_{k{-}\eps \vec{e}_{\beta} t} 
  \widetilde{\mathcal U}^\eps_{\beta,k}(t)\ L_{N_k,k}^{+}\partial_{\beta} P_{N_k,k
}\, \widetilde {\mathcal U}^\eps_{\beta,k}(t)^*)  \nonumber \\
  & + {\mathfrak R}^{\eps}_k(t), \label{eq:estim1}
\end{align}
where each term of the right-hand side is a well-defined real number and $L_{N,k}^{+}$ is a shorthand notation for the inverse
Liouvillian $L_{H_k,[\lambda_{1,k},\lambda_{N,k}]}^{+}$. In addition, we have the following bounds
\begin{align}
&|{\rm Tr}(\partial_{\alpha} {H}_{k-\eps e_\beta t} \widetilde{\mathcal U}_{\beta,k}^{\varepsilon}(t)\gamma_k(0) \widetilde{\mathcal U}_{\beta,k}^{\varepsilon}(t)^{*}) | \le C e^{\eta\eps t} ,\label{eq:first_bound}\\ 
&|\partial_{\alpha} \left( \Tr(H_{k{-}\eps\vec{e}_\beta t} P_{N_k,k{-}\eps \vec{e}_{\beta}t})\right)| \le C,  \nonumber \\
& | \eps
  \Tr(\partial_{\alpha} H_{k{-}\eps\vec{e}_\beta t} 
  L_{N_k,k{-}\varepsilon \vec{e}_{\beta}t}^{+}\partial_{\beta}
  P_{N_k,k{-}\eps \vec{e}_{\beta}t}) | \le C \frac{\eps}{g_k(\eps t)^4}, \nonumber  \\ 
  \label{68}
  & | \eps \Tr(\partial_{\alpha} H_{k-\eps \vec{e}_{\beta} t} 
  \widetilde{\mathcal U}^\eps_{\beta,k}(t)\ L_{N_k,k}^+\partial_{\beta} P_{N_k,k}\, \widetilde {\mathcal U}^\eps_{\beta,k}(t)^*) | \le  C \frac{\eps e^{\eta\eps t}}{g_k(0)^4},   \\
& | {\mathfrak R}^{\eps}_k(t)| \le \frac{C \varepsilon^{2} t e^{\eta \varepsilon t}}{\min_{s \in [0, \varepsilon t]} g_{k}(s)^{6}}, \label{eq:last_bound}
 \end{align}
 for a constant $C \in \R_+$ independent of $k$, $\eps$ and $t$.
\end{corollary}

\begin{proof}
Applying the second assertion in Proposition~\ref{prop:adiab2}, we get 
\begin{align}
  \widetilde{\mathcal U}_{\beta,k}^{\varepsilon}(t)\gamma_k(0) \widetilde{\mathcal U}_{\beta,k}^{\varepsilon}(t)^{*}
=  & P_{N_k,k{-}\eps \vec{e}_{\beta}t}
  +i\eps 
  L_{N_k,{k-\varepsilon \vec{e}_{\beta}t}}^{+}\partial_{\beta}
  P_{N_k,k-\eps \vec{e}_{\beta}t} \nonumber \\
&  -i \eps 
\widetilde{\mathcal U}^\eps_{\beta,k}(t)\  \left( L_{N_k,{k}}^{+}\partial_{\beta} P_{N_k,k
} \right) \, \widetilde {\mathcal U}^\eps_{\beta,k}(t)^*
  + R^\eps_k(t). \label{eq:gamma_eps_t}
\end{align}
Each term $A$ in \eqref{eq:gamma_eps_t} being a finite-rank self-adjoint operator, it holds
\begin{align*}
  \|\partial_{\alpha} {H}_{k-\eps e_\beta t} A\|_{\Schatten_{1}} \le \mbox{Rank}(A) \|\partial_{\alpha} {H}_{k-\eps e_\beta t} A\| \le \eta \mbox{Rank}(A) \|(H_{k-\varepsilon e_{\beta} t} + \mu) A\|,
\end{align*}
and again by Proposition \ref{prop:adiab2} we get
$$\displaylines{ 
   {\rm Tr}(\partial_{\alpha} {H}_{k-\eps e_\beta t} \widetilde{\mathcal U}_{\beta,k}^{\varepsilon}(t)\gamma_k(0) \widetilde{\mathcal U}_{\beta,k}^{\varepsilon}(t)^{*}) 
   = \Tr(\partial_{\alpha} H_{k{-}\eps\vec{e}_\beta t} P_{N_k,k{-}\eps \vec{e}_{\beta}t}) 
  +i\eps\Tr(\partial_{\alpha} H_{k{-}\eps\vec{e}_\beta t} 
  L_{N_k,k-\varepsilon \vec{e}_{\beta}t}^{+}\partial_{\beta}
  P_{N_k,k{-}\eps \vec{e}_{\beta}t}) \cr
  -i \eps \Tr(\partial_{\alpha} H_{k{-}\eps \vec{e}_{\beta} t} 
\widetilde{\mathcal U}^\eps_{\beta,k}(t)\ \left(  L_{N_k,k}^{+}\partial_{\beta} P_{N_k,k
} \right) \, \widetilde {\mathcal U}^\eps_{\beta,k}(t)^*)  + {\rm Tr} \left( \partial_{\alpha} {H}_{k-\eps e_\beta t}{R}^\eps_k(t)  \right)}
$$
with
\begin{align}
  \label{eq:defRepsk}
  R^\eps_{k}(t) = i \eps^2 \int_0^t \widetilde {\mathcal U}^\eps_{\beta,k}(t,t')  \partial_{k_{\beta}}\left(L_{N_k,k-\varepsilon e_{\beta} t'}^{+}\partial_{k_{\beta}} P_{N_k,k-\varepsilon e_{\beta} t'}\right) \widetilde {\mathcal U}^\eps_{\beta,k}(t,t')^{*}  \ dt'.
\end{align}
It results from the Hellmann-Feynman formula that
  $$
 \Tr(\partial_{\alpha} H_{k{-}\eps\vec{e}_\beta t} P_{N,k{-}\eps \vec{e}_{\beta}t})   =   \partial_{\alpha} \left( \Tr(H_{k{-}\eps\vec{e}_\beta t} P_{N,k{-}\eps \vec{e}_{\beta}t})\right).  
  $$
Finally, using Propositions \ref{prop:adiab1} and \ref{prop:adiab2}, we obtain the bounds \eqref{eq:first_bound}-\eqref{eq:last_bound}. In particular,
\begin{align*}
|{\mathfrak R}^\eps_k(t)| &=\Big| {\rm Tr} \left( \partial_{\alpha}
  {H}_{k-\eps e_\beta t}{R}^\eps_k(t)  \right)\Big| \\
&\le 4N_k \varepsilon^{2} \eta t \sup_{t' \in [0,t)} \left( e^{\eta  \varepsilon(t-t')} \left\|H_{k-\varepsilon e_{\beta} t'} \partial_{k_{\beta}}\left(L_{N_k,k-\varepsilon e_{\beta} t'}^{+}\partial_{k_{\beta}} P_{N_k,k-\varepsilon e_{\beta} t'}\right) \right\| \right) \\
&\le C  \frac{\eta \varepsilon^{2} t e^{\eta \varepsilon t}}{\inf_{s \in [0, \varepsilon t)} g_{k}(s)^{6}},
\end{align*}
where $C \in \R_+$ is independent of $k$, $\eps$ and $t$.
\end{proof}

\begin{remark}
  \label{rem:threeterms}The decomposition \eqref{eq:estim1} will be key to
  computing the current in insulators, non-degenerate metals and
  semimetals. The first three terms in the right-hand side of
  \eqref{eq:estim1} have different physical meanings. The first term
  is the adiabatic term: electrons simply are transported
  adiabatically across the Brillouin zone. This term will be
  responsible for the ballistic transport of electrons in metals. The
  second is the first-order static response, and will be the cause of
  the Hall conductivity in insulators. The third is oscillatory, and
  is related to the AC response of solids (not treated here). This
  decomposition only makes sense for a non-zero gap; in particular, it
  cannot be used to compute the contribution to the current for $k$
  points close to Dirac points for semimetals.
\end{remark}
\subsection{Linear response}
\label{subsec:lr}

We now aim at obtaining an expansion of the current to first order in
$\varepsilon$ for a given $t$, based on a Dyson expansion instead of
the adiabatic theorem. This is a classical computation in response
theory, sometimes known as the Kubo formula \cite{Kubo1957}. In contrast to the previous result, this gives
a remainder that does not depend on a gap, and will therefore be
useful for the study of semimetals near Dirac points.
\begin{proposition}\label{prop:Dyson}  Let $H$ be the periodic magnetic Hamiltonian defined in~\eqref{eq:periodic_Hamiltonian}. Under the additional assumptions that $V \in H^1_{\rm per}$ and
${\cal A} \in (H^2_{\rm per})^d$, there exists a constant $C \in \R_+$ such that for all $k \in \R^d$ such that $\lambda_{N_k+1,k}-\lambda_{N_k,k} > 0$, we have for all $\eps,t \in \R_+$,
\begin{align}
 {\rm Tr} \left(\partial_{\alpha} H_{k - \varepsilon e_{\beta} t}\widetilde{\mathcal U}^\eps_{\beta,k}(t)\, \gamma_{k}(0)\, \widetilde{\mathcal U}^\eps_{\beta,k}(t)^*\right) =& \partial_{\alpha} {\rm Tr} \left(H_k \gamma_{k}(0)\right) - \varepsilon t  {\partial_{\alpha}\partial_ {\beta}} \left( \Tr (H_{k} \gamma_{k}(0)) \right) \nonumber  \\ & + i \varepsilon \Tr(\partial_{\alpha} H_{k} (e^{-i t L_k} - 1) L_{k}^{+}  \partial_{\beta} \gamma_{k}(0)) + \rho^\eps_k(t),
\end{align}
with, when $\eps t \le 1$,
\begin{equation}\label{eq:boundrho}
|\rho^\eps_k(t)| \le C \eps^2 t^3 (1+ t^3 ).
\end{equation}
\end{proposition}

\begin{proof}[Proof]  Let $k$ be such that $\lambda_{N_k+1,k}-\lambda_{N_k,k} > 0$. Since $k' \mapsto {\rm Tr} \left(H_{k'} \gamma_{k'}(0)\right)$ is real-analytic in a neighborhood of $k$, we have by Hellmann-Feynman theorem
\begin{align*}
 \partial_{\alpha} {\rm Tr} \left(H_k \gamma_{k}(0)\right) =  {\rm Tr} \left(\partial_{\alpha} H_k \gamma_{k}(0)\right) \quad \mbox{and} \quad
  {\partial_{\alpha}\partial_ {\beta}} {\rm Tr} \left(H_k \gamma_{k}(0)\right) =  {\rm Tr}  
\left(  {\partial_{\alpha}\partial_ {\beta}} 
 H_k \gamma_{k}(0)\right) +  {\rm Tr} \left(\partial_{\alpha} H_k \partial_{\beta} \gamma_{k}(0)\right).
\end{align*}
We also have $\partial_{\alpha}H_{k-\varepsilon e_{\beta} t} =
\partial_{\alpha}H_k - \eps t e_{\alpha} \cdot e_{\beta}$. It follows that
\begin{align*}
\rho^\eps_k(t)=&{\rm Tr} \left(\partial_{\alpha} H_{k}\widetilde{\mathcal U}^\eps_{\beta,k}(t)\, \gamma_{k}(0)\, \widetilde{\mathcal U}^\eps_{\beta,k}(t)^*\right) - {\rm Tr} \left(\partial_{\alpha} H_k \gamma_{k}(0)\right) + \eps t  {\rm Tr} \left(\partial_{\alpha} H_k \partial_{\beta} \gamma_{k}(0)\right) \\ & - i \varepsilon \Tr(\partial_{\alpha} H_{k} (e^{-i t L_k} - 1) L_{k}^{+}  \partial_{\beta} \gamma_{k}(0)).
\end{align*}
We now expand the first term in the right-hand side of this equation. We set $\mu := 1 + \min \sigma(H)$, 
$$H(s):=H_{k-se_\beta}+\mu,\quad A=\partial_{\alpha}H_k, \quad
I_0=[0,\frac 12(\lambda_{N_k,k}+\lambda_{N_k+1,k})+\mu],\quad P(s)=\1_{I_0}(H(s)).$$
It holds
$$H(s)=h_0+sh_1+\frac{s^2 |e_{\beta}|^{2}}2$$
 with $h_0=H_k+\mu$ and $h_1=J_{\beta,k}=-\partial_{\beta}H_k$. The operators $h_0$, $h_1$ and $A$ are self-adjoint on $L^2_{\rm per}$ and we have $h_0 \ge 1$ and $h_1h_0^{-1/2}$ and $Ah_0^{-1/2}$ bounded.
  Besides, $P(s)= \gamma_{k-se_\beta}(0)$, so that $\dot{P}(0)=- \partial_{\beta} \gamma_{k}(0)$.  Let $(U^\eps(t,t'))_{t,t' \in \R}$ be the propagator associated with the family $(H(\eps t))_{t \in \R}$ and $U^\eps(t):=U^\eps(t,0)$. 
We have $U^\eps(t,t')=e^{-i\mu(t-t')} {\cal U}^\eps_k(t,t')$ and $U^\eps(t)=e^{-i\mu t} {\cal U}^\eps_k(t)$.  
With these notations, we have
$$
\rho^\eps_k(t)={\rm Tr} \left(A U^\eps(t)\, P(0) U^\eps(t)^*\right) - {\rm Tr} \left(AP(0) \right) - \eps t  {\rm Tr} \left(A  \dot{P}(0) \right) + i \varepsilon \Tr(A (e^{-i t L_0} - 1) L_0^{+}  \dot{P}(0)),
$$
where $L_0=L_{h_0,I_0}$ and $L_0^+=L_{h_0,I_0}^+$, and we focus on expanding the operator $U^\eps(t)\, P(0) U^\eps(t)^*$ close to $t=0$. 

\begin{lemma}\label{lem:Dyson+}
We have 
\begin{equation}\label{eq:Dyson+}
U^\eps(t)\, P(0) U^\eps(t)^* = P(0) +  \eps \left(t \dot{P}(0) - i \left( e^{-itL_0}-1 \right) (L_0^+\dot{P}(0))\right) + \Pi^\eps_2(t)
\end{equation}
with $\rho^\eps_k(t) = {\rm Tr} \left(A \Pi_2^\eps(t) \right)$.
Moreover, we have the bound~\eqref{eq:boundrho}
\begin{align*}
  |\rho^\eps_k(t)| \le C \eps^2 t^3 (1+ t^3 ).
\end{align*}

\end{lemma} 

Lemma~\ref{lem:Dyson+} closes the proof of Proposition~\ref{prop:Dyson}.
\end{proof}

\begin{proof}[Proof of Lemma~\ref{lem:Dyson+}]
We deduce from the Dyson expansion that
$$
U^\eps(t)=U^0(t)+V^\eps(t)+W^\eps(t),
$$
where {$U^0(t)=e^{-ith_0}$} and 
\begin{align*}
V^\eps(t) &= - i\eps \int_0^t U^0(t-t')  t' h_1 U^0(t') \, dt', \\
W^\eps(t) &=\eps^2 \left( -i \frac{t^3}6 U^0(t)+ \int_0^t \left( \int_0^{t'} U^\eps(t,t') t' (h_1+\eps t'/2) U^0(t'-t'') t'' (h_1+\eps t''/2) U^0(t'') \, dt'' \right) \, dt' \right).
\end{align*}
This induces $U^\eps(t)\, P(0) U^\eps(t)^* = P(0) + \Pi^\eps_1(t)+\Pi^\eps_2(t)$
where
\begin{align*}
\Pi^\eps_1(t)&=V^\eps(t)P(0)U^0(t)^* + {\rm h.c.} = - i\eps \int_0^t t' U^0(t-t')  [h_1,P(0)] U^0(t-t')^* \, dt', \\
\Pi^\eps_2(t)&=V^\eps(t)P(0)V^\eps(t)^* + \left( W^\eps(t) P(0) (U^0(t)+V^\eps(t))^* + {\rm h.c.} \right) + W^\eps(t)P(0)W^\eps(t)^*.
\end{align*}

\medskip 

We first analyze $\Pi^\eps_1(t)$ by computing
\begin{align*}
  U^0(t-t') [h_1,P(0)] U^0(t-t')^* &= - e^{-i(t-t') L_{0}} L_{0} \dot P(0)\\
  &= i\frac{d}{dt'} e^{-i(t-t') L_{0}} \dot P(0)\\
  &= \frac{d^{2}}{dt'^{2}} e^{-i(t-t') L_{0}} L_{0}^{+}\dot P(0),
\end{align*}
where we have used $\dot P(0) = L_{0}^{+}[P_{0},h_{1}]$ and
$\dot P(0) = L_{0}^{+} L_{0} \dot P(0)$. Using integration by parts, we obtain 
\begin{align*}
  \Pi^{\varepsilon}_{1}(t) = \eps \left(t \dot{P}(0) - i \left( e^{-itL_0}-1 \right) (L_0^+\dot{P}(0))\right)
\end{align*}
and \eqref{eq:Dyson+} follows.

We now work on the bound~\eqref{eq:boundrho}. For that purpose, we introduce the following
quantities, which are independent of $k, \varepsilon$ and $t$:
\begin{align*}
\nu_0 &= \max_{|\vec e|\le |e_{\alpha}|,|e_{\beta}|} \| (J\cdot e) (H+\mu)^{-1/2}\|_{\cL(L^2(\R^d;\C))}, \\
\nu_1&=  \max_{|\vec e|\le |e_{\alpha}|,|e_{\beta}|} \|(H+\mu)^{1/2} (J\cdot e) (H+\mu)^{-1} \|_{\cL(L^2(\R^d;\C))}, \\
\nu_2&=  \max_{|\vec e|\le |e_{\alpha}|,|e_{\beta}|} \|(H+\mu) (J\cdot e) (H+\mu)^{-2}\|_{\cL(L^2(\R^d;\C))}, 
\\ \lambda&=\max_{k \in \R^d,|k-k'| \le |e_\beta|}\|(H_k+\mu)^{1/2}(H_{k'}+\mu)^{-1/2}\|_{\cL(L^2_{\rm per})}.
\end{align*}
Note that the assumptions ${\cal A} \in (L^4_{\rm per})^d$, $\nabla \cdot {\cal A}=0$, and $V \in L^2_{\rm per}$ are sufficient to ensure that the quantities $\nu_0$, $\nu_1$ and $\lambda$ are finite. Besides, since
$\|h_0h_1h_0^{-2}\| \le \|(H+\mu)J_\beta (H+\mu)^{-2}\|$ and 
\begin{align*}
(H+\mu)J_\beta (H+\mu)^{-2} = J_\beta (H+\mu)^{-1}  - 2i
\sum_{\alpha=1}^d (\partial_\alpha {\cal A}_\beta-\partial_\beta {\cal
  A}_\alpha) J_\alpha (H+\mu)^{-2} &- (\Delta {\cal A}_\beta)
(H+\mu)^{-2} \\
&+ i\partial_{\beta}V(H+\mu)^{-2},
\end{align*}
we deduce from the assumptions ${\cal A} \in (H^2_{\rm per})^d$ and $V \in H^1_{\rm per}$ that $ \|(H+\mu)J_\beta (H+\mu)^{-2}\| < \infty$, hence that $\nu_2 < \infty$.

\medskip 

We now aim at controlling $\rho_k(t)$ thanks to $\nu_0$, $\nu_1$, $\nu_2$ and $\lambda$. 
Using the relations 
$P(0)=P(0)^2$ and $P(0)=h_0^{-m} h_0^mP(0)$
 with 
 $$\|h_0^mP(0)\|\le (\mu_{\rm F}+\mu)^m, \;\;\|h_0^{-1}\| \le \|h_0^{-1/2}\| \le 1\;\;
 \mbox{ and}\;\; \mbox{Rank}(P(0))=N_k \le \overline{N}:=\max_{k'} N_{k'},$$
  we deduce
\begin{align*}
|\rho^\eps_k(t)|& \le  \overline{N}  \biggl( (\mu_{\rm F}+\mu)^2 \|AV^\eps(t)h_0^{-1}\| \, \|V^\eps(t)h_0^{-1}\| +  2 (\mu_{\rm F}+\mu)^3 \|AW^\eps(t)h_0^{-2}\| \,  (\|h_0^{-1}\|+\|V^\eps(t)h_0^{-1}\| )
\\ & \qquad\qquad\qquad\qquad +  (\mu_{\rm F}+\mu)^3  \|AW^\eps(t)h_0^{-2}\| \|W^\eps(t)h_0^{-1}\|   \biggr) \\
& \le  \overline{N} \nu_0 \biggl(   (\mu_{\rm F}+\mu)^2  \|h_0^{1/2}V^\eps(t)h_0^{-1}\|^2 +  (\mu_{\rm F}+\mu)^3  \|h_0^{1/2}W^\eps(t)h_0^{-2}\|  (2+ 2\|V^\eps(t)h_0^{-1}\|\\
&\qquad\qquad\qquad\qquad 
+  \|W^\eps(t)h_0^{-1}\|) \bigg).
\end{align*}
Next, we get
\begin{align*}
\|V^\eps(t)h_0^{-1}\| &\le \frac{\eps t^2}2 \|h_1 h_0^{-1}\| \le \frac{\eps t^2}2 \nu_0, \;\;\;\;
\|h_0^{1/2}V^\eps(t)h_0^{-1}\|  \le \frac{\eps t^2}2 \nu_1, \\
\|W^\eps(t)h_0^{-1}\| & \le \eps^2 t^3 \bigg( \frac{1}6 + \nu_0\nu_1 \frac{t}8 + (\eps t) t \left( \frac{\nu_0}{30}+\frac{\nu_1}{20} \right) + (\eps t)^2 t \frac{1}{72}   \bigg), \\
\|h_0^{1/2}W^\eps(t)h_0^{-2}\| & \le \eps^2 t^3 \lambda e^{\eta \eps t/2} \bigg( \frac{1}6 + \nu_1\nu_2 \frac{t}8 + (\eps t) t \left( \frac {\nu_1}{30}+\frac{\nu_2}{20} \right) + (\eps t)^2 t \frac{1}{72} \bigg).
\end{align*}
 It follows that there exists a constant $C$ depending only on $V$, ${\cal A}$ and $\mu_{\rm F}$, such that
$$
|\rho^\eps_k(t)| \le C \eps^2 t^3 \left( t+e^{\eta \eps t/2} \left(1 + t (1+(\eps t)^2) + (\eps t) t^2 (1+(\eps t)^3) + (\eps t)^4 t^3 (1+(\eps t)^2) \right)  \right),
$$
which leads to \eqref{eq:boundrho} when $\eps t \le 1$.
\end{proof}


\section{Insulators}\label{sec:ins}
In this section and the
following ones, we use the notation $O(f(\varepsilon,t,t',\delta))$ to
denote a term that is bounded in absolute value by
$C f(\varepsilon,t,t',\delta)$, where $C$ is a constant that can depend
on the system under consideration (through ${\cal A}$, $V$, $\mu_{\rm F}$, $e_\alpha$ and $e_\beta$), but not on the parameters
$\varepsilon,t,t',\delta$. We will use the notation $\gamma_{k}^{0}$
for $\gamma_{k}(0)$.

\medskip 

We now prove Theorem \ref{thm:insulator}. 
For insulators, $N_{k}=N_{\rm ins}$ for all $k$, and
$\lambda_{N_{k}+1,k} - \lambda_{N_{k},k}$, hence $g_k$, is
uniformly bounded away from zero and $s_{k} = +\infty$. We use the notation $L_{k}^{+}$ for
$L_{N_{\rm ins}, k}^{+}$. We apply Corollary \ref{cor:adiab} and
obtain by integrating over the Brillouin zone
\begin{align*}
 j_{\alpha,\beta}^\eps(t)
=& 
{-}(2\pi)^{-d}\int_{{\mathcal B} } \partial_{\alpha} \left( \Tr(H_{k-\eps\vec{e}_\beta t} \gamma_{k-\eps \vec{e}_{\beta}t}^{0})\right) dk
{-i} \eps(2\pi)^{-d} \int_{{\mathcal B} } 
  \Tr( \partial_{\alpha} H_{k{-}\eps\vec{e}_\beta t} 
  L^{+}_{{k-\vec{e}_{\beta}t}} \partial_{\beta}
  \gamma^0_{k-\eps\vec{e}_\beta t}) dk \\
  &+i\eps (2\pi)^{-d}\int_{{\mathcal B} } \Tr( \partial_{\alpha} H_{k-\eps \vec{e}_{\beta} t} \;
\widetilde{\mathcal U}^\eps_{\beta,k}(t)\ (L_{k}^{+}  \partial_{\beta} \gamma^0_{k}) \, \widetilde{\mathcal U}^\eps_{\beta,k}(t)^*) dk + O\left(\eps^2 t{{\rm e}^{\eta \eps t}} \right).
\end{align*}

As mentioned in Remark \ref{rem:threeterms}, these three terms are
adiabatic, static and oscillatory respectively.
\begin{itemize}
\item The first term of the right-hand side vanishes for all $t$, as
  the integral of the derivative of the smooth periodic function
  $k \mapsto \Tr(H_{k} \gamma^0_{k})$ on a unit cell.
\item The second term is dealt with using the
  relation
$$L_{k}^{+}( (\partial_{\alpha} H_{k})^{\rm OD})=  
[\gamma^0_k\;,\; \partial_{\alpha} \gamma^0_k],
$$
where
$(\partial_{\alpha} H_{k})^{\rm OD} = \gamma^{0}_{k}
(\partial_{\alpha} H_{k}) (1-\gamma^{0}_{k}) + (1-\gamma^{0}_{k})
(\partial_{\alpha} H_{k}) \gamma^{0}_{k}$. By periodicity, we have
\[
  \int_{\mathcal B} {\rm Tr}\left( \partial_{\alpha} H_{k{-}\eps
      \vec e_\beta t} L_{{k{-}\eps \vec e_\beta t}}^{+}
    \partial_{\beta} \gamma^0 _{k-\eps \vec e_\beta t}\right) dk =
  \int_{\mathcal B} {\rm Tr}\left( \partial_{\alpha} H_{k}
    L_{{k}}^{+} \partial_{\beta} \gamma^0 _{k}\right) dk,\] and we
observe that
\begin{align}
  {\rm Tr}\left( \partial_{\alpha} H_{k} L_{k}^{+}  \partial_{\beta} \gamma^0 _{k}\right) &={\rm Tr}\left(( \partial_{\alpha} H_{k})^{\rm OD}L_{k}^{+} \partial_{\beta} \gamma^0 _{k}\right) \nonumber \\
  &=  {\rm Tr}\left(L_{k}^{+} ( (\partial_{\alpha} H_{k})^{\rm OD})  \partial_{\beta} \gamma^0 _{k}\right) =  {\rm Tr}\left([\gamma^0_{k},  \partial_{\alpha} \gamma^0_{k} ]  \partial_{\beta} \gamma^0_{k}\right), 
\end{align}
so that
\begin{align}
  \nonumber
  \int_{\mathcal B} {\rm Tr}\left( \partial_{\alpha} H_{k-\eps \vec e_\beta t} L_{k-\eps \vec e_\beta t}^{+}   \partial_{\beta} \gamma^0 _{k-\eps \vec e_\beta t} \right) dk &=
  \int_{\mathcal B} {\rm Tr}\left(\gamma^0_{k}[ \partial_{\alpha} \gamma^0_{k} ,  \partial_{\beta} \gamma^0 _{k}]\right) dk .
\end{align}

\item 
  We now focus on the time-average of the oscillating term
$$
\omega^\eps(t):= \frac1 t \int_0^t dt' \, \int_{{\mathcal B} } \Tr(
\partial_{\alpha} H_{k-\eps \vec{e}_{\beta} t'} \;
\widetilde{\mathcal U}^\eps_{\beta,k}(t')\ (i L_{k}^{+}
\partial_{\beta} \gamma^0_{k}) \, \widetilde{\mathcal
  U}^\eps_{\beta,k}(t')^*) \, dk.
$$
In order to bound this term, we apply Lemma~\ref{lem:oscillatory} to
$A=i L_{k}^{+} \partial_{\beta}\gamma^0_{k}$, which is a
self-adjoint off-diagonal operator for the decomposition
$L^2_{\rm per}=\mbox{Ran}(\gamma^0_{k})\oplus
\mbox{Ker}(\gamma^0_{k})$. We thus get
$$
\widetilde{\mathcal U}^\eps_{\beta,k}(t) \left( i L_{k}^{+}
  \partial_{\beta}\gamma^0_{k}\right) \widetilde{\mathcal
  U}^\eps_{\beta,k}(t)^* = \frac{d}{dt} \left( i L^{+}_{k- \varepsilon
    e_{\beta} t}\Big( \widetilde{\mathcal U}^\eps_{\beta,k}(t) (i
  L_{k}^{+} \partial_{\beta}\gamma^0_{k}) \widetilde{\mathcal
    U}^\eps_{\beta,k}(t)^*\Big) \right) + \widetilde R_k^\eps(t),
$$
where
\begin{align*}
  \widetilde  R_k^\eps(t)& =2 \widetilde{\mathcal U}^\eps_{\beta,k}(t) r_k^\eps(t) (i L_{k}^{+}  \partial_{\beta}\gamma^0_{k})  r_k^\eps(t)  \widetilde{\mathcal U}^\eps_{\beta,k}(t)^{*} - \left( \widetilde{\mathcal U}^\eps_{\beta,k}(t) (1-2\gamma^0_{k}) (i L_{k}^{+}  \partial_{\beta}\gamma^0_{k}) r_k^\eps(t)  \widetilde{\mathcal U}^\eps_{\beta,k}(t)^{*} + {\rm h.c.} \right) \\
  & \quad - \eps i  \partial_{\beta}L^+_{k-\eps t e_\beta} \widetilde{\mathcal U}^\eps_{\beta,k}(t) (i L_{k}^{+}  \partial_{\beta}\gamma^0_{k}) \widetilde{\mathcal U}^\eps_{\beta,k}(t))^*
\end{align*}
and
$$
r_k^\eps(t)= i\eps \widetilde{\mathcal U}^\eps_{\beta,k}(t)^* \left(
  L_{k-\eps t e_\beta}^{+} \partial_{\beta}\gamma^0_{k-\eps t
    e_\beta} \right) \widetilde{\mathcal U}^\eps_{\beta,k}(t) - i \eps
L_k^+ \partial_{\beta}\gamma^0_{k} + \widetilde{\mathcal
  U}^\eps_{\beta,k}(t) ^* R_k^\eps(t) \widetilde{\mathcal
  U}^\eps_{\beta,k}(t),
$$
where $R_k^\eps(t)$ is defined in \eqref{eq:defRepsk}. Therefore,
\begin{align*}
  & \Tr( \partial_{\alpha} H_{k-\eps \vec{e}_{\beta} t} 
  \widetilde{\mathcal U}^\eps_{\beta,k}(t)\ (iL_{k}^{+}  \partial_{\beta} \gamma^0_{k}) \, \widetilde{\mathcal U}^\eps_{\beta,k}(t)^*) \\
  & \qquad =
  \Tr \left(\partial_{\alpha} H_{k-\eps \vec{e}_{\beta} t}  \frac{d}{dt} \left(
      \widetilde{\mathcal U}^\eps_{\beta,k}(t)\ (iL_{k}^{+}  \partial_{\beta} \gamma^0_{k}) \, \widetilde{\mathcal U}^\eps_{\beta,k}(t)^* \right) \right) \, 
  + \Tr \left( \partial_{\alpha} H_{k-\eps \vec{e}_{\beta} t} \widetilde R_k^\eps(t) \right) \\
  & \qquad =
  \frac{d}{dt}  \Tr \left(\partial_{\alpha} H_{k-\eps \vec{e}_{\beta} t} \left(
      \widetilde{\mathcal U}^\eps_{\beta,k}(t)\ (iL_{k}^{+}  \partial_{\beta} \gamma^0_{k}) \, \widetilde{\mathcal U}^\eps_{\beta,k}(t)^* \right) \right) \,  
  + \Tr \left( \partial_{\alpha} H_{k-\eps \vec{e}_{\beta} t} \widetilde R_k^\eps(t) \right)
\end{align*}
since
\begin{align*}
  \Tr \left( \frac{d}{dt}(\partial_{\alpha} H_{k-\eps \vec{e}_{\beta} t}) \;  \widetilde{\mathcal U}^\eps_{\beta,k}(t)\ (iL_{k}^{+}  \partial_{\beta} \gamma^0_{k}) \, \widetilde{\mathcal U}^\eps_{\beta,k}(t)^* \right)  &= -\eps   \Tr \left(  \partial_{k_\alpha k_\beta} H_{k-\eps \vec{e}_{\beta} t}
    \widetilde{\mathcal U}^\eps_{\beta,k}(t) \;  (iL_{k}^{+}  \partial_{\beta} \gamma^0_{k}) \, \widetilde{\mathcal U}^\eps_{\beta,k}(t)^* \right) \\
  &= -\eps  e_{\alpha} \cdot e_{\beta}  \Tr \left(\widetilde{\mathcal U}^\eps_{\beta,k}(t) \; (iL_{k}^{+}  \partial_{\beta} \gamma^0_{k}) \, \widetilde{\mathcal U}^\eps_{\beta,k}(t)^* \right) \\ & = 
  -\eps  e_{\alpha} \cdot e_{\beta}  \Tr \left( iL_{k}^{+}  \partial_{\beta} \gamma^0_{k} \right) = 0,
\end{align*}
where we have used the fact that
$\partial_{k_\alpha k_\beta} H_k = - e_{\alpha} \cdot e_{\beta}$ and
the off-diagonal character of
$iL_{k}^{+} \partial_{\beta} \gamma^0_{k}$. Hence, using the bounds
from Proposition~\ref{prop:adiab1}, we obtain
\begin{align*}
  \omega^\eps(t)&= \frac1 t  \int_{{\mathcal B}} \Tr \left(\partial_{\alpha} H_{k-\eps \vec{e}_{\beta} t} \left(
      \widetilde{\mathcal U}^\eps_{\beta,k}(t)\ (iL_{k}^{+}  \partial_{\beta} \gamma^0_{k}) \, \widetilde{\mathcal U}^\eps_{\beta,k}(t)^* \right) \right) \,  dk 
  +  \frac1 t  \int_0^t dt'   \int_{{\mathcal B}}  \Tr \left( \partial_{\alpha} H_{k-\eps \vec{e}_{\beta} t'} \widetilde  R_k^\eps(t') \right) \, dk \\
  &= O\left( \left( \frac 1t + \eps \right) e^{\eta \eps t} \right).
\end{align*}
\end{itemize}
The result follows.

\section{Metals}\label{sec:met}
We prove the two assertions of Theorem \ref{thm:metal} in sequence.

\subsection{Linear response}
We prove the first assertion of Theorem~\ref{thm:metal}: 
We first note that, for $\varepsilon > 0$ small enough and $t \le
\frac 1 \varepsilon \varepsilon^{\theta}$, the function $k\mapsto
  \lambda_{N_{k}+1, k-\eps \vec{e}_{\beta} t} - \lambda_{N_{k}, k-\eps \vec{e}_{\beta} t}$
is bounded away from zero, 
and therefore so is $g_k(\eps t)$. We can
therefore apply Corollary \ref{cor:adiab} on each $B_{N}$ to obtain
\begin{align}
\label{eq:formule0}
j_{\alpha,\beta}^{\varepsilon}(t)
= 
(2\pi)^{-d} \sum_{N \in \N} \Bigg( &{- \int_{B_{N}}}  \Tr( \partial_{\alpha} H_{k-\eps\vec{e}_\beta t} P_{N,k-\eps \vec{e}_{\beta}t})  dk\\
\nonumber
  &
  {-i} \eps \int_{B_{N}}
  \Tr( \partial_{\alpha} H_{k-\eps\vec{e}_\beta t} 
  L^{+}_{{k-\varepsilon \vec{e}_{\beta}t}}\partial_{\beta}
  P_{N,k-\eps \vec{e}_{\beta}t})dk \\
  \nonumber
  &{+i} \eps \int_{B_{N}} \Tr( \partial_{\alpha} H_{k-\eps \vec{e}_{\beta} t} 
  \widetilde{\mathcal U}^\eps_{\beta,k}(t)\ L^{+}_{k} \partial_{\beta} P_{N,k}\, \widetilde{\mathcal U}^\eps_{\beta,k}(t)^*) dk \Bigg)\\
\nonumber
&\hspace{-1.5cm}+{O\left(\eps^2 t{{\rm e}^{\eta \eps t}} \right)}
\end{align}
and so
\begin{align}
j_{\alpha,\beta}^{\varepsilon}(t)=& - (2\pi)^{-d} \sum_{N \in \N}{\int_{B_{N}}}  \Tr( \partial_{\alpha} H_{k-\eps\vec{e}_\beta t} P_{N,k-\eps \vec{e}_{\beta}t}) dk + O(\varepsilon)
\end{align}
when $t \le
\frac 1 \varepsilon \varepsilon^{\theta}$. In contrast to the case of
insulators however, the adiabatic term
\begin{align*}
  -(2\pi)^{-d} \sum_{N \in \N} \int_{\BZ_{N}}  \partial_{\alpha}  \Tr( H_{k{-}\eps\vec{e}_\beta t} P_{N,k{-}\eps \vec{e}_{\beta}t}) dk
  &={} \varepsilon t (2\pi)^{-d} \sum_{N \in \N} \int_{\BZ_{N}} { \partial_\alpha\partial_\beta}
   \left( \Tr(H_{k} P_{N,k})\right) dk + O(\varepsilon^{2}t^2 )
\end{align*}
has a non-zero first-order contribution (the zeroth-order term
vanishes by Proposition~\ref{prop:null_current}). The proportionality factor is computed by
the Stokes formula as
\begin{align}\nonumber 
  \sum_{N \in \N^\ast} \int_{B_N}   { \partial_\alpha\partial_\beta} \Tr(H_{k}\gamma^0_{k})dk 
  = &\sum_{N \in \N^\ast} \int_{B_N}  { \partial_\alpha\partial_\beta}\sum_{n=1}^{N} \lambda_{n,k} = \sum_{N\in\N^\ast} \left(\int_{\mathcal S_{N}}-\int_{\mathcal S_{N-1}}\right)  \partial_{\alpha} \sum_{n=1}^{N} \lambda_{n,k} (ds \cdot \vec e_{\beta})\\
  \label{Stokes}
  = & \sum_{N\in\N^\ast} \int_{\mathcal S_{N}}  \partial_{\alpha} \lambda_{N,k} (ds \cdot \vec e_{\beta})=(2\pi)^d D_{\alpha,\beta}
\end{align}
and the result follows.

\subsection{Bloch oscillations}
Under the assumptions of the second assertion, $N_{k}$ is either
$N_{\rm met}$ or $N_{\rm met}-1$, and in both cases
\begin{align*}
  \lambda_{N_{k}+1, k-\eps \vec{e}_{\beta} t} - \lambda_{N_{k}, k-\eps \vec{e}_{\beta} t}
\end{align*}
is bounded away from zero uniformly in $k$,$t$. We can therefore apply
Corollary \ref{cor:adiab} and obtain
\begin{align*}
  j_{\alpha,\beta}^{\varepsilon}(t)
&= 
{- (2\pi)^{-d} \sum_{N \in \N}\int_{B_{N}}} \left(  \partial_{\alpha} \Tr(H_{k{-}\eps\vec{e}_\beta t} P_{N_k,k{-}\eps \vec{e}_{\beta}t})\right) dk + O((\varepsilon + \varepsilon^{2} t){\rm e}^{\eta\eps t}).
\end{align*}
From the decomposition
\begin{align*}
P_{N_{k},k{-}\eps \vec{e}_{\beta}t} = P_{N_{\rm
    met}-1, k{-}\eps \vec{e}_{\beta}t} + \mathds 1(\lambda_{N_{\rm met}, k}
\leq \mu_{\rm F}) |u_{N_{\rm met},k{-}\eps \vec{e}_{\beta}t}\rangle\langle u_{N_{\rm met},k{-}\eps \vec{e}_{\beta}t}|
\end{align*}
and since $k \mapsto P_{N_{\rm met}-1, k{-}\eps \vec{e}_{\beta}t}$ is smooth and $\cR^\ast$-periodic, we have
\begin{align*}
  j_{\alpha,\beta}^{\varepsilon}(t)
&=
{- (2\pi)^{-d}\int_{\BZ}}  \mathds 1(\lambda_{N_{\rm met}, k}
\leq \mu_{\rm F}) \langle u_{N_{\rm met},k{-}\eps \vec{e}_{\beta}t}| \partial_{\alpha}H_{k{-}\eps\vec{e}_\beta t} |u_{N_{\rm met},k{-}\eps \vec{e}_{\beta}t}\rangle dk + O((\varepsilon + \varepsilon^{2} t){\rm e}^{\eta \eps t})\\
&=
{- (2\pi)^{-d}\int_{\BZ}}  \mathds 1(\lambda_{N_{\rm met}, k}
\leq \mu_{\rm F}) \partial_{\alpha} \lambda_{N_{\rm met}, k{-}\varepsilon e_{\beta} t}dk+ O((\varepsilon + \varepsilon^{2} t){\rm e}^{\eta\eps t}),
\end{align*}
which concludes the proof.
\section{Semi-metals}\label{sec:semimet}

We prove here Theorem \ref{theo:semi-metal}. We decompose the integral
defining~$j_{\alpha,\beta}^\eps(t)$ into several parts depending
whether one integrates far from the Dirac points or not.

\medskip 

We introduce a small parameter $\delta>0$ controlling the size of the
neighborhood of the Dirac points, which is independent of
$t,\varepsilon$. We decompose $\mathcal B$ as the disjoint union
\begin{align*}
  \mathcal B=\mathcal B_{\rm out}^{\delta}\; \cup  \, (\cup_{i \in
  \mathcal I} \mathcal B_{i}^{\delta})
\end{align*}
with
\[ \mathcal B_i^\delta= B(k_{i}, \delta),\]
where $\delta > 0$ is small enough so that 
\[\mathcal B_{\rm out}^{\delta}\subset \{k\in\mathcal B,\;
  \lambda_{N_{\rm sm},k} \leq \mu_F-c\delta\}\]
  for some constant $c>0$. 
Note that this decomposition is time-reversal
symmetric in the sense that 
\[
  -\mathcal B_{\rm out}^{\delta} = \mathcal
B_{\rm out}^{\delta}\;\; \mbox{and}\;\;
  - (\cup_{i \in \mathcal I} \mathcal B_{i}^{\delta}) = (\cup_{i \in \mathcal I} \mathcal B_{i}^{\delta}).
\]
We work in the regime $\varepsilon t \ll \delta \ll 1$,
$\varepsilon \ll \delta \ll 1$.

\medskip 

In the following analysis, we first treat the regions
$\mathcal B_{\rm out}^{\delta}$, where we will use adiabatic theory
with a non-zero gap larger than a constant times $\delta$. In the sets
$\mathcal B_{i}^{\delta}$, where the gap closes, we study the structure of the Taylor expansion of the Hamiltonian~$H_k$ close to the Dirac points  and construct  two-band reduced Hamiltonians~$H^R_{i,k}$. Then, we use the linear
response Proposition \ref{prop:Dyson}, reducing successively from the
Hamiltonian $H$ to the reduced Hamiltonian
$H^{R}_{i,k}$,
 and finally to the
Dirac Hamiltonian 
\[H^{D}_k = \begin{pmatrix} 0 & k_1-ik_2 \\ k_1+ik_2 & 0\end{pmatrix}\] 
for which we can explicitly
compute the current. Adding the contributions, we will obtain
\begin{align*}
\sigma_{\alpha,\beta}:=  \lim_{t \to \infty} \lim_{\varepsilon \to 0}\frac 1 {\varepsilon t} \int_{0}^{t} j_{\alpha,\beta}^{\varepsilon}(t') dt' = \frac {|\mathcal I|} {16}e_{\alpha} \cdot e_{\beta}  + O(\delta)
\end{align*}
Finally, we will pass to the limit $\delta \to 0$.

\subsection{Far from the Dirac points}
We set 
\[  j^{\eps,{\rm out}}_{\alpha,\beta}(t'):= -\frac 1{4\pi^2} \int_{\mathcal B^\delta_{\rm out}}{\rm Tr} \left(\partial_\alpha H_k\widetilde{\mathcal U}^\eps_{\beta,k}(t')\, \gamma_k^{0}\, \widetilde{\mathcal U}^\eps_{\beta,k}(t')^*\right) \, dk.
\]
Let $k \in \mathcal B_{\rm out}^{\delta}$. In the regime we consider,
$\gamma^{0}_{k - \varepsilon e_{\beta} t} = P_{N_{\rm sm},k-\varepsilon
  e_{\beta} t}$ is gapped with a gap larger than a constant times~$ \delta$. Applying
the analysis of the previous sections, we obtain that
\begin{align*}
 -{\rm Tr} \left(\partial_\alpha H_k\widetilde{\mathcal U}^\eps_{\beta,k}(t)\, \gamma_k^{0}\, \widetilde{\mathcal U}^\eps_{\beta,k}(t)^*\right) \\
 =&
-  \partial_\alpha \left( \Tr(H_{k{-}\eps\vec{e}_\beta t} \gamma^{0}_{k{-}\eps \vec{e}_{\beta}t})\right) 
-i \eps  
  \Tr(\partial_\alpha H_{k-\eps\vec{e}_\beta t} 
  L^{+}_{k-\eps \vec{e}_{\beta}t}\partial_{\beta}
  \gamma^{0}_{k-\eps\vec{e}_\beta t})  \\
  &+i \eps\Tr(\partial_\alpha H_{k{-}\eps \vec{e}_{\beta} t} 
  \widetilde{\mathcal U}^\eps_{\beta,k}(t)\left( L_{{k}}^{+} \partial_{\beta} \gamma^{0}_{k}\right) \widetilde{\mathcal U}^\eps_{\beta,k}(t)^*)  + O\left(\eps^2 t \delta^{-6}\right).
\end{align*}

We treat each term separately.
\begin{itemize}
\item For the first (adiabatic) term, we have
  \begin{align*}
    -\partial_\alpha\Tr(H_{k-\eps\vec{e}_\beta t} \gamma^{0}_{k-\eps \vec{e}_{\beta}t})  = -\partial_\alpha \Tr( H_{k} \gamma^{0}_{k}) &+ \varepsilon t \partial_{\alpha} \partial_{\beta} \Tr(H_k \gamma^{0}_k)\\
    &+ O\left({\varepsilon^{2} t^{2}}{\delta^{-4}}\right).
  \end{align*}
  By time-reversal symmetry, the first term vanishes when integrated
  on $\mathcal B_{\rm out}^{\delta}$. Using Stokes formula for the
  second term as in the metallic case, we get
  \begin{equation}\label{secondslot}
    \int_{\mathcal B_{\rm out}^{\delta}} -\partial_\alpha \left( \Tr(H_{k-\eps\vec{e}_\beta t} \gamma^{0}_{k-\eps \vec{e}_{\beta}t})\right) dk = \varepsilon t \sum_{ n\leq N_{\rm sm}} \int_{\partial{\mathcal B}_{\rm in}^{\delta}}\partial_\alpha \lambda_{n,k} (ds\cdot \vec e_\beta) + O\left({\varepsilon^{2} t^{2}}{\delta^{-4}}\right)
  \end{equation}

\item For the second (static) term we similarly expand in
  $\varepsilon$
  \begin{align*}
    -i\eps \int_{\mathcal B_{\rm out}^{\delta}}
  \Tr(
  \partial_\alpha H_{k-\eps\vec{e}_\beta t} 
  L_{k-\eps \vec{e}_{\beta}t}^{+}\partial_{\beta}
  \gamma^{0}_{k-\eps\vec{e}_\beta t}
  ) dk
 &=  -i \eps \int_{\mathcal B_{\rm out}^{\delta}} {\rm Tr}\left(\gamma^0_{k}[ \partial_{\alpha} \gamma^0_{k} ,  \partial_{\beta} \gamma^0 _{k}]\right) dk
 +O(\eps^2 t \delta^{-6})\\
 &=O(\eps^2 t \delta^{-6}),
  \end{align*}
where we used the fact that the function $k \mapsto {\rm
  Tr}\left(\gamma^0_{k}[ \partial_{\alpha} \gamma^0_{k} ,
  \partial_{\beta} \gamma^0 _{k}]\right)$ is odd.
\item For the third (oscillatory) term, we use the same arguments as
  in the insulating case, and obtain
  \begin{align*}
  &\frac{1}{t}\int_{0}^{t} dt' \int_{\mathcal B_{\rm out}^{\delta}} i\varepsilon\Tr( \partial_{\alpha} H_{k{-}\eps \vec{e}_{\beta} t'} 
  \widetilde{\mathcal U}^\eps_{\beta,k}(t')\left( L_{k}^{+}  \partial_{\beta} \gamma^0_{k} \right)\, \widetilde{\mathcal U}^\eps_{\beta,k}(t')^*) dk = O\left(\varepsilon  \delta^{-6} \left( \frac{1}{t}+\eps\right) \right).
  \end{align*}
\end{itemize}

We are left with
\begin{align}
\nonumber
\frac 1 {t} \int_{0}^{t}  j^{\eps,{\rm out}}_{\alpha,\beta}(t') \, dt'
= & \frac {\eps}{4\pi^2 t} \int_{0}^{t}t' dt'   \sum_{ n\leq N_{\rm sm}} \int_{\partial{\mathcal B}_{\rm in}^{\delta}}\partial_\alpha \lambda_{n,k} (ds\cdot \vec e_\beta) 
+ O\left( \frac{\varepsilon}{\delta^{6}} \left(\frac 1 t + \varepsilon(1+t^2) \right)\right)\\
    \label{Bout}
  = & {\frac {  \eps t}{2\pi^2}   \sum_{ n\leq N_{\rm sm}} \int_{\partial{\mathcal B}_{\rm in}^{\delta}}\partial_\alpha \lambda_{n,k} (ds\cdot \vec e_\beta) + O\left( \frac{\varepsilon}{\delta^{6}} \left(\frac 1 t + \varepsilon(1+t^2) \right)\right).}
\end{align}

\subsection{Close to Dirac points: reduction to the 2-band case}
\label{sec:dirac_reduc}

We set 
\begin{equation} \label{eq:jepsin}
j^{\eps,{\rm in}}_{\alpha,\beta}(t):= -\frac1{4\pi^2} \sum_{i \in {\cal I}} \int_{\mathcal B^\delta_{\rm in}}{\rm Tr} \left(\partial_\alpha H_k\widetilde{\mathcal U}^\eps_{\beta,k}(t)\, \gamma_k^{0}\, \widetilde{\mathcal U}^\eps_{\beta,k}(t)^*\right) \, dk.
\end{equation}
Using the linear response Proposition \ref{prop:Dyson}, we have for almost all  $k\in \mathcal B^\delta_{\rm in}$ and all $\eps,t \ge 0$,  
\begin{align*}
  -{\rm Tr} \left(\partial_\alpha H_{k - \varepsilon e_{\beta} t}\widetilde {\mathcal U}^\eps_{\beta,k}(t)\, \gamma_k^{0}\, \widetilde{\mathcal U}^\eps_{\beta,k}(t)^*\right) = 
  &-\partial_\alpha {\rm Tr} \left(H_k \gamma_k^{0}\right) 
  + \varepsilon t  { \partial_{\alpha}\partial_{\beta}} \left( \Tr (H_{k} \gamma_{k}^{0}) \right)\\
  & - i \varepsilon \Tr(\partial_{\alpha} H_{k} (e^{-i t L_k} - 1) L_{k}^{+} \partial_{\beta} \gamma_{k}^{0}) + O(\varepsilon^{2} {t}^{3}(1+t^3)).
\end{align*}
Using~\eqref{eq:first_bound},  the equality
${\rm Tr} \left(H_k \gamma_k^{0}\right) = \sum_{n=1}^{N_{\rm sm}}
\lambda_{n,k}$ and Assumption~\ref{ass:semimetal}, it is easily seen that the
left-hand side, as well as the first and fourth terms of the right-hand side of that equation,
are bounded uniformly in $k$ and therefore integrable on
$\mathcal B_{i}^{\delta}$. Besides, for $k\in\mathcal B^\delta_{\rm in}$, the second term is bounded by a constant
multiple of $(1/|k-k_{i}|)$ as $k \to k_{i}$, and is therefore
integrable. It follows that the third term is also integrable on
$\mathcal B_{i}^{\delta}$.

\medskip 

We treat the three leading terms of the right-hand side in sequence.
\begin{itemize}
\item The first term vanishes when integrated on the time-reversal symmetric
set $\cup_{i \in \mathcal I} \mathcal B_{i}^{\delta}$.
\item For the second, arguing as in the
metallic case, we get
\begin{align*}
  \int_{\mathcal B_{i}^{\delta}} \varepsilon t \partial_{\alpha\beta} \left( \Tr (H_{k} \gamma_{k}^{0}) \right) dk = -\varepsilon t \sum_{ n\leq N_{\rm sm}} \int_{\partial{\mathcal B}_{i}^{\delta}}\partial_\alpha \lambda_{n,k} (ds\cdot \vec e_\beta),
\end{align*}
so that the corresponding term in \eqref{eq:jepsin} cancels the contribution~\eqref{Bout} from
$\mathcal B_{\rm out}^{\delta}$.
\item For the third term, we use
  \begin{align*}
  &  \Tr(\partial_{\alpha} H_{k} (e^{-i t L_{k}} - 1) L_{k}^{+} \partial_{\beta} \gamma_{k}^{0}) = \Tr(\partial_{\alpha} H_{k} (e^{-i t L_{k}} - 1) (L_{k}^{+})^{2} [\gamma_{k}^{0}, \partial_{\beta} H_{k}])\\
    &= \sum_{n \le N_{\rm sm}}\sum_{m > N_{\rm sm}} (e^{-i t (\lambda_{n,k} - \lambda_{m,k})} - 1)\frac{\langle u_{n,k}, \partial_{\beta} H_{k} u_{m,k} \rangle\langle u_{m,k}, \partial_{\alpha} H_{k} u_{n,k}\rangle}{(\lambda_{m,k} - \lambda_{n,k})^{2}} - \rm c.c.
  \end{align*}
  with the sum converging from the asymptotics \eqref{eq:asymptotics_eigenvalues}.

  When $n \neq N_{\rm sm}$ or $m \neq N_{\rm sm}+1$, the denominators
  in that equation are bounded from below independently of $\delta$.
  The constant term vanishes when integrated over the time-reversal
  symmetric set $\cup_{i \in \mathcal I} \mathcal B_{i}^{\delta}$, and
  the oscillatory term can be treated using the formula
  \begin{align*}
    \frac 1 t \int_{0}^{t} e^{-i\omega t'} dt' = \frac{e^{-i\omega t} - 1}{-i\omega t}
  \end{align*}
  with $\omega = \lambda_{n,k} - \lambda_{m,k}$ bounded away from zero
  independently of $\delta$.
\end{itemize}

Putting all the results of the previous two sections together, we get that
\begin{align*}
  \frac 1 t \int_{0}^{t} j_{\alpha,\beta}^{\varepsilon}(t')dt' &=  \frac 1 t \int_{0}^{t} j_{\alpha,\beta}^{\varepsilon,{\rm out}}(t')dt' +  \frac 1 t \int_{0}^{t} j_{\alpha,\beta}^{\varepsilon,{\rm in} }(t')dt' \\
  &=-\frac {i\varepsilon }{4\pi^2 t}\sum_{i \in \mathcal I} \int_{0}^{t} dt'\int_{\mathcal B_{i}^{\delta}} (e^{-i t' (\lambda_{N_{\rm sm},k} - \lambda_{N_{\rm sm+1},k})} - 1)\\
  &\qquad\qquad \qquad\qquad \times \frac{\langle u_{N_{\rm sm},k}, \partial_{\beta} H_{k} u_{N_{\rm sm+1},k} \rangle\langle u_{N_{\rm sm+1},k},\partial_{\alpha} H_{k}u_{N_{\rm sm},k}\rangle}{(\lambda_{N_{\rm sm}+1,k} - \lambda_{N_{\rm sm},k})^{2}}dk - \rm c.c.\\
  & \qquad\qquad  + O\left( {\varepsilon^{2}\delta^{-6}(1+t^{2})} + {\varepsilon}{\delta^{-6} t^{-1}}+{\eps^2t^3(1+t^3)}\right).
\end{align*}
At this stage of the proof, only two modes are involved in the formula giving the current, namely the two modes that cross at the Fermi level. Everything happens as for a two-band model that we now study. We write for short
 \[ 
 \frac 1 t \int_{0}^{t} \frac{j_{\alpha,\beta}^{\varepsilon}(t')}{\eps} \, dt' =  -\frac {i }{4\pi^2 t}\sum_{i \in \mathcal I} \int_{0}^{t}  I^{R,i}_{\alpha,\beta}(t') \, dt'+O\left( {\varepsilon \delta^{-6}(1+t^{2})} + {\delta^{-6} t^{-1}}+{\eps t^3(1+t^3)}\right),
 \]
with 
\[ I^{R,i}_{\alpha,\beta}(\delta,t)=\int_{\mathcal B_{i}^{\delta}} (e^{-i t (\lambda_{N_{\rm sm},k} - \lambda_{N_{\rm sm+1},k})} - 1)\frac{\langle u_{N_{\rm sm},k}, \partial_{\beta} H_{k} u_{N_{\rm sm+1},k} \rangle\langle u_{N_{\rm sm+1},k},\partial_{\alpha} H_{k}u_{N_{\rm sm},k}\rangle}{(\lambda_{N_{\rm sm}+1,k} - \lambda_{N_{\rm sm},k})^{2}} \, dk - {\rm c.c.}.
\]

\subsection{Close to the Dirac points: the local model}
We now are interested in the computation of $I^{R,i}_{\alpha,\beta}(t)$. 
In the following, we drop the index $i$ and assume without loss of
generality that $k_{i} = 0$.

Hypothesis \eqref{hyp:semimetals_1} implies that, for $k$ small
enough, the Bloch Hamiltonian $H_k$ has exactly two eigenvalues
(counting multiplicities) close to $\mu_{\rm F}$. Consider an
arbitrary orthonormal basis $(v_0,w_0)$ of
$\Ran(P_{N_{\rm sm}+1,0} - P_{N_{\rm sm}-1,0})$. For all $k$ small
enough, we can construct an orthonormal basis $(v_k,w_k)$ of
$\Ran(P_{N_{\rm sm}+1,k} - P_{N_{\rm sm}-1,k})$ by L\"owdin
orthonormalization of
$((P_{N_{\rm sm}+1,k} - P_{N_{\rm sm}-1,k})v_0,(P_{N_{\rm sm}+1,k} -
P_{N_{\rm sm}-1,k})w_0)$, and set
$$
H_k^R = [v_{k} | w_{k}]^{*}
H_k[v_{k} | w_{k}] = \left( \begin{array}{cc} \langle v_k, H_k v_k \rangle &  \langle v_k, H_k w_k \rangle \\
\langle w_k, H_k v_k \rangle &  \langle w_k, H_k w_k \rangle \end{array} \right).
$$
It follows that the reduced Hamiltonian
\begin{align*}
  H_{k}\Big|_{\Ran(P_{N_{\rm sm}+1,k} - P_{N_{\rm sm-1},k}) }
\end{align*}
is equivalent through a unitary transform that depends analytically on $k$
to the reduced $2 \times 2$ Hamiltonian
\begin{align}
  \label{eq:Hmatk}
  H^{R}_{k} = \sum_{p=0}^3 b^{p}(k) \sigma_p,
\end{align}
where $$\sigma_{0} = {\rm Id}_{\C^{2}}, \;\; \sigma_1=\left( \begin{array}{cc} 0 & 1 \\ 1 & 0 \end{array} \right), \;\;\sigma_2=\left( \begin{array}{cc} 0 & -i \\ i & 0 \end{array} \right), \;\;\sigma_3=\left( \begin{array}{cc} 1 & 0 \\ 0 & -1 \end{array} \right)$$ 
are the Pauli matrices, and $(b^{p})_{p=0,1,2,3}$ are real-valued
analytic functions of $k$ in a neighborhood of $0$.

The matrix $H^{R}_{k}$ has eigenvalues
$$  \lambda_{\pm}(k) = b^{0}(k) \pm \sqrt{\sum_{p=1}^{3}b^{p}(k)^{2}}.
$$
It follows that
\begin{align*}
  b^{0}(k) &= \mu_{F} + O(|k|^{2}), \quad 
  b^{p}(k) = v_{F}\langle q^{p}, k\rangle + O(|k|^{2}), \quad p=1,2,3,
\end{align*}
where the $(q^{p})_{p=1,2,3}$ are the rows of a $3 \times 2$
matrix $Q$ with orthogonal columns, so that
\begin{align}\label{def:HR}
  H^{R}_{k} = \mu_{F} + v_{F} (Q k) \cdot \sigma + O(|k|^{2}).
\end{align}
Let $R \in {\rm SO}(3)$ be a rotation matrix that maps
${\rm Ran}(Q)$ to ${\rm Span}(e^0_{1},e^0_{2})$, where $(e^0_{1},e^0_{2},e^0_{3})$ is the canonical basis of $\R^3$. Let $U$ be one of
its associated $2 \times 2$ unitary matrices through the two-to-one ${\rm
  SU(2)} \to {\rm SO}(3)$ mapping, so that \cite{cornwell1984group}
\begin{align*}
  R_{pq} = \frac 1 2 \Tr(\sigma_{p} U \sigma_{q} U^{*}).
\end{align*}
It follows that
\begin{align*}
  \Tr(\sigma_{3} U H^{R}_{k} U^{*}) &= v_{F}\sum_{q=1}^{3} (Q k)_{q} \Tr(\sigma_{3} U \sigma_{q} U^{*}) + O(|k|^{2})\\
  &= 2 v_{F}\sum_{q=1}^{3} \langle  e^0_{3}, R Q k\rangle+ O(|k|^{2})\\
  &= O(|k|^{2}).
\end{align*}
Up to a unitary transform, we can therefore assume $Q$ to be
a $2\times 2$ matrix in~\eqref{def:HR}.

\subsection{The two-band case: reduction to the Dirac Hamiltonian}
\paragraph{Reduction to $H_{k}^{R}$}
For $k \neq 0$, let $\lambda^{R}_{\pm,k}$ be the larger and smaller
eigenvalues of $H^{R}_{k}$ respectively, and $u^{R}_{\pm,k}$
associated orthonormal eigenvectors in $\C^2$. We have
$\lambda_{-,k}^{R} = \lambda_{N_{\rm sm},k}, \lambda_{+,k}^{R} =
\lambda_{N_{\rm sm}+1,k}$, and
\begin{align*}
  [v_{k} | w_{k}]
u_{+,k}^{R} = e^{i\theta_{+}(k)}u_{N_{\rm sm} + 1,k},\qquad[v_{k} | w_{k}]
u_{-,k}^{R} = e^{i\theta_{-}(k)}u_{N_{\rm sm},k}
\end{align*}
for some phases $\theta_{\pm}(k) \in \R$. We have
\begin{align*}
  \partial_{\alpha} H^{R}_{k}
  &= [v_{k} | w_{k}]^{*}
\partial_{\alpha} H_k[v_{k} | w_{k}]+ \partial_{\alpha} [v_{k} | w_{k}]^{*} [v_{k} | w_{k}] H^{R}_{k} + H^{R}_{k}  [v_{k} | w_{k}]^{*}\partial_{\alpha} [v_{k} | w_{k}]
\\
&= [v_{k} | w_{k}]^{*}
\partial_{\alpha} H_k[v_{k} | w_{k}] + O(|k|)
\end{align*}
where we have used for the first line that $H_{k}$ commutes with
$P_{N_{\rm sm}+1,k} - P_{N_{\rm sm},k} = [v_{k} | w_{k}]
[v_{k} | w_{k}]^{*}$, and for the second that $H_{k}^{R} = \mu_{F} {\rm Id}_{2} + O(|k|)$ and
$\partial_{\alpha} \Big([v_{k} | w_{k}]^{*} [v_{k} | w_{k}] \Big) =
\partial_{\alpha} {\rm Id}_{2} = 0$. We therefore obtain
\begin{align*}
  \langle u^{R}_{+,k},\partial_{\alpha} H^{R}_{k}u^{R}_{-,k}\rangle&= e^{-i(\theta_{+}(k)-\theta_{-}(k))}\langle u_{N_{\rm sm}+1,k},\partial_{\alpha} H_{k}u_{N_{\rm sm},k}\rangle + O(|k|)\\
  \langle u^{R}_{-,k}, \partial_{\beta} H^{R}_{k} u^{R}_{+,k} \rangle&= e^{+i(\theta_{+}(k)-\theta_{-}(k))}\langle u_{N_{\rm sm},k},\partial_{\beta} H_{k}u_{N_{\rm sm}+1,k}\rangle + O(|k|).
\end{align*}
Since $(\lambda^{R}_{-,k} - \lambda^{R}_{+,k})$ is bounded from below
by a constant multiple of $|k|$, it follows that
\begin{equation}
  \label{eq:IRdt}
  I^{R}_{\alpha,\beta}(\delta,t) = \int_{B(0,\delta)} (e^{-i t (\lambda^{R}_{-,k} - \lambda^{R}_{+,k})} - 1)\frac{\langle u^{R}_{-,k}, \partial_{\beta} H^{R}_{k} u^{R}_{+,k} \rangle\langle u^{R}_{+,k},\partial_{\alpha} H^{R}_{k}u^{R}_{-,k}\rangle}{(\lambda^{R}_{+,k} - \lambda^{R}_{-,k})^{2}}dk - {\rm c.c.} + O(\delta).
\end{equation}

\paragraph{Reduction to $H_{k}^{D}$}
By standard results of perturbation theory \cite{Kato} applied to
$H^{R}_{k} = \mu_{F} + v_{F} (Q k) \cdot \sigma + O(|k|^{2})
$ with gap greater than a constant multiple of $|k|$, 
\[
  \lambda^{R}_{\pm,k} = \lambda^{Q}_{\pm,k} +
  O(|k|^{2})\;\;\mbox{and}\;\;   u^{R}_{\pm,k} = u^{Q}_{\pm,k} + O(|k|)
\]
where the superscript $Q$ refers to eigenvalues and appropriately
chosen orthonormal eigenvectors of the Hamiltonian
\begin{align*}
  H^{Q}_{k} = v_{F}(Qk) \cdot \sigma.
\end{align*}
It follows that
\begin{align*}
  I^{R}_{\alpha,\beta} (\delta,t) &= I^{Q}_{\alpha,\beta}(\delta,t) + O(\delta),
\end{align*}
where $I^{R}_{\alpha,\beta}(\delta,t)$ is defined similarly to
\eqref{eq:IRdt} as
\begin{align*}
  I^{Q}_{\alpha,\beta}(\delta,t) = \int_{B(0,\delta)}  (e^{-i t (\lambda^{Q}_{-,k} - \lambda^{Q}_{+,k})} - 1)\frac{\langle u^{Q}_{-,k}, \partial_{\beta} H^{Q}_{k} u^{Q}_{+,k} \rangle\langle u^{Q}_{+,k},\partial_{\alpha} H^{Q}_{k}u^{Q}_{-,k}\rangle}{(\lambda^{Q}_{+,k} - \lambda^{Q}_{-,k})^{2}}dk - {\rm c.c.}.
\end{align*}

We perform the change of variable $k' = Qk$ (recall that $Q$ is orthogonal) and obtain
\begin{align*}
I^{Q}_{\alpha,\beta}(\delta,t) = \, e_\beta^T \,   I^D(\delta,t) e_\alpha,
\end{align*}
where the coefficients $I^D_{ij}(\delta,t) $ of the $2\times 2$ matrix  $I^D(\delta,t) $ are given by 
\begin{align*}
  I^{D}_{ij}(\delta,t) = \int_{B(0,\delta)} (e^{-i t (\lambda^{D}_{-,k} - \lambda^{D}_{+,k})} - 1)\frac{\langle u^{D}_{-,k}, \partial_{k_j} H^{D}_{k} u^{D}_{+,k} \rangle\langle u^{D}_{+,k},\partial_{k_i} H^{D}_{k}u^{D}_{-,k}\rangle}{(\lambda^{D}_{+,k} - \lambda^{D}_{-,k})^{2}}dk - {\rm c.c.}
\end{align*}
and the superscript $D$ refers to the Dirac Hamiltonian
\begin{align*}
  H^{D}_{k} = v_{F} k \cdot \sigma.
\end{align*}

\paragraph{The Dirac Hamiltonian $H^{D}_{k}$}
We finish by computing $I^{D}(\delta,t)$ explicitly. Let
$k=r(\cos \theta, \sin \theta)$. We have
\[  \lambda^{D}_{\pm,k} = \pm v_{F} r,\;\;
  u^{D}_{+,k} = \frac 1 {\sqrt 2}
  \begin{pmatrix}
    1\\e^{i\theta}
  \end{pmatrix}\;\;\mbox{and}\;\;
  u^{D}_{-,k}= \frac 1 {\sqrt 2}
  \begin{pmatrix}
    -e^{-i\theta}\\1
  \end{pmatrix}.
  \]
By an explicit calculation, we obtain 
\begin{align*}
  \int_{0}^{2\pi}(e^{-i t (\lambda^{D}_{-,k} - \lambda^{D}_{+,k})}-1)\frac{\langle u^{D}_{-,k}, \partial_{k_i} H^{D}_{k} u^{D}_{+,k} \rangle\langle u^{D}_{+,k},\partial_{k_j} H^{D}_{k} u^{D}_{-,k}\rangle}{(\lambda^{D}_{+,k} - \lambda^{D}_{-,k})^{2}} \, d\theta - \text{c.c.} &= i\pi \frac 1 {2r^{2}}\sin(2v_Frt) \delta_{ij}.
\end{align*}
It follows that
\begin{align*}
  \frac 1 t \int_{0}^{t}I^{D}_{ij}(\delta,t') \, dt' &= \frac {i\pi}{2}\delta_{ij} {1\over t}\int_0^t \int_{0}^{\delta}\frac{\sin(2v_Frt')}{r} \, dr \, dt' = \frac{i\pi}{4 t}\delta_{ij}\int_{0}^{\delta}\frac{1 -  \cos(2v_Frt)}{v_Fr^{2}} \, dr\\
  &=\frac{i\pi}{4} \delta_{ij} \int_{0}^{\delta v_Ft}\frac{1-\cos(2r')}{(r')^2} \, dr' = \frac{i \pi^{2}}{4}\delta_{ij} + O\left((\delta t)^{-1}\right).
\end{align*}
We finally get by summing all the estimates
\begin{align*}
  \frac 1{t} \int_{0}^{t}  \frac{j_{\alpha,\beta}^{\varepsilon}(t')}{\eps} \, dt' = \frac {|\mathcal I|} {16} e_{\alpha} \cdot e_{\beta}   +  O\left( \delta +\varepsilon t^3 (1+t^3) +\frac{1}{\delta^{6}} \left(\frac 1 t + \varepsilon(1+t)^{2} \right)\right),
\end{align*}
hence the result.

\appendix

\section{Proofs of two technical lemmata} 

\subsection{Proof of Lemma~\ref{MagneticFarisLavine}}\label{sec:FL_lemma}
\begin{proof}
We replicate the proof of the Faris-Lavine Theorem given in \cite{ReedSimon2}, replacing the Laplacian by $\frac 12(-i\nabla + \mathcal{A})^2$. It consists in verifying the following two hypotheses of \cite[Theorem X.37]{ReedSimon2}.  Let $A = \frac 12(-i\nabla+\mathcal{A})^2 + W+V$ and $N=A+2c|x|^2 +b$, where $b\in\R$ will be specified below: 
\begin{equation}
\label{HypX37.1}\textit{there exists $h$, such that for any } \phi \in \mathcal{C},
\quad \|A\phi\|\leq h \|N\phi\|;
\end{equation}
\begin{equation}
\label{HypX37.2}
\textit{for some $\ell$, for any }  \phi \in \mathcal{C},\quad |(A\phi,N\phi) - (N\phi,A\phi)| \leq \ell \|N^\frac{1}{2} \phi\|^2.
\end{equation}

By hypothesis \textit{3} in Lemma~\ref{MagneticFarisLavine} and the conditions on $W$, it is possible to choose $b$ so that $N\geq 1$.
As quadratic forms on $\mathcal{C}$, 
\[N^2 = (A+b)^2 + 4c \sum_{j=1}^d x_j (A+b+c|x|^2)x_j - 2cd.  \]
Hypotheses \textit{1} and \textit{3} guarantee that $A+b+c|x|^2$ is bounded below.
Hence, increasing the value of~$b$ if necessary to make this operator positive, we have
\[ \|(A+b)\phi\|^2_{L^2} \leq \|N\phi \|^2_{L^2} +4cd\|\phi \|^2_{L^2},\]
which proves (\ref{HypX37.1}). 

\medskip 

For (\ref{HypX37.2}), we observe that 
\[ \pm i [A,N] = \pm 2c(x \cdot (-i\nabla+\mathcal{A}) + (-i\nabla +\mathcal{A})\cdot x) \leq 2 c \left((-i\nabla +\mathcal{A})^2 +|x|^2\right) \leq \ell N, \]
where we have used
\[(-i\nabla + \mathcal{A})^2 + |x|^2 \pm (x\cdot(-i\nabla + \mathcal{A}) + (-i\nabla + \mathcal{A})\cdot x) = (-i\nabla + \mathcal{A} \pm x)^2 \geq 0 \]
and 
\[ N = \left(\frac a 2 (-i\nabla + \mathcal{A})^2 + V\right) + (W+c|x|^2) + \frac{1-a}{2}(-i\nabla + \mathcal{A})^2 + c|x|^2 +b \geq e((-i\nabla + \mathcal{A})^2+|x|^2), \]
where $e = \min(c,\frac{1-a}{2})>0$ and where $b$ is chosen so that 
$$b- f +\min \,\sigma\left(\frac a 2 (-i\nabla + \mathcal{A})^2 + V\right) \geq 0.$$
This proves  (\ref{HypX37.2}). 
Hence $A$ is essentially self-adjoint on $\mathcal{C}$.
\end{proof}

\subsection{Proof of Lemma~\ref{lemmaUnitaryPropagator}}\label{sec:unitaryprop}

\begin{proof}By the Kato-Rellich theorem, for any $0\leq t\leq T$,
  $H(t)$ is self-adjoint on $L^2_{\rm per}$ with domain~$H^2_{\rm per}$, and
  bounded below. We will show that there exists $\mu>0$ so that the
  graph norm of $(H(t)+\mu)$ for any $0\leq t \leq T$ is equivalent to
  the $H^2_{\rm per}$-norm. This will prove Lemma~\ref{lemmaUnitaryPropagator} by Proposition 2.1
  in \cite{teufel2003adiabatic} (see also Theorem X.70 in
  \cite{ReedSimon2}).

\medskip 

We have for any $\mu>0$, $0\leq t\leq T$ and $\phi \in H^{2}_{\rm per}$,
\[ \|(H(t)+\mu)\phi \|_{L^{2}_{\rm per}} \leq (1+a)\|H_0 \phi \|_{L^{2}_{\rm per}} + (b+\mu) \|\phi\|_{L^{2}_{\rm per}} \leq (1+a+b+\mu) \|\phi\|_{H^2_{\rm per}},\]
and so the graph norm is controlled by the $H^2_{\rm per}$-norm. 

\medskip 

For the other inequality, we relate the resolvent of $H(t)$ to that of $H_0$ by a bounded operator, with bounded inverse. 
Notice that, for any $\mu>0$, since $H_0$ is positive, 
\[\forall \; 0\leq t\leq T, \quad (H(t)+\mu) = (1+H_1(t)(H_0+\mu)^{-1})(H_0+\mu).\]
Furthermore,
\[\forall \; 0\leq t\leq T, \quad\|H_1(t)(H_0+\mu)^{-1}\| \leq
  a\|H_0(H_0+\mu)^{-1}\| +b\|(H_0+\mu)^{-1}\| \leq a
  +\frac{b}{\mu}.\]
and so, for $\mu > \frac{b}{1-a}$, the operator
$1+ H_1(t)(H_0+\mu)^{-1}$ is bounded and invertible with bounded
inverse in $L^{2}_{\rm per}$. Therefore $(H(t)+\mu)^{-1}$ is bounded 
from $L^2_{\rm per}$ to $H^2_{\rm per}$, which means there exists $C>0$ such
that, for any $\phi \in H^2_{\rm per}$ and $0\leq t\leq T$,
\[  \|\phi\|_{H^2_{\rm per}} = \|(H(t)+\mu)^{-1}(H(t)+\mu)\phi\|_{H^2_{\rm per}} \leq C\|(H(t)+\mu)\phi\|_{L^{2}_{\rm per}}, \]
which concludes the proof.
\end{proof}

\vspace{1cm}

\noindent {\bf Acknowledgements}. We are grateful to Caroline Lasser
for stimulating discussions all along the preparation of this article.
We thank the two reviewers for their constructive suggestions, and in
particular for pointing out to us a more elegant proof of
Proposition~\ref{prop:null_current} and the extension to uniform
magnetic fields discussed in Remark \ref{rem:only_depends_on_Hk}. This
project has been supported by Labex Bezout and has received funding
from the European Research Council (ERC) under the European Union's
Horizon 2020 research and innovation programme (grant agreement No
810367). The second author thanks the mathematics department of the
Technische Universit\"at M\"unchen for hosting her during the final
writing of this article.

\bibliographystyle{plain}
\bibliography{refs}

\vskip 1cm

\noindent *Universit\'e Paris-Est, Laboratoire d'Analyse et de Math\'ematiques Appliqu\'ees (UMR 8050), UPEM, UPEC, CNRS, F-94010, Cr\'eteil, France (clotilde.fermanian@u-pec.fr)\\
$\dagger$ Universit\'e Paris-Est, CERMICS, \'Ecole des Ponts ParisTech and Inria Paris,
F-77455, Marne-la-Vall\'ee, France
(eric.cances@enpc.fr, antoine.levitt@inria.fr, sami.siraj-dine@enpc.fr)\\

\end{document}